%% file: icfp2015-full.tex
\documentclass[11pt]{lmcs}
	\makeatletter
	\def\endfront@text{}
	\makeatother
\title{Denotational cost semantics for functional languages with 
inductive types}

\usepackage{adjustbox}
\usepackage{amsmath}
\usepackage{amsfonts}
\usepackage{bussproofs}
	
\usepackage[labelfont=bf]{caption}
\usepackage{doi}
\usepackage{flushend}
\usepackage[T1]{fontenc}
\usepackage{graphicx}
\usepackage{listings}
	\lstMakeShortInline|
\usepackage{microtype}
\usepackage[]{natbib}
\usepackage{paralist}
\usepackage{proof}
\usepackage{stmaryrd}
\usepackage{suffix}
\usepackage[normalem]{ulem}
	\newcommand{\comment}[3][\relax]{%
	\ifx#1\relax\else\uwave{#1}\fi\marginpar{\tiny {#2}: {#3}}}
	\newcommand{\ndComment}[2][\relax]{%
	\ifx#1\relax\comment{ND}{#2}\else\comment[#1]{ND}{#2}\fi}
	\newenvironment{longComment}[1]{\begin{quotation}\footnotesize {#1}:\enspace}{\end{quotation}}
	
\usepackage{url}

\usepackage{hyperref}

\newcommand{\proofcase}[1]{\smallbreak\noindent\textsc{Case:\enspace {#1}.}}

\input defs

\input lambda-defs

\input local_defs

\def\den#1#2{\llbracket{#1}\rrbracket{#2}}

\def\size{\mathop{\lstinline!size!}}

\def\code#1{{\lstinline!#1!}}

\def\bind#1#2{{{#1}.{#2}}}

\renewcommand{\typingG}[3][{}]{\typejudge[#1]\sctx{#2}{#3}}

\usepackage{hyperref}

\begin{document}

% \title{Denotational cost semantics for functional languages with 
% inductive types}

\author{Norman Danner}
\address{Wesleyan University, USA}
\email{ndanner@wesleyan.edu}
\thanks{Norman Danner's research is supported by 
the National Science Foundation under grant no.\ 1318864.}

\author{Daniel R. Licata}
\address{Wesleyan University, USA}
\email{dlicata@wesleyan.edu}
\thanks{Daniel R. Licata's research is sponsored by The United States Air
Force Research Laboratory under agreement number FA9550-15-1-0053. The
U.S. Government is authorized to reproduce and distribute reprints for
Governmental purposes notwithstanding any copyright notation thereon.
The views and conclusions contained herein are those of the authors and
should not be interpreted as necessarily representing the official
policies or endorsements, either expressed or implied, of the United
States Air Force Research Laboratory, the U.S. Government or Carnegie
Mellon University.}

\author{Ramyaa Ramyaa}
\address{Wesleyan University, USA}
\curraddr{New Mexico Tech, USA}
\email{ramyaa@wesleyan.edu}



\titlecomment{A version of this paper appears in \emph{ICFP 2015}.}

\begin{abstract}
A central method for analyzing the asymptotic complexity of a functional
program is to extract and then solve a recurrence that expresses
evaluation cost in terms of input size. The relevant notion of input
size is often specific to a datatype, with measures including the length
of a list, the maximum element in a list, and the height of a tree.  In
this work, we give a formal account of the extraction of cost and size
recurrences from higher-order functional programs over inductive
datatypes. Our approach allows a wide range of programmer-specified
notions of size, and ensures that the extracted recurrences correctly
predict evaluation cost.  To extract a recurrence from a program, we
first make costs explicit by applying a monadic translation from the
source language to a complexity language, and then abstract datatype
values as sizes.  Size abstraction can be done semantically, working in
models of the complexity language, or syntactically, by adding rules to
a preorder judgement.  We give several different models of the
complexity language, which support different notions of size.
Additionally, we prove by a logical relations argument that recurrences
extracted by this process are upper bounds for evaluation cost; the
proof is entirely syntactic and therefore applies to all of the models
we consider.
\end{abstract}

\maketitle

\section{Introduction}
\label{sec:intro}

The typical method for analyzing the asymptotic complexity of a
functional program is to extract a recurrence that relates the
function's running time to the size of the function's input, and then solve the
recurrence to obtain a closed form and big-$O$ bound.  Automated
complexity analysis (see the related work in
Section~\ref{sec:related_work}) provides helpful information to
programmers, and could be particularly useful for giving feedback to
students.  In a setting with higher-order functions and
programmer-defined datatypes, automating the extract-and-solve method
requires a generalization of the standard theory of recurrences.  This
generalization must include a notion of recurrence for higher-order
functions such as $\mathtt{map}$ and $\mathtt{fold}$, as well as a
general theory of what constitutes ``the size of the input'' for
programmer-defined datatypes.

One notion of recurrence for higher-order functions was developed in
previous work by \citet{danner-royer:two-algs} and
\citet{danner-et-al:plpv13}.  Because the output of one function is the
input to another, it is necessary to extract from a function not only a
recurrence for the running time, but also a recurrence for the size of
the output.  These can be packaged together as a single recurrence that,
given the size of the input, produces a pair consisting of the running
time (called the \emph{cost}) and the size of the output (called the
\emph{potential}).  Whereas the former is the cost of executing the
program to a value, the latter determines the cost of using that value.
This generalizes naturally to higher-order functions: a recurrence for a
higher-order function is itself a higher-order function, which expresses
the cost and potential of the result in terms of a given recurrence for
the cost and potential of the argument function.  The process of
extracting recurrences can thus be seen as a denotational semantics of
the program, where a function is interpreted as a function from input
potential to cost and output potential.

Building on this work, we give a formal account of the extraction of
recurrences from higher-order functional programs over inductive
datatypes, focusing how to soundly allow programmer-specified sizes of
datatypes.  We show that under some mild conditions on sizes, the cost
predicted by an extracted recurrence is in fact an upper bound on the
number of steps the program takes to evaluate.  The size of a value
can be taken to be (essentially) the value itself, in which case one
gets exact bounds but must reason about all the details of program
evaluation, or the size of a value can forget information
(e.g. abstracting a list as its length), in which case one gets weaker
bounds with more traditional reasoning.

We start from a call-by-value source language, defined in
Section~\ref{sec:src_lang}, with strictly positive inductive datatype
definitions (which include lists and finitely branching trees, as well
as infinitely branching trees).  Datatypes are used via case-analysis
and structural recursion (so the language is terminating), 
but unlike in~\citet{danner-et-al:plpv13},
recursive calls are only evaluated if necessary---for example, recurring
on one branch of a tree has different cost than recurring on both
branches.  The cost of a program is defined by a standard operational
cost semantics, an evaluation relation annotated with costs.  For
simplicity, the cost semantics measures only the number of
function applications and recursive calls made during evaluation, but
our approach to extracting recurrences generalizes to other cost models.

We extract a recurrence from such a program in two steps.  First, in
Section~\ref{sec:complexity_lang}, we make the cost of evaluating a
program explicit, by translating a source program~$e$ to a
program~$\trans e$ in a complexity language.  The complexity language
has an additional type \C\/ for costs, and the translation to the
complexity language is a call-by-value monadic translation into the
writer monad $\C \times -$~\citep{moggi:ic91:monads,wadler:popl92:essence}.  The translated
program $\trans e$ returns an additional result, which is the cost of
running the original program $e$.

Second, we abstract values to sizes; we study both semantic and
syntactic approaches.  In Section~\ref{sec:size_semantics}, we give a
size-based semantics of the complexity language, which relies on
programmer-specified size functions mapping each datatype to the natural
numbers (or some other preorder).  Typical size functions include the
length of a list and the size or depth of a tree.  The semantics
satisfies a \emph{bounding theorem} (Theorem~\ref{thm:bounding}), which
implies that the denotational cost given by composing the
source-to-complexity translation with the size-based semantics is in
fact an upper bound on the operational cost.  We show on some examples
that the recurrence or cost extracted by this process is the expected
one; we also will later show that all examples in
\citet{danner-et-al:plpv13} carry over.

Alternatively, the abstraction of values to sizes can be done
syntactically in the complexity language, by imposing a preorder
structure on the values of the datatype themselves.  For example, rather
than mapping lists to numbers representing their lengths, we can order
the list values by rules including $\tmle {xs} {(x \mathtt{::} xs)}$ and
$\tmle {(x \mathtt{::} xs)} {(y \mathtt{::} xs)}$.  The second rule says
that the elements of the list are irrelevant, quotienting the lists down
to natural numbers, and the first generates the usual order on natural
numbers.  Formally, we equip the complexity language with a judgement
$\tmle {\mtm}{\mtm'}$ that can be used to make such abstractions.  In
Section~\ref{sec:monotonic}, we identify properties of this judgement
that are sufficient to prove a syntactic bounding theorem
(Theorem~\ref{thm:mon-bounding}), which states that the operational cost
is bounded by the cost component of the complexity translation.  The key
technical notion is a logical relation between the source and complexity
languages that extends the bounding relation of
\citet{danner-et-al:plpv13} to inductive types.  This proof gives a
bounding theorem for any model of the complexity language that validates
the rules for $\le$.  In Section~\ref{sec:mon-interp-examples}, we show
that these rules are valid in the size-based semantics of
Section~\ref{sec:size_semantics} (thereby proving
Theorem~\ref{thm:bounding}), and we discuss several other models of the
complexity language.

This gives a formal account of what it means to extract a recurrence
from higher-order programs on inductive data.  We leave an investigation
of what it means to solve these higher-order recurrences to future work.

%% before comparing our approach with other work on extracting cost
%% information from programs and discussing future directions for this
%% research

\section{Source Language with Inductive Data Types}
\label{sec:src_lang}

The source language is a simply-typed $\lambda$-calculus with product
types, function types, suspensions, and strictly positive inductive
datatypes.  Its syntax, typing, and operational semantics are given in
Figure~\ref{fig:source_lang}.  We bundle sums and inductive types
together as datatypes, rather than using separate $+$ and $\mu$ types,
because below we do not want to consider sizes for the sum part
separately.  

We assume a top-level signature $\psi$ consisting of
datatype declarations of the form
\[
\sdatadecl*{\D}{\sconstrdecl{\sC_0^\D}{\subst*{\phi_{\sC_0}}{\D}}\sconstrsep\dots\sconstrsep\sconstrdecl{\sC^\D_{n-1}}{\subst*{\phi_{\sC_{n-1}}}{\D}}}
\]
Each constructor's argument type is specified by a strictly positive
functor $\phi$.  These include the identity functor ($t$), representing
a recursive occurrence of the datatype; constant functors ($\tau$),
representing a non-recursive argument; product functors ($\phi_1 \times
\phi_2$), representing a pair of arguments; and constant exponentials
($\tau \to \phi$), representing an argument of function type.  
% For
% example for $\tau \: \slist$, the argument type for $\snil$ is $\sunit$
% (constant functor), and the argument type for $\scons$ is $\tau \times
% t$ (product of constant and recursive arguments).  
We write
$\subst\sspf\stp t$ or just $\subst*\sspf\stp$ for substitution of the
type~$\stp$ for the single
free type variable $t$ in $\phi$.  
We frequently
drop the indexing superscripts, write
$\sdatadecl{\D}{\sC}{\phi_\sC}$, and write $\sC$ rather than~$\sC_i$
to refer to one of the constructors of the declaration.  
In the
signature, each $\phi_{C}$ in each $\sdatatypekw$ declaration must refer
only to datatypes that are declared earlier in the sequence, to avoid
introducing general recursive datatypes.  
We write $\sC\oftype(\sspf\to\D)\in\ssig$
to mean that the signature~$\ssig$ contains a datatype declaration of
the form
$\sdatadecl*{\D}{\dots\sconstrsep\sconstrdecl{\sC}{\subst*\sspf\D}\sconstrsep\dots}$.
The formal definitions
of signatures, types, and constructor arguments are given
in Figure~\ref{fig:src_types_sigs}.  

\begin{figure}[t]
\captionsetup{singlelinecheck=off}
%\begin{minipage}{\textwidth}
%Signatures:\enspace $\swfsig{\ssig}$.
\caption*{Signatures: $\swfsig{\ssig}$.}
\begin{gather*}
\ndAXC{$\swfsig{\emptysig}$}
\DisplayProof
\qquad
  \AXC{$\delta \notin \psi$}
  \AXC{$\forall C \: (\ok{\ssig}{\phi_C})$}
\ndBIC{$\swfsig{\ssig,\sdatadecl{\D}{\sC}{\subst*{\sspf_{\sC}}{\D}}}$}
\DisplayProof
\end{gather*}
%\end{minipage}

\medbreak
\caption*{Types:\enspace $\istype\ssig\stp$.}
\begin{gather*}
\ndAXC{$\istype\ssig\sunit$}
\DisplayProof
\qquad
  \AXC{$\istype\ssig{\stp_0}$}
  \AXC{$\istype\ssig{\stp_1}$}
\ndBIC{$\istype\ssig{\sprod{\stp_0}{\stp_1}}$}
\DisplayProof
\qquad
  \AXC{$\istype\ssig{\stp_0}$}
  \AXC{$\istype\ssig{\stp_1}$}
\ndBIC{$\istype\ssig{\sarr{\stp_0}{\stp_1}}$}
\DisplayProof
\\[0.5\baselineskip]
  \AXC{$\istype\ssig\stp$}
\ndUIC{$\istype\ssig{\ssusp\stp}$}
\DisplayProof
\qquad
  \AXC{$\delta \in \psi$}
\ndUIC{$\istype{\psi}{\D}$}
\DisplayProof
\end{gather*}

\medbreak
\caption*{Constructor arguments:\enspace $\ok\ssig\sspf$.}
\begin{gather*}
\ndAXC{$\ok\ssig\reccall$}
\DisplayProof
\qquad
  \AXC{$\istype\ssig\stp$}
\ndUIC{$\ok\ssig\stp$}
\DisplayProof
\qquad
  \AXC{$\ok\ssig{\sspf_0}$}
  \AXC{$\ok\ssig{\sspf_1}$}
\ndBIC{$\ok\ssig{\sprod{\sspf_0}{\sspf_1}}$}
\DisplayProof
\qquad
  \AXC{$\istype\ssig{\stp}$}
  \AXC{$\ok\ssig{\sspf}$}
\ndBIC{$\ok\ssig{\sarr{\stp}{\sspf}}$}
\DisplayProof
\end{gather*}

\captionsetup{singlelinecheck=on}
\caption{Valid signatures, types, and constructor arguments.}
\label{fig:src_types_sigs}
\hrule
\end{figure}

We define the
expressions~$\stm$ and typing judgment $\typejudgeS\stm\stp$ in
Figure~\ref{fig:source_lang}.  As we will do in most of the rest
of the paper, here we
elide reference to the signature and just refer to types and
constructor arguments.  On the occasion when precision is crucial,
we notate the typing judgment
with the signature, as in $\typejudgeS[\ssig]\stm\stp$.

%%% Source language types, expressions, typing.
\begin{figure}
\captionsetup{singlelinecheck=off}
%% Types.
\caption*{Types:}
\[
\begin{array}{rcl}
\tau &::= & \sunit \mid \tau\cross\tau \mid \tau\arrow\tau \mid \ssusp\tau \mid \D  \\
\phi &::= &t \mid \tau \mid \phi\cross\phi \mid \tau\to\phi \\
\sdatatype\:\D &=
&\sconstrdecl{{\sC}^\D_0}{\subst*{\phi_{\sC_0}}\D}\sconstrsep\dots\sconstrsep\sconstrdecl{{\sC}^\D_{n-1}}{\subst*{\phi_{\sC_{n-1}}}\D}
\end{array}
\]

\medbreak
%% Expressions.
\caption*{Expressions:}
\[
\begin{array}{rcl}
v &::= &x \mid\striv \mid \spair v v \mid \lambda x.e \mid \sdelaykw(e) \mid \sC\,v \\
e &::= 
  &x \mid \striv \mid \spair e e \mid \ssplit e x x e \mid \slam x e \mid e\,e  \\
& &\mid  \sdelaykw(e) \mid \sforcekw(e) \\
& &\mid  \sC^\D\, e
    \mid 
	\srec[\D]{e}{\sC}{\bind{x}{e_{\sC}}} \\
& &\mid  \smapkw^\phi(\bind x v, v) \mid \sletkw(e, \bind x e)\\
n &::= & 0 \mid 1 \mid n + n
\end{array}
\]

\medbreak
%% Typing.
\caption*{Typing: $\typingG{\stm}{\stp}$.}
\begin{gather*}
% Identifiers
\ndAXC{$\typing[\sctx,x\oftype\sigma]x\sigma$}
\DisplayProof
\\
% pair, projection.
\ndAXC{$\typingG{\striv}{\sunit}$}
\DisplayProof
\qquad
  \AXC{$\typingG{e_0}{\tau_0}$}
  \AXC{$\typingG{e_1}{\tau_1}$}
\ndBIC{$\typingG{\spair{e_0}{e_1}}{\tau_0\cross\tau_1}$}
\DisplayProof
\qquad
  \AXC{$\typejudge\sctx{\stm_0}{\stp_0\cross\stp_1}$}
  \AXC{$\typejudge{\sctx,x_0\oftype\stp_0,x_1\oftype\stp_1}{\stm_1}{\stp}$}
\ndBIC{$\typejudge\sctx{\ssplit{\stm_0}{x_0}{x_1}{\stm_1}}{\stp}$}
\DisplayProof
\\
% term lambda, application.
  \AXC{$\typing[\sctx,x\oftype\sigma]{e}{\tau}$}
\ndUIC{$\typingG{\lambda x.e}{\sigma\arrow\tau}$}
\DisplayProof
\qquad
  \AXC{$\typingG{e_0}{\sigma\arrow\tau}$}
  \AXC{$\typingG{e_1}{\sigma}$}
\ndBIC{$\typingG{e_0\,e_1}\tau$}
\DisplayProof
\\
% delay, force.
  \AXC{$\typingG{e}{\tau}$}
\ndUIC{$\typingG{\sdelaykw(e)}{\ssusp\tau}$}
\DisplayProof
\qquad
  \AXC{$\typingG e {\ssusp\tau}$}
\ndUIC{$\typingG{\sforcekw(e)}{\tau}$}
\DisplayProof
\\
% constructors, rec.
  \AXC{$\typingG{e}{\subst*{\phi_\sC}{\D}}$}
\ndUIC{$\typingG{\sC^\D\,e}{\D}$}
\DisplayProof
\qquad
  \AXC{$\typingG e \D$}
  \AXC{$\forall \sC\left(\typing[\sctx,x\oftype{\subst*{\phi_\sC}{\D\cross\ssusp\tau}}] {e_\sC} {\tau}\right)$}
\BIC{$\typingG{\srec[\D] e \sC {\bind x {e_\sC}}}{\tau}$}
\DisplayProof
\\
% map, let.
  \AXC{$\typing[\sctx,x\oftype\tau_0]{v_1}{\tau_1}$}
  \AXC{$\typingG {v_0} {\subst*\phi{\tau_0} }$}
\ndBIC{$\typingG{\smapkw^\phi(\bind x {v_1}, {v_0})}{\subst*{\phi}{\tau_1}}$}
\DisplayProof
\qquad
  \AXC{$\typingG{e_0}{\sigma}$}
  \AXC{$\typing[\sctx,x\oftype\sigma]{e_1}{\tau}$}
\ndBIC{$\typingG{\sletkw(e_0, \bind x {e_1})}{\tau}$}
\DisplayProof
\end{gather*}

\captionsetup{singlelinecheck=on}
\caption{Source language syntax and typing.}
\label{fig:source_lang}
\hrule
\end{figure}

\begin{figure}
\begin{minipage}{\textwidth}
Operational semantics: $\evalin e v n$.
\begin{gather*}
% Pairing, projection.
  \AXC{$\evalin{e_0}{v_0}{n_0}$}
  \AXC{$\evalin{e_1}{v_1}{n_1}$}
\BIC{$\evalin{\spair{e_0}{e_1}}{\spair{v_0}{v_1}}{n_0+n_1}$}
\DisplayProof
\qquad
  \AXC{$\evalin{e_0}{\spair{v_0}{v_1}}{n_0}$}
  \AXC{$\evalin{\cl{e_1}{v_0/x_0,v_1/x_1}}{v}{n_1}$}
\BIC{$\evalin{\ssplit{e_0}{x_0}{x_1}{e_1}}{v}{n_0+n_1}$}
\DisplayProof
\\[.5\baselineskip]
%
% Term application.
  \AXC{$\evalin{e_0}{\slam x {e_0'}}{n_0}$}
  \AXC{$\evalin{e_1}{v_1}{n_1}$}
  \AXC{$\evalin{\subst{e_0'}{v_1}{x}}{v}{n}$}
\TIC{$\evalin{e_0\,e_1}{v}{n_0+n_1+n}$}
\DisplayProof
\\[.5\baselineskip]
%
% Delay, force.
\ndAXC{$\evalin{\sdelaykw(e)}{\sdelaykw(e)}{0}$} 
\DisplayProof
\quad
  \AXC{$\evalin{e}{\sdelaykw(e_0)}{n_0}$}
  \AXC{$\evalin{e_0}{v}{n_1}$}
\ndBIC{$\evalin{\sforcekw(e)}{v}{n_0+n_1}$}
\DisplayProof
\\[.5\baselineskip]
%
% Constructors.
  \AXC{$\evalin{e}{v}{n}$}
\UIC{$\evalin{\sC e}{\sC v}{n}$}
\DisplayProof
\qquad
%
% Recursor.
  \AXC{$\evalin{e}{\sC\,v_0}{n_0}$}
  \AXC{$\evalin{\smapkw^{\phi_\sC}(\bind y{\spair{y}{\sdelaykw(\srec{y}{\sC}{\bind{x}{e_\sC}})}}, v_0)}{v_1}{n_1}$}
  \AXC{$\evalin{\subst{e_\sC}{v_1}{x}}{v}{n_2}$}
\TIC{$\evalin{\srec{e}{\sC}{\bind{x}{e_\sC}}}{v}{1+n_0+n_1+n_2}$}
\DisplayProof
\\[.5\baselineskip]
%
% Map.
\ndAXC{$\evalin{\smapkw^t(\bind x {v}, v_0)}{\subst {v} {v_0} x}{0}$}
\DisplayProof
\qquad
\RightLabel{($t$ not free in~$\tau$)}
\ndAXC{$\evalin{\smapkw^\tau(\bind x {v}, v_0)}{v_0}{0}$}
\DisplayProof
\\
  \AXC{$\evalin{\smapkw^{\phi_0}(\bind x v, v_0)}{v_0'}{n_0}$}
  \AXC{$\evalin{\smapkw^{\phi_1}(\bind x v, v_1)}{v_1'}{n_1}$}
\ndBIC{$\evalin{\smapkw^{\phi_0\cross\phi_1}(\bind x v,\spair{v_0}{v_1})}{\spair{v_0'}{v_1'}}{n_0+n_1}$}
\DisplayProof
\\[-.25\baselineskip]
\ndAXC{$\evalin{\smapkw^{\tau\to\phi}(\bind x v,\lambda y.e)}{\lambda y.\sletkw(e, \bind z {\smapkw^\phi(\bind x v, z)})}{0}$}
\DisplayProof
\\[.5\baselineskip]
%
% Let.
  \AXC{$\evalin{e_0}{v_0}{n_0}$}
  \AXC{$\evalin{\subst{e_1}{v_0}{x}}{v}{n_1}$}
\ndBIC{$\evalin{\sletkw(e_0, \bind x {e_1})}{v}{n_0+n_1}$}
\DisplayProof
\end{gather*}
\end{minipage}
\caption{Source language operational semantics.}
\label{fig:source_semantics}
\hrule
\end{figure}

Evaluation (defined in Figure~\ref{fig:source_semantics})
is call-by-value and products and datatypes are strict.
However, unfolding datatype recursors requires substituting expressions
(the recursor applied to the components of the value) for the variables
standing for the recursive calls---running the recursive call first and
substituting its value would require a function to make all possible
recursive calls.  We handle this using suspensions: when computing a
$\tau$ by recursion, the result of a recursive call is given the type
$\ssusp\tau$.  The values of type $\ssusp \tau$ are $\sdelaykw(e)$ where
$e$ is an expression of type $\tau$; the elimination form $\sforcekw$
forces evaluation.  In general, when defining a recursive computation of
result type~$\tau$, the branch for a constructor $C$, $e_C$, has access
to a variable of type $\subst*{\phi_C}{\D\cross\ssusp\tau}$, which gives
access both to the ``predecessor'' values of type $\delta$ and to the
recursive results.  This recursor supports both case-analysis and
structural recursion, and recursive calls are only computed if they are
used.

For any strictly positive functor $\phi$, the $\smapkw^\phi$ expression
witnesses functoriality, essentially lifting a function $\tau_0 \to
\tau_1$ to a function $\subst*{\phi}{\tau_0} \to \subst*{\phi}{\tau_1}$.
It is used in the operational semantics for the recursor to insert
recursive calls at the right places in $\phi$ (\citet{harper:pfpl}
provides an exposition).  We will only need to lift maps $x:\tau_0.v:\tau_1$
whose bodies are syntactic values (or variables), and apply them to
syntactic values (or variables), and we restrict $\smapkw$ to this
special case to simplify its cost semantics.

A couple of examples may
be more edifying than the formalism.  In these examples and the
future, we use a sugared syntax of pattern variables for the constructor
arguments.  So in our first example we write
$\sreckw(\dots,\sconstrdecl{N}{\bind{\langle n,\spair{t_0}{r_0},\spair{t_1}{r_1}\rangle}{e_N}})$, where
$e_N = e_N(n,t_0,r_0,t_1,r_1)$ as syntactic sugar for
$\sreckw(\dots,\sconstrdecl{N}{\bind x {e_N'}})$, where
\[
e_N' =
\ssplit{x}{n}{y}{\ssplit{y}{u}{v}{\ssplit{u}{t_0}{r_0}{\ssplit{v}{t_1}{r_1}{e_N}}}}.
\]

As a first example, consider the type of $\sint$-labeled binary trees:
\[
\sdatadecl*{\stree}{\sconstrdecl{E}{\sunit}\sconstrsep\sconstrdecl{N}{\sint\cross\stree\cross\stree}}
\]
Now consider a recursive definition
$\srecm{N(n, t_0, t_1)}{E\mapsto\bind x{e_E},N\mapsto\bind x{e_N}}$.
For the $N$-clause,
$x\oftype\sint\cross(\stree\cross\ssusp\tau)\cross(\stree\cross\ssusp\tau)$.
Thus the evaluation must
substitute $(n, (t_0, r_0), (t_1, r_1))$ for $x$ in~$e$, where
$r_i$ is the result of the recursive call on the subtree~$t_i$.
In this case, to evaluate
$\srecm{N(n, t_0, t_1)}{\ldots}$, we set
$e = \sdelaykw(\srecm{x}{\ldots})$ and compute
\begin{align*}
\smapkw^{\sint\cross t\cross t}(\bind x{\spair{x}{e}}, (n, t_0, t_1))
&= \langle\smapkw^{\sint}(\bind {x}{\spair x e}, n), \smapkw^t(\bind x {\spair x e}, t_0), \smapkw^t(\bind x {\spair x e}, t_1)\rangle \\
&= \langle n, \spair{t_0}{\subst e{t_0}x}, \spair{t_1}{\subst e{t_1}x}\rangle
\end{align*}
and substitute the result for~$x$ in $e_N$.

As a second example, consider the type of infinite, infinitely-branching
$\sint$-labeled trees:
\[
\sdatadecl*{\stree'}{\sconstrdecl{N}{\sint\cross(\snat\to\stree')}}.
\]
Now consider the evaluation of
$\srecm{N(n, \lambda y.e_0)}{N\mapsto \bind {\spair z f}{e_N}}$ where
$\spair z f\oftype \sint\cross(\snat\to(\stree'\cross\ssusp\tau))$.  Set
$e = \sdelaykw(\srecm{x}{N\mapsto \bind {\spair z f}{e_N}})$, compute
\begin{align*}
\smapkw^{\sint\cross(\snat\to t)}(\bind x{\spair x e}, \spair{n}{\lambda y.e_0})
&= \spair{\smapkw^{\sint}(\bind x{\spair x e}, n)}{\smapkw^{\snat\to t}(\bind x{\spair x e}, \lambda y.e_0)} \\
&= \spair{n}{\lambda y.\sletkw(e_0, \bind z {\smapkw^t(\bind x{\spair x e}, z)})} \\
&= \spair{n}{\lambda y.\sletkw(e_0, \bind z {\spair{z}{\subst e z x}})} \\
&= \spair{n}{\lambda y.\spair{e_0}{\subst e{e_0}{x}}}.
\end{align*}
Now subsitute the result for $\spair z f$ in $e_N$.  Presumably $e_N$
has a subexpression of the form $f\,q\oftype\stree'\cross\ssusp\tau$.  This
last substitution has the moral effect of replacing $f\,q$ with
$\spair{e_0(q)}{\subst{e}{e_0(q)}{x}}$; the first component is
the subtree, and the second is the result of the recursive call
at that subtree.

The cost semantics in Figure~\ref{fig:source_semantics} defines the
relation $\evalin e v n$, which means that the expression~$e$ evaluates
to the value $v$ in $n$ steps.  Our cost model charges only for the
number of function applications and recursive calls made by datatype
recursors.  This prevents constant-time overheads from the encoding of
datatypes using product and suspension types
%% (using $\ssplitkw$, $\sforcekw$, $\smapkw$, and
%% $\sletkw$) 
from showing up in the extracted recurrences.  It is simple
to adapt the denotational cost semantics below to other operational cost
semantics, such as one that charges for these steps, or assigns
different costs to different constructs.
%%ENH: could explain why products could be exposed but suspensions are
%%     not intended to be. 

Substitutions are defined as usual:
%% Here as elsewhere we
%% silently use inversion principles for typing and evaluation; e.g., from
%% the assumption that $\typingG{\smap \phi {\bind x {v_1}}
%%   {v_0}}{\subst*\phi{\tau_1}}$, we assume that
%% $\typejudge{\sctx,\tmoftype x {\tau_0}}{v_1}{\tau_1}$ and
%% $\typingG{v_0}{\subst*\phi{\tau_0}}$.

\begin{defi}
We write $\ssub$ for substitutions $v_1/x_1,\ldots,v_n/x_n$, 
and $\ssub\oftype\sctx$ to mean 
that $\dom\ssub\subseteq\dom\sctx$ and
$\emptyset \vdash {\ssub(x)} : {\sctx(x)}$ for all $x\in\dom\ssub$.
%% We write $\ssub,v/x$ for the extension of $\ssub$ obtained by binding
%% $x$ to~$v$.
We define the application of a substitution~$\ssub$ to an expression~$e$
as usual and denote it $\tsubst\stm\ssub$.
\end{defi}

\begin{lemma} \label{lem:source-subst-composition}
If $x$ does not occur in $\ssub$, then 
$\subst{\tsubst{\stm}{\ssub,x/x}}{\stm_1}{x} = \tsubst{\stm}{\ssub,\stm_1/x}$.
\end{lemma}

For source cost expressions $\sc$, we write \homj{}{\sc}{\sc'} for the
order given by interpreting these cost expressions as natural numbers
(i.e. the free precongruence generated by the monoid equations for
$(+,0)$ and $\homj{}{0}{1}$).  
We have the following syntactic
properties of evaluation:

\begin{lemma}[Value Evaluation] \label{lem:value-eval-inv}~
\begin{itemize}
\item If \seval{\sv}{\sc}{\sv'} then \homj{}{\sc}{\costzero} and $\sv =
  \sv'$.  
\item For all $\sv$, \seval{\sv}{\costzero}{\sv}.
\end{itemize}
\end{lemma}

\begin{lemma}[Totality of $\smapkw$]
\label{lem:map_no_cost}
If $\typingG{\smapkw^\phi(\bind x {v_1}, v_0)}{\subst*{\phi}{\tau_1}}$ then
$\evalin{\smapkw^{\phi}(\bind x {v_1}, v_0)}{v}{0}$ for some~$v$.
\end{lemma}
%% \begin{proof}
%% The proof is by induction on~$\phi$; the relevant point is that a
%% substitution $\subst{v_1}{v_0}{x}$ is itself a value.
%% \end{proof}

%% ENH: the way we've set up the costs now, couldn't we make map admissible?

\section{Making Costs Explicit}
\label{sec:complexity_lang}

\subsection{The Complexity Language}

%%% Complexity language types, expressions, typing.
\begin{figure*}
\captionsetup{singlelinecheck=off}
%% Types.
\caption*{Types:}
\[
\begin{array}{rcl}
\mtp &::= & \C \mid \cunit \mid \mdtp\mid\mtp\cross\mtp \mid \mtp\arrow\mtp \\
\mspf &::= &t \mid \mtp \mid \mspf\cross\mspf \mid \mtp\to\mspf \\
\cdatatypekw\:\mdtp &=
&\cconstrdecl{{\cC}^\mdtp_0}{\subst*{\mspf_{\cC_0}}\mdtp}\sconstrsep\dots\sconstrsep  \sconstrdecl{{\cC}^\mdtp_{n-1}}{\subst*{\mspf_{\cC_{n-1}}}\mdtp}
\\
\\\end{array}
\]
\medbreak

%% Expressions.
\caption*{Expressions:}
\[
\begin{array}{rcl}
\mtm &::= 
  &x \mid \costzero\mid\costone\mid\mtm\costplus{\mtm}\mid  \\
  & & \striv \mid \spair \mtm \mtm \mid \fst\mtm \mid \snd\mtm \mid \slam x \mtm \mid \mtm\,\mtm  \\
& &\mid  \cC^\mdtp\, \mtm
    \mid 
	\crec[\mdtp]{\mtm}{\sC}{\bind{x}{\mtm_{\sC}}} 
\end{array}
\]
\medbreak

%% Typing.
\caption*{Typing:  $\typejudgeM\mtm\mtp$.}
\begin{gather*}
% Identifiers
\ndAXC{$\typing[\mctx,x\oftype\mtp]x\mtp$}
\DisplayProof
\\
% Costs.
\ndAXC{$\typejudgeM\costzero\C$}
\DisplayProof
\qquad
\ndAXC{$\typejudgeM\costone\C$}
\DisplayProof
\qquad
  \AXC{$\typejudgeM{\mtm_0}\C$}
  \AXC{$\typejudgeM{\mtm_1}\C$}
\ndBIC{$\typejudgeM{\mtm_0\costplus\mtm_1}{\C}$}
\DisplayProof
\\
% pair, projection.
\ndAXC{$\typejudgeM{\striv}{\cunit}$}
\DisplayProof
\qquad
  \AXC{$\typejudgeM{\mtm_0}{\tau_0}$}
  \AXC{$\typejudgeM{\mtm_1}{\tau_1}$}
\ndBIC{$\typejudgeM{\spair{\mtm_0}{\mtm_1}}{\tau_0\cross\tau_1}$}
\DisplayProof
\qquad
  \AXC{$\typejudgeM\mtm{\mprod{\mtp_0}{\mtp_1}}$}
\ndUIC{$\typejudgeM{\mproj_i\mtm}{\mtp_i}$}
\DisplayProof
\\
% \qquad
%   \AXC{$\typejudgeM{e}{\tau_0\cross\tau_1}$}
% \ndUIC{$\typejudgeM{\pi_i(e)}{\tau_i}$}
% \DisplayProof
% term lambda, application.
  \AXC{$\typing[\mctx,x\oftype\mtp_0]{\mtm}{\mtp_1}$}
\ndUIC{$\typejudgeM{\lambda x.\mtm}{\mtp_0\arrow\mtp_1}$}
\DisplayProof
\qquad
  \AXC{$\typejudgeM{\mtm_0}{\mtp_0\arrow\mtp_1}$}
  \AXC{$\typejudgeM{\mtm_1}{\mtp_0}$}
\ndBIC{$\typejudgeM{\mtm_0\,\mtm_1}\mtp_1$}
\DisplayProof
\\
% constructors, rec.
  \AXC{$\typejudgeM{\mtm}{\subst*{\mspf_\cC}{\mdtp}}$}
\ndUIC{$\typejudgeM{\cC^\mdtp\,\mtm}{\mdtp}$}
\DisplayProof
\\
  \AXC{$\typejudgeM \mtm \mdtp$}
  \AXC{$\forall \cC\left(\typing[\mctx,x\oftype{\subst*{\mspf_\cC}{\mdtp\cross\mtp}}] {\mtm_\cC} {\mtp}\right)$}
\BIC{$\typejudgeM{\crec[\mdtp] \mtm \cC {\bind x {\mtm_\cC}}}{\mtp}$}
\DisplayProof
\\
\end{gather*}

\captionsetup{singlelinecheck=on}
\caption{Complexity language types, expressions, and typing.}
\label{fig:cpy_lang}
\hrule
\end{figure*}

%% \begin{figure*}
%% \[
%% \begin{aligned}
%% % \mproj_i\mpair{\mtm_0}{\mtm_1} &= \mtm_i \\
%% % (\slam x{\mtm_0})\mtm_1 &= \subst{\mtm_0}{\mtm_1}{x}
%% \end{aligned}
%% \qquad
%% \begin{aligned}
%% % \crec[\D]{\sC\mtm}{\sC}{\bind x {\mtm_C}} &=
%% % \subst{\mtm_C}{\cmap{\mspf}{\bind y {\mpair{y}{\crec[\D]{y}{\sC}{\bind x{\mtm_C}}}}}{E}}{x} \\
%% \cmap t {\bind x {\mtm}}{\mtm_0} &= \subst{\mtm}{\mtm_0}{x} \\
%% \cmap T {\bind x {\mtm}}{\mtm_0} &= \mtm_0 \\
%% \cmap {\mprod{\mspf_0}{\mspf_1}} {\bind x {\mtm}} {\mtm_0} &=
%%   \mpair{\cmap{\mspf_0}{\bind x \mtm}{\cproj_0\mtm_0}}{\cmap{\mspf_1}{\bind x \mtm}{\cproj_1\mtm_0}} \\
%% \cmap {\arr\mtp\mspf} {\bind x\mtm}{\mtm_1}
%%   &= \lambda y.\cmap\mspf{\bind x\mtm}{\mtm_1\,y}
%% \end{aligned}
%% \]
%% \caption{Complexity language equational semantics.}
%% \label{fig:cpy_lang_eqsem}
%% \end{figure*}

The complexity language will serve as a monadic
metalanguage \citep{moggi:ic91:monads} in which we make evaluation cost
explicit.  The syntax and typing are given in Figure~\ref{fig:cpy_lang}.
The preorder judgement defined in Section~\ref{sec:monotonic} will play
a role analogous to an operational or equational semantics for the complexity
language.
%% We write $\Psi$ for a datatype signature,
%% which is analogous to the source language $\psi$.

Because we are not concerned with the evaluation steps of the complexity
language itself, we remove features of the source language
that were used to control evaluation costs.  Product types are
eliminated by projections, rather than $\ssplitkw$.  We allow
substitution of arbitrary expressions for variables, which is used in
recursors for datatypes.  Consequently, suspensions are not necessary.
We treat $\cmap{\Phi}{x.\mtm}{\mtm_1}$ as an admissible rule (macro),
defined by induction on $\Phi$:
\[
% map.
  \AXC{$\typing[\mctx,x\oftype\mtp_0]{\mtm_1}{\mtp_1}$}
  \AXC{$\typejudgeM {\mtm_0} {\subst*\mspf{\mtp_0} }$}
\ndBIC{$\typejudgeM{\cmapkw^\mspf(\bind x {\mtm_1}, {\mtm_0})}{\subst*{\mspf}{\mtp_1}}$}
\DisplayProof
\]
\[
\begin{aligned}
% \crec[\D]{\sC\mtm}{\sC}{\bind x {\mtm_C}} &=
% \subst{\mtm_C}{\cmap{\mspf}{\bind y {\mpair{y}{\crec[\D]{y}{\sC}{\bind x{\mtm_C}}}}}{E}}{x} \\
\cmap t {\bind x {\mtm}}{\mtm_0} &:= \subst{\mtm}{\mtm_0}{x} \\
\cmap T {\bind x {\mtm}}{\mtm_0} &:= \mtm_0 \\
\cmap {\mprod{\mspf_0}{\mspf_1}} {\bind x {\mtm}} {\mtm_0} &:=
  \mpair{\cmap{\mspf_0}{\bind x \mtm}{\cproj_0\mtm_0}}{\cmap{\mspf_1}{\bind x \mtm}{\cproj_1\mtm_0}} \\
\cmap {\arr\mtp\mspf} {\bind x\mtm}{\mtm_1}
  &:= \lambda y.\cmap\mspf{\bind x\mtm}{\mtm_1\,y}
\end{aligned}
\]

The type~$\C$ represents some domain of costs. The term constructors for
$\C$ say only that it is a monoid $(\costplus,\costzero)$ with a value
$\costone$ representing the cost of
a single step.  Costs can be interpreted in a variety of
ways---e.g. as natural numbers and as natural numbers with infinity
(Section~\ref{sec:size_semantics}).

Substitutions~$\msub$ in the complexity language are defined as usual,
and satisfy standard composition properties:

\begin{lemma} \label{lem:mono-subst-composition}~
\begin{itemize}
\item If $x$ does not occur in $\msub$, then 
$\subst{\tsubst{\mtm}{\msub,x/x}}{\mtm_1}{x} = \tsubst{\mtm}{\msub,\mtm_1/x}$.
\item 
If $x_1,x_2$ do not occur in $\msub$, then 
$\tsubst{\subst{\subst{\mtm}{\mtm_1}{x_1}}{\mtm_2}{x_2}}{\msub} = \tsubst{\mtm}{\msub,\tsubst{\mtm_1}{\msub}/x_1,\tsubst{\mtm_2}{\msub}/x_2}$.
\end{itemize}
\end{lemma}

\subsection{The Complexity Translation}

A notion of complexity that considers only cost is insufficient
for handling higher-order functions such as
\[
% \slistmap = \slam{f,xs}{\srecm{xs}{\snil\mapsto\snil\sconstrsep\scons\mapsto {\bind {\spair{y}{\spair{ys}{r}}}{\scons(f\,y, r)}}}}.
\begin{split}
\slistmap = \lambda (f,xs).\sreckw(&xs, \\
&\snil\mapsto\snil\\
\sconstrsep&\scons\mapsto{\bind {\spair{y}{\spair{ys}{r}}}{\scons(f\,y, \sforce r)}})
\end{split}
\]
The cost of $\slistmap(f,xs)$ depends on the cost of evaluating $f$ on
the elements of~$xs$, and hence (indirectly) on the sizes of the elements
of~$xs$.
And since $\slistmap(f,xs)$ might itself be an argument to another
function (e.g. another $\slistmap$), we also need to predict the
sizes of the elements of $\slistmap(f,xs)$, which depends on the size of
the output of $f$.  
Thus, to analyze $\slistmap$, we should be given a
recurrence for the cost and size of $f(x)$ in terms of the size of $x$,
from which we produce a recurrence that gives the cost and size of
$\slistmap(f,xs)$ in terms of the size of $xs$.  We call the size of
the value of an expression that expression's \emph{potential}, because
the size of the value determines what future uses of that value will
cost.\footnote{\emph{Use cost} would be another reasonable term for
potential.}

This discussion motivates translations~$\typot\cdot$ from source language
types to complexity types and
$\trans\cdot$ from source
language terms to complexity language terms so that if $e\oftype\tau$,
then $\trans e\oftype\C\cross\typot\tau$.  In the complexity language,
we call an 
an expression of type $\typot{\tau}$ a \emph{potential},
an expression of type $\C$ a \emph{cost}, and 
expression of type $\C\cross\typot\tau$ a \emph{complexity}.
We abbreviate $\C \times \typot
\tau$ by $\cpxy \tau$. The first component of $\trans e$ is the cost of
evaluating~$e$, and the second component of $\trans e$ is the potential
of $e$.

To gain some intuition for the full definition of potential, we 
first consider the
type-level~$0$ and~$1$ cases.  At type-level~$0$, the potential cost
of an expression is a measure of the size of that expression's value;
it is the size of the value that determines the cost the expression
contributes to the cost of future computations.  Now consider a
type-level~$1$ expression~$e_0$.  The \emph{use} of~$e_0$ is its
application to a type-level~$0$ expression~$e_1$.  The cost of such
an application is the sum of (i)~the cost of evaluating~$e_0$ to a
value~$\lambda x.e_0'$; (ii)~the cost of evaluating $e_1$ to a
value~$v_1$; (iii)~the cost of evaluating $\subst{e_0'}{v_1}{x}$; and
(iv)~a possible charge for the $\beta$-reduction.  Since (iii)
depends in part on the size of $v_1$ (i.e., the potential of~$e_1$),
by compositionality complexities must capture both cost and
potential.  Furthermore, (iii) is defined in terms of the potential
of $e_0$ (i.e., the potential of $\lambda x.e_0'$).  Thus the
potential of a type-level~$1$ expression should be a map from
type-level~$0$ potentials to type-level~$0$ complexities.

With this in mind, consider (the type of) $\slistmap$.
Its potential should describe what future uses of
$\slistmap$ will cost, in terms of the potentials of its arguments.  For
the type of $\slistmap$ (uncurried), the above discussion suggests that
$\typot{(\tau \to \sigma) \cross (\tau \: \slist) \to \sigma \: \slist}$
ought to be
$(\typot \tau \to \C \times \typot \sigma) \cross \typot {\tau \: \slist} \to \C \times \typot {\sigma \: \slist}$.
For the argument function, we are provided a recurrence that maps
$\tau$-potentials to $\sigma$-complexities. For the argument
list, we are provided a $\tau \: \slist$-potential.  Using these,
the potential of $\slistmap$ must give the cost for doing the whole
$\mathtt{map}$ and give a $\sigma \: \slist$-potential for the value.
This illustrates how the potential of a higher-order function is itself a
higher-order function.  

\newcommand\bindkw{\:\text{>>=}\:}
\newcommand\bindignorekw{\:\text{>>}\:}
\newcommand\returnkw{\mathsf{return}}

As discussed above, we stage the extraction of a recurrence, and in the
first phase, we do not abstract values as sizes (e.g. we do not replace
a list by its length).  Because of this, the complexity translation has
a succinct description.  For any monoid $(\C,+,0)$, the writer monad~\citep{wadler:popl92:essence} $\C
\times -$ is a monad with
\[
\begin{array}{l}
\mathsf{return}(E) := (0,E) \\
E_1 \bindkw E_2 := (\fst(E_1) + \fst(E_2(\snd(E_1))), \snd(E_2(\snd(E_1)))) \\
\end{array}
\]
The monad laws follow from the monoid laws for $\C$.  Thinking of $\C$
as costs, these say that the cost of $\returnkw(e)$ is zero, and that
the cost of bind is the sum of the cost of $E_1$ and the cost of $E_2$
on the potential of $E_1$.  The complexity translation is then a
call-by-value monadic translation from the source language into the
writer monad in the complexity language, where source expressions that
cost a step have the ``effect'' of
incrementing the cost component, using the monad operation
\[
\mathsf{incr}(E:\C) : \C \times \cunit := (E,\ctriv)
\]

We write this translation out explicitly in Figure~\ref{fig:expr_trans}.
When $E$ is a complexity, we write $E_c$ and $E_p$ for $\fst E$ and
$\snd E$ respectively (for ``cost'' and ``potential'').  We will often
need to ``add cost'' to a complexity; when $E_1$ is a cost and $E_2$ a
complexity, we write $\costpluscpy {E_1} {E_2}$ for the
complexity~$(E_1+(E_2)_c, (E_2)_p)$ (in monadic notation, $\mathsf{incr}(E_1)
\bindignorekw E_2$).  The type translation is extended pointwise to
contexts, so $x : \tau \in \sctx$ iff $x \oftype \typot\tau
\in\typot\sctx$---the translation is call-by-value, so variables
range over potentials, not complexities.  For example, 
$\trans x = (0, x)$, where the $x$ on the left is
a source variable and the $x$ on the right is
a potential variable.  Likewise we assume that for every
datatype~$\D$ in the source signature, we have a corresponding
datatype~$\D$ declared in the complexity language.

\begin{figure*}
\[
\begin{array}{rcl}
\cpxy\tau & = & \C\cross\typot\tau
\\
\typot\sunit & = & \cunit
\\
\typot{\sigma\cross\tau} & = & \typot\sigma\cross\typot{\tau}
\\ 
\typot{\sigma\arrow\tau} & = & \typot\sigma\arrow\cpxy\tau
\\
\typot{\ssusp\tau} & = & \cpxy\tau
\\
\typot\D & = & \D
\\ \\
\cpxy\phi & = & \C \cross \typot \phi
\\
\typot t & = & t
\\
\typot \tau & = & \typot\tau
\\
\typot{\phi_0\cross\phi_1} & = & \typot{\phi_0}\cross\typot{\phi_1}\\
\typot{\tau\to\phi} & = & \typot\tau\to\cpxy\phi
\\
\end{array}
\]
\[
\begin{array}{l}
\text{$\typot{\psi}$ has, for each datatype $\delta$ in $\psi$} \\
\cdatadecl* {\D} {\cconstrdecl{\sC^{\D}_0}{\subst*{\typot{\phi_{\sC_0}}}{\D}},\dots,\cconstrdecl{\sC^{\D}_{\sC_{n-1}}}{\subst*{\typot{\phi_{n-1}}}{\D}}}
\end{array}
\]
% \[
% \trans{\sdatadecl*{\D}{\sconstrdecl{\sC^\D_0}{\subst*{\phi_{\sC_0}}{\D}},\dots,\sconstrdecl{\sC^\D_{n-1}}{\subst*{\phi_{\sC_{n-1}}}{\D}}}} = 
% \cdatadecl* {v\D} {\cconstrdecl{\sC^{v\D}_0}{\typot{\subst*{\phi_{\sC_0}}{\D}}},\dots,\cconstrdecl{\sC^{v\D}_{\sC_{n-1}}}{\typot{\subst*{\phi_{n-1}}{\D}}}}
% \]
\[
\begin{aligned}
\\
\trans x &= \cpair 0 x \\
\trans{\striv} &= \cpair 0 {\ctriv} \\
\trans{\spair{e_0}{e_1}} &= \cpair{\trans{e_0}_c + \trans{e_1}_c}{\cpair{\trans{e_0}_p}{\trans{e_1}_p}} \\
% \trans{\pi_i(e)} &= (\trans{e}_c, \pi_i(\trans{e}_p)) \\
\trans{\ssplit{e_0}{x_0}{x_1}{e_1}} &=
% \costpluscpy{\cst{\trans{e_0}}}{\cl{\trans{e_1}}{\sproj_0{\pot{\trans {e_0}}}/x_0,\sproj_1{\pot{\trans {e_1}}}/x_1}} \\
{\cst{\trans{e_0}}}\plusc {\cl{\trans{e_1}}{\sproj_0{\pot{\trans {e_0}}}/x_0,\sproj_1{\pot{\trans {e_1}}}/x_1}} \\
\trans{\lambda x.e} &= \cpair 0 {\lambda x.e} \\
\trans{e_0\,e_1} &= \costpluscpy{(1+(e_0)_c+(e_1)_c)} {(e_0)_p(e_1)_p} \\
\trans{\sdelaykw(e)} &= \cpair 0 {\trans e} \\
\trans{\sforcekw(e)} &= \costpluscpy{\trans e_c}{\trans e_p} \\
\trans{\sC_i^\D e} &= \cpair {\trans e_c}{\cC_i^{\D}\trans{e}_p} \\
\trans{\srec[\D] e \sC {\bind x {e_\sC}}}
  &= {\trans e_c}\plusc {\crecdecl[\D] {{\trans e}_p} \cC {\bind x {\costpluscpy 1 {\trans{e_\cC}}}}} \\
\trans{\smapkw^\phi(\bind x {v_0}, v_1)}
  &= \cpair 0{\cmap{\typot\phi}{\bind x {\trans{v_0}_p}}{{\trans{v_1}}_p}} \\
\trans{\sletkw(e_0, \bind x {e_1})}
  &= \costpluscpy{\trans{e_0}_c }{ \subst {\trans{e_1}} {{\trans{e_0}}_p} {x}}
\end{aligned}
\]
\caption{Translation from source types and expressions 
to complexity types and expressions.  Recall that
$\trans e_c = \fst\trans e$ and $\trans e_p = \snd\trans e$.}
\label{fig:expr_trans}
\hrule
\end{figure*}

We note some basic facts about the translation: the type translation
commutes with the application of a strictly positive functor, which is
used to show that the translation preserves types.

\begin{lemma}[Compositionality] \label{lem:compositionality} \mbox{}
\begin{itemize}
\item \comptr{\sarg{\sspf}{\stp}} = \marg{\comptr{\sspf}}{\pottr{\stp}}
\item \pottr{\sarg{\sspf}{\stp}} = \marg{\pottr{\sspf}}{\pottr{\stp}}
\end{itemize}
\end{lemma}

\begin{thm}
If $\oftps{\gamma}{\psi}{e}{\tau}$, then $\oftps{\trans\sctx}{\trans \psi}{\trans e}{\trans\tau}$.
\end{thm}

\section{A Size-Based Complexity Semantics}
\label{sec:size_semantics}

In the above translation, the potential of a value has just as much
information as that value itself.  Next, we investigate how to
abstract values to sizes, such as replacing a list by its length.  In
this section, we make this replacement by defining a size-based
denotational semantics of the complexity language.  

We need to be able to treat potentials of inductively-defined data in
two different ways.  On the one hand, potentials must reflect intuitions
about sizes.  To that end, we will insist that potentials be partial
orders.  On the other hand, to interpret $\creckw$ expressions, we must be
able to distinguish the datatype constructor that a potential
represents.  In other words, we need the potentials to also be
(something like) inductive data types.  We will have our cake and eat it
too using an approach similar to the work on views
\citep{wadler:popl87:views}.  As hinted above, we interpret each
datatype~$\DD$ in the complexity language as a partial order $\den \DD{}$.  But we will also make use
of the sum type $D^\DD =
\den{\subst*{\mspf_{C_0}}\DD}{}+\dots+\den{\subst*{\mspf_{C_{n-1}}}\DD}{}$
(representing the unfolding of the datatype) and a function
$\semsize_\DD:D^\DD\to\den\DD{}$ (which represents the size of a
constructor, in terms of the size of the argument to the constructor).
When $\mspf_{\cC_i} = t$ (i.e. the argument to the constructor is a
single recursive occurrence of the datatype), 
$\semsize(\seminj_i\,x)$ is intended to represent an upper bound on the
size of the values of the form $\sC\,v$, where $v$ is a value of size at
most~$x$.  To define the semantics of $\crec[\DD] y C {\bind x {E_C}}$,
we consider all values~$z\in D^\DD$ such that $\semsize_\DD(z)\leq y$.  We
can distinguish between such values to (recursively) compute the
possible values of the form~$\subst{E_C}{\ldots}{x}$, and then take a
maximum over all such values.

For example, for the inductive definitions of $\snat$ and $\slist$
(where the list elements have type $\snat$), suppose we want to construe the
size of a $\slist$ to be the number of all $\snat$ and $\slist$
constructors.  We implement this in the complexity semantics as
\[
\begin{array}{rcl}
% \stype\:\snat & = & \sint \\
% \sdatatype\:\snat^* & = & \sconstrdecl{\sZero}{\sunit}\sconstrsep \sconstrdecl{\sSucc}{\sunit\cross\snat} \\
% \size_{\snat}(\sZero) & = & 1 \\
% \size_{\snat}(\sSucc(\striv, n)) & = & 1 + n
% \\ \\
% \stype\:\slist & = & \sint \\
% \sdatatype\:v\slist & = & \sconstrdecl{\snil}{\sunit}\sconstrsep \sconstrdecl{\scons}{\sunit\cross\slist} \\
% \size_{\slist}(\snil) & = & 1 \\
% \size_{\slist}(\scons(n, m)) & = & 1 + m + n
\den\cnat{} & = & \Z^+ \\
D^{\cnat} &=& \set{*} + \den{\cnat}{} \\
\semsize_{\cnat}(*) &=& 1 \\
\semsize_{\cnat}(m) &=& 1+m \\
%\snat^* & = & \sconstrdecl{\sZero}{\sunit}\sconstrsep \sconstrdecl{\sSucc}{\sunit\cross\den\snat{}} \\
%\size_{\snat}(\sZero) & = & 1 \\
%\size_{\snat}(\sSucc(\striv, n)) & = & 1 + n
\\
\den\clist{} & = & \Z^+ \\
D^{\clist} &=& \set{*} + (\den\cnat{}\cross\den\clist{}) \\
\semsize_{\clist}(*) &=& 1 \\
\semsize_{\clist}((m, n)) &=& 1 + m + n
%\slist^* & = & \sconstrdecl{\snil}{\sunit}\sconstrsep \sconstrdecl{\scons}{\den\snat{}\cross\den\slist{}} \\
%\size_{\slist}(\snil) & = & 1 \\
%\size_{\slist}(\scons(n, m)) & = & 1 + m + n
\end{array}
\]
where $\Z^+$ is the non-negative integers.\footnote{We refer to
$\Z^+$ rather than the natural numbers to emphasize that the
intepretation of $\D$ need not be an inductive datatype.}

We define the size-based complexity semantics as follows.  The
base cases for an inductive definition of $(S^T, \leq_T)$ for
every complexity type~$\mtp$ consist of 
well-founded partial orders $(S^\Delta,\leq_\Delta)$ for every
datatype $\Delta$ in the signature, such that $\leq_\Delta$ is closed
under arbitrary maximums (see below for a discussion).  We define
$\N^\infty = \N\union\set{\infty}$, where $\N$ is the natural numbers
with the usual order and addition.
We extend the order and addition to~$\infty$ by
$n\leq_{\N^\infty}\infty$ and
$n+\infty = \infty+n = \infty+\infty = \infty$ for all $n\in\N$. 
For products and functions we define
$S^{\cunit} = \set{*}$ and $S^{\mprod{T_0}{T_1}} = S^{T_0}\cross
S^{T_1}$ and $S^{\sarr{T_0}{T_1}} = (S^{T_1})^{S^{T_0}}$, with the
trivial, componentwise, and pointwise partial orders, respectively.
Complexity types are interpreted into this 
type structure by setting
$\den{\C}{} = \N^\infty$ and $\den{T}{} = S^T$ for each
complexity type~$\mtp$.

Stating the conditions on programmer-defined size functions requires
some auxiliary notions.  For ~$\cdatadecl\DD\cC{\mspf_\cC}$, set $D^\DD =
\den{\subst*{\mspf_{\cC_0}}{\DD}}{}+\dots+\den{\subst*{\mspf_{\cC_{n-1}}}\DD}{}$,
writing $\seminj_i:\den{\subst*{\mspf_{{\cC_i}}}\DD}{}\to D^\DD$ for the
$i^{th}$ injection.  
Next, we define a function $\semsz^\mspf$ with domain
$\den{\subst*\mspf\DD}{}$ (the semantic analogue of the
argument type of a datatype constructor).
% A value~$a$ of type $\den{\subst*\mspf\DD}{}$ 
% is built up from values of type~$\den\DD{}$
% by pairing with other values and applying functions; 
$\semsz^\mspf(a)$
is intended to be the maximum of the values of type~$\den{\DD}{}$ from which~$a$
is built using pairing and function application.  
We want to define $\semsz^\mspf$ by induction on $\mspf$, computing
the maximum at each step.  To ignore values not of type $\den{\DD}{}$
we assume an element $\bot\notin S^\DD$ that serves as an identity
for $\bmax$; that is,
we order $S^\DD \union \set{\bot}$ so that $\bot < a$ for all
$a \in S^\Delta$.  We define $\semsz^\mspf :
\den{\subst*{\mspf}{\DD}}{}\to S^\DD \union \set{\bot}$ by induction on
$\mspf$ as follows:
\[
\begin{array}{l}
 \semsz^t(a) = a \\
 \semsz^\mtp(a) = \bot \\
 \semsz^{\mprod{\mspf_0}{\mspf_1}}(a) = \semsz^{\mspf_0}(a)\bmax\semsz^{\mspf_1}(a) \\
 \semsz^{\mtp\to\mspf}(f) = \bigvee_{a\in\den\mtp{}}\semsz^\mspf(f(a)) 
\end{array}
\]
% \item $f$ is increasing with respect to $\mspf = t$ if
% $a < f(a)$ for all $a\in S^\D$;
% \item $f$ is increasing with respect to $\mspf = \mtp$ for
% any~$\mtp$;
% \item $f$ is increasing with respect to $\mspf = \mspf_0\cross\mspf_1$
% if $f$ is increasing with respect to each $\mspf_i$; and
% \item $f$ is increasing w.r.t.\ $\mspf = \mtp\to\Phi_0$ if
% \item $\semsize$ is increasing with respect to $\mspf = \sarr\mtp{\mspf_0}$
% if for all~$f\in S^{\den\mtp{}}$, $

The key input to the size-based semantics is programmer-supplied
size functions $\semsize_\DD : D^\Delta \to S^\Delta$ such that
\[
\semsz^{\mspf_{\cC_i}}(a) <_{S^\Delta \union \set{\bot}} ({\semsize_\DD}\comp{\seminj_i})(a)
\]
% In words, we require that $\semsize_\D(a)$ be strictly greater than 
% $\semsize_\D(b)$ for every substructure $b$ of type~$S^\D$ in~$a$.
$\semsize_\DD$ represents the programmer's notion of size for
inductively-defined values.  The only condition, which is used to
interpret the recursor, is that the size of a value is strictly greater
than the size of any of its substructures of the same type.
For example, this condition permits interpreting the size of a list as
its length or its total number of constructors, and the size of a tree as
its number of nodes or its height.  Non-examples include defining the size
of a list of natural numbers to be the number of successor constructors,
and defining the size of all natural numbers to be a constant (though
see Section~\ref{sec:inf-costs} for a discussion of this latter
possibility).

% Note that there is no partial order defined on $D^\delta$,
% at least not explicitly;
% the $\semsize_\D$ function can be used to define a pre-order on~$D^\delta$, and
% it is this pre-order that we will end up working with when we establish
% our formal results.
% \begin{enumerate}
% \item For each base type~$\tau$ in the source language (including
% datatypes), we posit a partial order~$\leq_\tau$ that is closed under
% arbitrary maximums (see below for a discussion).
% \item The complexity semantics contains a type $\N^\infty$ (for costs)
% with the intended interpretation being the natural numbers with an
% ``infinite'' value (see below for a formal definition).
% \item For each datatype declaration~$\sdatadecl\D\sC{\phi_\sC}$ there
% is a disjoint sum 
% $\D^* = \subst*{\phi_{C_0}}{\D}+\dots+\subst*{\phi_{C_{n-1}}}\D$ and
% a function $\size_\D:\D^*\to\D$.
% These disjoint sums are \emph{not} considered base types and there
% is no partial oder defined on them.\footnote{At least not explicitly;
% the $\size_\D$ function can be used to define a pre-order on~$\D^*$, and
% it is this pre-order that we will end up working with when we establish
% our formal results.}
% \item Base types are closed under finite products and arrows.
% The partial orders are extended to all
% product and arrow types component-wise and point-wise, with the maximums
% defined in the obvious ways.
% \end{enumerate}

\begin{figure*}
\begin{align*}
\semcase^\D &\oftype
D^\D\cross\prod_{\cC}(S^{\den{\subst*{\mspf_\cC}{\D}}{}}\arrow S^\tau) \arrow S^\tau \\
\semcase(\cC x, (\dots,f_\cC,\dots)) &= f_\cC(x) \\ \\
\den{\cC e}\xi &= \semsize(\cC(\den e\xi)) \\
% \den{\mmap\mspf{\bind {x^{\mtp_0}} {\mtm_1^{\mtp_1}}}{\mtm_0^{\subst*\mspf{\mtp_0}}}}{\xi} &=
% \semmapkw^{\mspf,\den{\mtp_0}{},\den{\mtp_1}{}}(\semlam a {\den{\mtm_1}{(\extend\xi a x)}}, \den{\mtm_0}\xi) \\
\den{\crecdecl[\D] {\mtm^\D} \cC {\bind {x^{\subst*{\phi_\cC}{\D\cross\tau}}} {\mtm_\cC^\tau}}}\xi & =
\bigvee_{\semsize z\leq \den{\mtm}\xi}\semcase(z, (\dots,f_\cC,\dots)) \\
\intertext{where}
f_\cC(x) &= \den{\mtm_\cC}{\extendenv \xi x {\den{\cmapkw^{\mspf_\cC}(\bind w{\spair{w}{\crecdecl w \cC {\bind x {\mtm_\cC}}}}}{\xi}, x)}} \\
&= \den{\mtm_\cC}{\extendenv\xi x {\semmapkw^{\mspf_\cC}(\semlam a {\sempair{a}{\den{\crecdecl w \cC {\bind x {\mtm_\cC}}}{\extendenv\xi w a}}}, x)}}
\end{align*}
\caption{The interpretation of $\creckw$ in 
the size-based semantics for the complexity language.}
\label{fig:size_sem}
\hrule
\end{figure*}

% \sloppypar Except for constructors, $\cmapkw$, and $\creckw$ expressions,
% the interpretation of terms
% is the standard interpretation.
% Constructor expressions are defined in terms of~$\semsize$:
% \[
% \den{\cC e}\xi = \semsize(\cC(\den e\xi)).
% \]
% $\cmapkw$ expressions are defined in terms of semantic functions
% $\semmapkw^{\Phi,\tau_0,\tau_1}\oftype S^{\sarr{\tau_0}{\tau_1}}\cross S^{\subst*\mspf{\tau_0}}\to S^{\subst*\mspf{\tau_1}}$ that are defined
% by induction on~$\mspf$ by
% \begin{align*}
% \semmapkw^{\reccall,\tau_0,\tau_1}(f, x) &= f(x) \\
% \semmapkw^{\tau,\tau_0,\tau_1}(f, x) &= x \\
% \semmapkw^{\mprod{\mspf_0}{\mspf_1},\tau_0,\tau_1}(f, x) &= \mpair{\semmapkw^{\mspf_0,\tau_0,\tau_1}(f, x)}{\semmapkw^{\mspf_1,\tau_0,\tau_1}(f, x)} \\
% \semmapkw^{\tau\to\mspf,\tau_0,\tau_1}(f, x) &=
% \semlam y {\semmapkw^{\mspf,\tau_0,\tau_1}(f, x\,y)}.
% \end{align*}
% We now take
% \begin{multline*}
% \den{\mmap\mspf{\bind {x^{\mtp_0}} {\mtm_1^{\mtp_1}}}{\mtm_0^{\subst*\mspf{\mtp_0}}}}{\xi} = \\
% \semmapkw^{\mspf,\den{\mtp_0}{},\den{\mtp_1}{}}(\semlam a {\den{\mtm_1}{(\extend\xi a x)}}, \den{\mtm_0}\xi)
% \end{multline*}
The interpretation of most terms is standard except for that of
constructors and $\creckw$, which are given in
Figure~\ref{fig:size_sem}.  We write $\semmapkw^{\mspf,\mtp_0,\mtp_1}$
for semantic functions that mirror the definition of $\cmapkw$, and 
we overload the notation~$\cC_i$ to stand for 
$\seminj_i:\den{\subst*{\mspf_{{\cC_i}}}\D}{}\to D^\D$.
The implementation of the recursors requires a bit of explanation, and is motivated by
the goal to have $\trans e$ bound the cost and potential of~$e$.  
We
expect that $\den{\trans{{\srec[\D] e \sC {\bind x {e_\sC}}}}}{}$, which
depends on $\den{\crec [\D]
{{\trans e}_p} \sC {\bind x {\trans {e_\sC}}}}{}$, should branch on $\den{\trans
e_p}{}$, evaluating to the appropriate~$\den{\trans{e_\sC}}{}$.  However, 
$\den{\trans e_p}{}$ will be a semantic value
of type $S^\D$, whereas to branch, we need a semantic value of type~$D^\D$.
Furthermore, $\den{\trans e_p}{}$ is only 
an \emph{upper bound} on the size of~$e$, 
so we cannot
use $\den{\trans e_p}{}$ to predict which branch the evaluation of the
source $\sreckw$ expression will follow.  We solve these problems by
introducing a semantic $\semcase$ function, and define the denotation of
$\creckw$ expressions by taking a maximum over the branches for all
semantic values that are bounded by the upper bound $\den{\trans e_p}{}$.  This
is the source of the requirement that base-type potentials be closed
under arbitrary maximums.
Although this requirement seems rather strong, in most
examples it seems easy to satisfy.
In particular, we think of most
datatype potentials (sizes) as being natural numbers, and so
we satisfy the condition by interpreting them by~$\N^\infty$.

The restriction on $\semsize_\DD$ ensures that the recursion used to
interpret $\creckw$ expressions descends along a well-founded partial
order, and hence is well-defined.  The maximum may end up being a
maximum over all possible values, but this simply indicates that our
interpretation fails to give us precise information.

We illustrate this semantics on some examples.  In order to ease the
notation, we will occasionally write syntactic expressions for the
corresponding semantic values (in effect, dropping $\den\cdot{}$).  We
also write the $\semcase$ function as a branch on constructors; for
example, we write $\semcase(t, \cemp\mapsto\bind x{\cpair 1
  1}\mid\cnode\mapsto\bind{\langle y,t_0,t_1\rangle}e)$ for $\semcase(t,
\lambda x.\cpair 1 1, \lambda\langle y,t_0,t_1\rangle. e)$.

\subsection{Booleans and Conditionals}
In the source language we define booleans and 
their $\scase$ construct:
\begin{align*}
\sdatatype\:\sbool &= \sconstrdecl \strue\sunit\mid\sconstrdecl\sfalse\sunit\\
\scase(e^{\sbool}, e_0^\tau, e_1^\tau) &= \srecm e {\strue\mapsto {e_0}\mid\sfalse\mapsto {e_1}}
\end{align*}
(recall our convention on writing $e_\cC$ for $\bind x {e_\cC}$ when
$\sspf_\cC=\sunit$).
In the semantics of the complexity language, we
interpret $\cbool$ as a one-element set~$\set{1}$, so $\strue$ and
$\sfalse$ are indistinguishable by ``size.''
Our interpretation yields
\begin{align*}
\den{\trans{\scase(e, e_0, e_1)}}{}
  &= {\trans{e}_c}\plusc
     {\crecdeclm {\trans e_p} {\ctrue\mapsto{{\costpluscpy 1 {\trans {e_0}}}}\mid\cfalse\mapsto{{\costplusone{\trans {e_1}}}}}} \\
  &= \trans{e}_c \plusc
  % &\qquad\bigvee_{\semsize b\leq\trans e_p}\ccasedeclm b {\ctrue\mapsto {1 \plusc \trans{e_0}}\mid\cfalse\mapsto {1 \plusc \trans{e_1}}} \\
     \bigvee_{\semsize b\leq\trans e_p}\semcase(b, {\ctrue\mapsto {1 \plusc \trans{e_0}}\mid\cfalse\mapsto {1 \plusc \trans{e_1}}}) \\
  &= \trans{e}_c \plusc (\semcase(\ctrue, {\ctrue\mapsto {1 \plusc \trans{e_0}}\mid\cfalse\mapsto{1 \plusc \trans{e_1}}}) \\
  &\qquad \vee \semcase(\cfalse, {\ctrue\mapsto {1 \plusc \trans{e_0}}\mid\cfalse\mapsto{1 \plusc \trans{e_1}}})) \\
  &= (1+\trans e_c) \plusc (\trans{e_0}\bmax\trans{e_1}).
\end{align*}
In other words, if we cannot distinguish between $\strue$ and $\sfalse$
by size, then
the interpretation of a conditional is just the maximum of its branches
(with the additional cost of evaluating the test).
This is precisely the interpretation used by
\citet{danner-et-al:plpv13}.

\subsection{Tree Membership}
\label{ex:tree-mem}
Next we consider an example that shows that the ``big'' maximum used
to interpret the recursor can typically be simplified to the recurrence
that one expects to see.
We analyze the cost of checking membership in an $\sint$-labeled tree.  We
write $e_0 \srcconstr{orelse} e_1$ as an abbreviation for 
$\scase{(e_0,\strue \mapsto \strue \mid \sfalse \mapsto e_1)}$.

%% \[
%% \begin{array}{rcl}
%% \sdatatype\:\stree &=& \sconstrdecl \semp \sunit\mid  \\
%% & & \sconstrdecl \snode {\sint\cross\stree\cross\stree} \\
%% \smem(t, x) &=& \sreckw(t, \\
%%   & &\qquad \semp\mapsto \sfalse \\
%% %  & &\qquad \snode\mapsto \bind {\stuple{y, \spair{t_0}{r_0}, \spair{t_1}{r_1}}} {\scase(y=x,} \\
%%   & &\qquad \snode\mapsto \bind {\stuple{y, \spair{t_0}{r_0}, \spair{t_1}{r_1}}}{} \\
%%   & &\qquad\quad {\scase(y=x,} \\
%%   & &\qquad\qquad \strue\mapsto\strue \\
%%   & &\qquad\qquad \sfalse\mapsto\scase(\sforcekw\,r_0,\\
%%   & &\qquad\qquad\qquad \strue\mapsto\strue \\
%%   & &\qquad\qquad\qquad \sfalse\mapsto \sforcekw\,r_1)))
%% \end{array}
%% \]

\[
\begin{array}{l}
\sdatatype\:\stree = \sconstrdecl \semp \sunit \mid \sconstrdecl \snode {\sint\cross\stree\cross\stree}  \\
\smem(t, x) = \sreckw(t, \\
  \quad \semp\mapsto \sfalse \\
  \quad  \snode\mapsto \bind {\stuple{y, \spair{t_0}{r_0}, \spair{t_1}{r_1}}}{} \\
  \qquad y=x \: \srcconstr{orelse} \: (\sforcekw\,r_0 \: \srcconstr{orelse} \:  \sforcekw\,r_1))
\end{array}
\]

For this example, we treat $\sint$ (in the source and complexity
languages) as a datatype with $2^{32}$ constructors where the equality test
$x = y$ is implemented by a rather large case analysis.  Let us
define the size of a tree to be the number of nodes:
\[
\begin{array}{rcl}
\den\ctree{} &=& \N^\infty \\
D^{\ctree} &=& \set{*} + \set{1}\cross\N^\infty\cross\N^\infty \\
\semsize_{\ctree}(\cemp) &=& 0 \\
\semsize_{\ctree}(\cnode(1, n_0, n_1)) &=& 1 + n_0 + n_1
\end{array}
\]
We would like to get the following recurrence for the
cost of the $\sreckw$ expression when~$t$ has size~$n$:
% (the
% cost of $\smem(t, x)$ will be one greater to account for the
% initial application):
\[
T(0) = 1
\qquad
T(n) = \bigmax_{n_0+n_1+1=n} 6 + T(n_0) + T(n_1)
\]
($x=y$ requires an application and two $\scase$ evaluations; 
each $\srckeyw{orelse}$ evaluation costs $1$; and we charge for the
$\sreckw$ reduction).

Working through the interpretation yields
$\den{\trans{\smem(t, x)}}{}_c = \trans t_c + g(\trans t_p) + 1$
% \[
% \begin{split}
% \den{\trans{\smem(t, x)}}{}_c = \trans t_c + g(\trans t_p) + 1
% \sreckw(&\trans t_p \\
%   &E\mapsto\bind z {(1, 1)} \\
%   &N\mapsto\bind {\ctuple{y, \cpair{t_0}{r_0}, \cpair{t_1}{r_1}}} {\costpluscpy{(3 + (r_0)_c)}{r_1}}
% \end{split}
% \]
where
\[
g(n) 
= \llbracket\creckw(z, \cemp\mapsto 1
\cnode\mapsto\bind{\ctuple{y, \cpair{t_0}{r_0}, \cpair{t_1}{r_1}}}{6 + (r_0)_c + (r_1)_c}
\rrbracket\unitenv z n.
\]
We can calculate that $g(0) = 1$, and for $n>0$:
\begin{align*}
g(n)
&= \bigmax_{\size t\leq n}\semcase(t, \\
&\qquad\qquad \cemp\mapsto 1 \\
&\qquad\qquad \cnode\mapsto \bind{\ctuple{y, n_0, n_1}}{6+g(n_0)+g(n_1)} \\
&= g(n-1)\bmax \bigmax_{\size t = n}\semcase(t,\dots) \\
&= g(n-1)\bmax \bigmax_{1+n_0+n_1 = n}\semcase(\cnode(1, n_0, n_1),\dots) \\
&= g(n-1)\bmax \bigmax_{1+n_0+n_1 = n}(6 + g(n_0) + g(n_1))
\end{align*}
% We now show that the cost component of $g(n)$ satisfies the recurrence
% given above.
% For $n>0$
% \begin{align*}
% g(n)
%   &= \bigmax_{\semsize t\leq n}\semcase(t,  \\
%   &\qquad\qquad \cemp\mapsto \bind x {(1, 1)}  \\
%   &\qquad\qquad \cnode\mapsto\bind{(y, t_0, t_1)}{} \\
%   &\qquad\qquad\qquad\sletkw(\cmapkw(\bind w {\spair{w}{\crecdeclm{w}{\dots}}}, (y, t_0, t_1)), \\
%   &\qquad\qquad\qquad\qquad \bind{\ctuple{\blank, \ctuple{\blank, r_0}, \ctuple{\blank, r_1}}}{\costpluscpy{(3+(r_0)_c)}{r_1}})) \\
%   &= g(n-1)\bmax \bigmax_{\semsize t=n}\semcase(t,\dots) \\
% \end{align*}
% %% Pagination hack
% \begin{align*}
%   &= g(n-1)\bmax \bigmax_{1+n_0+n_1=n}\semcase(\cnode(1, n_0, n_1),\dots) \\
%   &= g(n-1)\bmax{} \\
%   &\qquad\bigmax_{1+n_0+n_1=n}\sletkw(\cmapkw(\bind w {\spair{w}{\crecdeclm{w}{\dots}}}, (1, n_0, n_1)),\\
%   &\qquad\qquad\qquad\qquad\bind{\ctuple{\blank, \ctuple{\blank, r_0}, \ctuple{\blank, r_1}}}{\costpluscpy{(3+(r_0)_c)}{r_1}}) \\
%   &= g(n-1)\bmax \bigmax_{1+n_0+n_1=n}(\costpluscpy{(3+(g(n_0))_c)}{g(n_1)})
% \end{align*}
We now notice that when we take $n_0 = 0$ and $n_1 = n-1$ we have
% \begin{multline*}
% \costpluscpy{(3+(g(n_0))_c)}{g(n_1)} = \costpluscpy{(3+(g(0))_c)}{g(n-1)} = \\
% \costpluscpy 4 g(n-1)\geq g(n-1)
% \end{multline*}
\[
6 + g(n_0) + g(n_1) = 6 + g(0) + g(n-1) \geq g(n-1)
\]
and hence
\begin{align*}
g(n)
  % &= g(n-1)\bmax \bigmax_{1+n_0+n_1=n}(\costpluscpy{(3+(g(n_0))_c)}{g(n_1)}) \\
  % &= \bigmax_{1+n_0+n_1=n}(\costpluscpy{(3+(g(n_0))_c)}{g(n_1)}
  &= g(n-1)\bmax\bigmax_{1+n_0+n_1}(6+g(n_0)+g(n_1)) \\
  &= \bigmax_{1+n_0+n_1}(6+g(n_0)+g(n_1))
\end{align*}
which is precisely the recurrence we would expect.

\subsection{Tree Map}

Next, we consider an example that illustrates reasoning about higher-order
functions and the benefits of choosing an appropriate notion of size.
We analyze the cost of the map function for $\snat$-labeled binary trees:
\begin{align*}
\streemap(f, t) &= \sreckw(t, \\
&\qquad \semp\mapsto\semp \\
&\qquad \snode\mapsto\bind{\stuple{y,\spair{t_0}{r_0}, \spair{t_1}{r_1}}}{} \\
&\qquad \qquad{\snode(f(y), \sforcekw r_0,\sforcekw r_1)}.
\end{align*}
Suppose the cost of evaluating~$f$
is monotone with respect to the size of its argument, where 
we define the size of a natural number $n$ to be $1+n$ (to count the
zero constructor).  The cost of evaluating
$\streemap(f, t)$ should be bounded by $1+n\cdot (1+f(s)_c)$, 
where $n$ is the number of nodes in~$t$,
$s$ is the maximum size of all labels in~$t$,
and we write $f(s)_c$ for the cost of evaluating
$f$ on a natural number of size $s$ (the map runs $f$ on an input
of size at most $s$ for each of the $n$ nodes, and takes an additional $n$
steps to traverse the tree).  

We take $\den\ctree{} = \N^\infty\cross\N^\infty$, where we think of the
pair $(n, s)$ as (number of nodes, maximum size of label), and
use the mutual ordering on pairs ($(n, s) < (n', s')$ iff $n\leq n'$ and
$s< s'$ or $n<n'$ and $s\leq s'$).  The size function is defined as follows:
\begin{align*}
\semsize(\cemp) &= (0, 0) \\
\semsize(\cnode(n, (n_0, s_0), (n_1,s_1))) &= (1+n_0+n_1, \max\set{n,s_0,s_1}).
\end{align*}
Let us write $g(m, s) = \den{\trans{\sreckw(\dots)}}{\unitenv t{(m, s)}}$, 
so that $(\den{\trans{\streemap}}{}(f, (m, s)))_c = g(m, s) + 1$.
We now show that
$g(m, s) \leq m(1+ f(s)_c)$ by induction:
%\begin{align*}
\[
\begin{split}
g&(m, s) \\
&=\bigmax_{\semsize z\leq(m, s)}\semcase(z, \\
&\qquad\qquad \cemp\mapsto 1{} \\
&\qquad\qquad \cnode\mapsto\bind{\langle{n,(n_0,s_0),(n_1,s_0)}\rangle}{} \\
&\qquad\qquad\qquad\bigl(1+(f(n))_c + (g(n_0, s_0))_c + (g(n_1, s_1))_c\bigr) \\
&= 1 \bmax {}
   \bigmax_{\substack{1+n_0+n_1\leq m\\ \max\set{n,s_0,s_1}\leq s}}
   \bigl(1+f(n)_c + (g(n_0, s_0))_c + (g(n_1, s_1))_c\bigr) \\
&\leq \bigmax_{\substack{1+n_0+n_1\leq m\\ \max\set{n,s_0,s_1}\leq s}}
      (1 + f(n)_c + n_0\cdot (1+f(s_0)_c) + n_1\cdot (1+f(s_1)_c)) \\
&\leq \bigmax_{\substack{1+n_0+n_1\leq m\\ \max\set{n,s_0,s_1}\leq s}}
(1+n_0+n_1)(1+f(\max\set{n,s_0,s_1})_c) \\
&\leq m\cdot (1+f(s)_c).
\end{split}
\]
%\end{align*}

\subsection{The Bounding Theorem for the Size-Based Semantics}
\label{sec:bounding-theorem}

The most basic correctness criterion for our technique is that a closed
source program's operational cost is bounded by the cost component of
the denotation of its complexity translation.  However, to know that
extracted \emph{recurrences} are correct, it is not enough to consider
closed programs; we also need to know that the potential of a function
bounds that function's operational cost on all arguments, and so on at
higher type.  Thus, we use a logical relation. We first show
a simplified case of the logical relation, where for this
  subsection only we do not allow datatype constructors to take
functions as arguments (i.e., drop the $\tau\to\phi$ clause from constructor
argument types~$\phi$).  In Section~\ref{sec:monotonic}, we consider the
general case, which requires some non-trivial technical additions
to the main definition.

\begin{defn}[Bounding relation]
\label{def:bounding-reln-prelim}~
\begin{enumerate}
\item 
Let $e$ be a closed source language expression and $a$ a 
semantic value.
We write $e\bounded_\tau a$ to mean:
if $\evalin{e}{v}{n}$, then
\begin{enumerate}
\item $n\leq a_c$; and
\item $v \vbounded_\tau a_p$.
\end{enumerate}

\item Let $v$ be a source language value and $a$ a semantic value.
We define $v\vbounded_\tau a$ by:
\begin{enumerate}
\item $()\vbounded_{|unit|} 1$.
\item $\spair{v_0}{v_1}\vbounded_{\tau_0\cross\tau_1}\cpair{a_0}{a_1}$ if
$v_i\vbounded_{\tau_i}a_i$ for $i=0,1$.
\item $\sdelaykw(e)\vbounded_{\ssusp\tau} a$ if $e\bounded_\tau a$.
\item \label{item:vbounded_D}
$\sC(v)\vbounded_\D a$ if there is $a'$ such that
$v\vbounded_{\subst*{\phi_\sC}{\D}} a'$ and
$\semsize({\sC(a')})\leq {a}$.\footnote{
Our restriction on the form of~$\sspf_\sC$ allows us to conclude that
this definition is well-founded, even though the type gets bigger in clause~(\ref{item:vbounded_D}), because we can treat the definition of $\vbounded_\D$
as an inner induction on the values.  
Allowing
datatype constructors to take function arguments complicates the situation,
and in Section~\ref{sec:monotonic} we must define a more general
relation.}
\item $\lambda x.e\vbounded_{\sigma\arrow\tau} a$ if whenever
$v\vbounded_\sigma a'$, $\subst e v x\bounded_\tau a(a')$.
\end{enumerate}
% Our restriction on the form of~$\sspf_\sC$ allows us to conclude that
% this definition is well-founded, even though clause~(\ref{item:vbounded_D})
% is not a structural induction on type.
% % We note that clause~(\ref{item:vbounded_D}) 
% % of the $\vbounded$ relation is not a structural
% % induction on type, as $\subst*{\phi_\sC}{\D}$ is not a structural subtype
% % of $\D$.  However, our restriction on the form of~$\phi_\sC$ allows us
% % to conclude that $\phi_\sC$ is a product of types~$\tau$ such that
% % $t\notin\fv(\tau)$ and $t$'s, and from there we can
% % define $v\vbounded_\D E$ by, for example, induction on the
% % maximum nesting depth of constructors in~$v$.
% Allowing
% datatype constructors to take function arguments complicates the situation,
% and in Section~\ref{sec:monotonic} we must define a more general
% relation.

% \item For a source substition $\ssub\oftype\sctx$ and a complexity
% substitution $\msub\oftype\Gamma$, we write
% $\subbound\ssub\msub\sctx$ to mean that
% for all $(x\oftype\tau)\in\sctx$, $\theta(x)\vbounded_\tau\Theta(x)$.

% \item For $\typejudge\sctx\stm\tau$ and
% $\typejudge\mctx\mtm{\trans\tau}$, we write
% $\expbound\stm\mtm\tau$ to mean that for all
% $\ssub\oftype\sctx$ and $\msub\oftype\mctx$, if
% $\subbound\ssub\msub\sctx$, then
% $\expbound{\cl\stm\ssub}{\den{\cl\mtm\msub}{}}{\tau}$.
\end{enumerate}
\end{defn}

\begin{thm}[Bounding theorem]
\label{thm:bounding}
If $\stm\oftype \stp$ in the source language, then $\expbound\stm{\den{\trans e}{}}{\stp}$.
\end{thm}

Rather than proving this bounding theorem directly, in
Section~\ref{sec:monotonic} we identify syntactic constraints on the
complexity language which allow the proof to be carried through
(Theorem~\ref{thm:mon-bounding}).  Because the size-based semantics
satisfies these syntactic constraints (see
Section~\ref{sec:size-based-is-mon}), we can prove that the logical
relation defined in Section~\ref{sec:monotonic} implies the one defined
above, giving Theorem~\ref{thm:bounding} as a corollary.

\section{The Syntactic Bounding Theorem}
\label{sec:monotonic}

Rather than proving the bounding theorem for a particular model,
such as the one from the previous section, we use a syntactic judgement
$\lejudgeMP\mtp{\mtm_0}{\mtm_1}$ to axiomatize the properties that are necessary
to prove the theorem.  The rules are in
Figure~\ref{fig:monotonic_type_theory}; we omit typing
premises from the figure, but formally each rule has sufficient premises to make the two
terms have the indicated type.  The first two rules state reflexivity
and transitivity.  The next rule (congruence) says that term contexts 
of a certain form 
(in the sequel, \emph{congruence contexts}) 
are monotonic.  
The next three rules state the monoid laws for $\C$; we write $E_0 = E_1$ to
abbreviate two rules $\homj{}{E_0}{E_1}$ and $\homj{}{E_1}{E_0}$.  The
final three rules (which we call ``step rules'') say that a
$\beta$-redex is bigger than or equal to its reduct.
The first five congruence contexts are the standard
head elimination contexts used in logical relations arguments (principal
arguments of elimination forms) and the next two say that $\costplus$ is
monotone.
%% ---i.e. $\C$ should be a monoid in the category of preorders,
%% not just in sets, so the binary operation should be a monotone map.  The

These preorder rules are sufficient to prove the bounding theorem, and
permit a variety of interpretations and extensions.  If we
impose no further rules, then $\homj{}{E_0}{E_1}$ is basically weak head
reduction from $E_1$ to $E_0$ (plus the monoid laws for $\C$).
We can also add rules that identify elements of datatypes, in
order to make those elements behave like sizes.  For example, for
lists of $\cint$s, we can say
\[
\ndAXC{\homj{}{E}{\ccons (\_,E)}}
\DisplayProof
\qquad
\ndAXC{\homj{}{\ccons(E_1, E)}{\ccons(E_2 , E)}}
\DisplayProof
\]
and extend the congruence contexts with $\ccons(x,\congctx)$.  Then the
second rule equates any two lists with the same number of elements,
quotienting them to natural numbers, and the first rule orders these
natural numbers by the usual less-than.  Thus, considered up to
$\lejudgekw$, lists are lengths.

Combining these rules with the ones used to prove the bounding theorem,
the recursor for lists behaves like a monotonization of the original
recursion (like the $\bigvee$ in the size-based complexity semantics).
For example, for any specific list value $\ccons(x,xs)$, by the usual
step rule, we have
\[
{\subst{E_1} {(x,xs,{\clistrec{xs}{E_0}{p}{E_1}})} {p}} \le
{\clistrec{\ccons(x,xs)}{E_0}{p}{E_1}}
\]
But we can derive $\homj{}{\cnil}{\ccons(x,xs)}$, so we also have
\[
\begin{array}{ll}
\homj{}{\creckw(\cnil,\ldots)}{\creckw(\ccons(x,xs),\ldots)}
& \text{by congruence}\\  %% with context ${\clistrec{\hole}{E_0}{x}{E_1}}$

\homj{}{E_0}{\clistrec{\cnil}{E_0}{p}{E_1}} 
& \text{by the step rule} \\

\homj{}{E_0}{\clistrec{\ccons(x,xs)}{E_0}{p}{E_1}} &
\text{by transitivity}
\end{array}
\]
and similarly for non-empty lists that are $\le \ccons(x,xs)$.  Thus, when we quotient lists
to their lengths, the congruence and step rules for $\creckw$ (used to
prove the bounding theorem) imply that the recursor is bigger than all
of the branches for all smaller lists.  This is in contrast to
the interpretation of the recursor-like construct given by
\citet{danner-et-al:plpv13}, which includes a explicit maximization
that includes the base case.

In Section~\ref{sec:size_semantics}, we used reasoning in the size-based
semantics to massage the recurrence extracted from a program into a
recognizable and solvable form.  In future work, we plan to investigate
how to do this massaging within the syntax of complexity language, using
the rules we have just discussed and others.  For example, while a
recursion bounds what it steps to on all smaller values, we do not yet
have a rule stating that it is a least upper bound.  Here, we lay a
foundation for this by proving the bounding theorem for the small set of
rules in Figure~\ref{fig:monotonic_type_theory}.

\begin{figure*}
\begin{align*}
\congctx ::= &\hole \mid \fst\congctx \mid \snd\congctx \mid  
\app\congctx\mtm \mid
\crec[]{\congctx}{\cC}{\bind x {\mtm_{\cC}}}
%% \mid \cpair{\congctx}{\mtm}\mid\cpair{\mtm}{\congctx} Don't think we need these?
\mid \congctx+\mtm\mid \mtm+\congctx 
\end{align*}
\begin{gather*}
\RightLabel{(reflexivity)}
\ndAXC{$\lejudgeMP\mtp\mtm\mtm$}
\DisplayProof
\qquad
  \AXC{$\lejudgeMP\mtp{\mtm_0}{\mtm_1}$}
  \AXC{$\lejudgeMP\mtp{\mtm_1}{\mtm_2}$}
\RightLabel{(transitivity)}
\ndBIC{$\lejudgeMP\mtp{\mtm_0}{\mtm_2}$}
\DisplayProof
\\
\AXC{$\typejudge[]{\mctx,x\oftype\mtp'}{\congctx[x]}{\mtp}$}
\AXC{$\lejudgeMP {\mtp'}{\mtm_0}{\mtm_1}$}
\RightLabel{(congruence)}
\ndBIC{$\lejudgeMP{\mtp}{\congctx[\mtm_0]}{\congctx[\mtm_1]}$}
\DisplayProof
\\
\ndAXC{$\Gamma \vdash \costzero \costplus \mtm =_{\C} \mtm$} 
\DisplayProof
\qquad
\ndAXC{$\Gamma \vdash \mtm \costplus \costzero =_{\C} \mtm$} 
\DisplayProof
\qquad
\ndAXC{$\Gamma \vdash (\mtm_0\costplus\mtm_1) \costplus \mtm_2 =_{\C} \mtm_0 \costplus(\mtm_1\costplus\mtm_2)$} 
\DisplayProof
\\
\ndAXC{$\lejudgeMP\mtp{\subst{\mtm_0}{\mtm_1}{x}}{(\lambda x.\mtm_0)\mtm_1}$}
\DisplayProof
\qquad
\ndAXC{$\lejudgeMP{\mtp_i}{\mtm_i}{\cproj_i\cpair{\mtm_0}{\mtm_1}}$}
\DisplayProof
\\
  \AXC{$\cC\oftype (\sarr\mspf\DD)\in\msig$}
\ndUIC{$\lejudgeMP\mtp{\subst{\mtm_\cC}{\cmap\mspf{\bind y{\cpair{y}{\crecdecl{y}{\cC}{\bind x {\mtm_{\cC}}}}}}{\mtm_0}}{x}}{\crecdecl[\DD]{\cC\mtm_0}{\cC}{\bind x {\mtm_\cC}}}$}
\DisplayProof
\end{gather*}
\caption{Congruence contexts and the preorder judgement}
\label{fig:monotonic_type_theory}
\hrule
\end{figure*}

\subsection{The Bounding Relation}

First, we extend Definition~\ref{def:bounding-reln-prelim} to arbitrary datatypes.
Fix a signature~$\ssig$. We will mutually define the following
relations in definition~\ref{def:bounding-reln}:

\begin{enumerate}

\item \expbound{\stm}{\mtm}{\stp}, where \oftps{\emptyset}{\ssig}{\stm}{\stp} and
\oftps{\emptyset}{\comptr{\ssig}}{\mtm}{\comptr{\stp}}.

\item \valbound{\sv}{\mtm}{\tau}, where \oftps{\emptyset}{\ssig}{\sv}{\stp} and
\oftps{\emptyset}{\comptr{\ssig}}{\mtm}{\pottr{\stp}}.

\item \label{item:spfvalbound}
\spfvalbound{\sv}{\mtm}{R}{\sspf}, where $\oftps{\emptyset}{\ssig}{\sv}{\sarg{\sspf}{\D}}$
  and $\oftps{\emptyset}{\comptr{\ssig}}{\mtm}{\marg {\pottr{\sspf}} {\D}}$.  

\item \label{item:spfexpbound}
\spfexpbound{\stm}{\mtm}{R}{\sspf}, where
$\oftps{\emptyset}{\ssig}{\stm}{ \sarg{\sspf}{\D}}$ and
$\oftps{\emptyset}{\comptr \ssig}{\mtm}{ \marg {\comptr{\sspf}} {\D}}$
\end{enumerate}
In (\ref{item:spfvalbound}) and (\ref{item:spfexpbound}),
$R(\oftps{\emptyset}{\ssig}{\sv}{\D},\oftps{\emptyset}{\comptr{\ssig}}{\mtm}{\D})$,
is any relation; these parts interpret strictly positive functors as
relation transformers.

The definition is by induction on $\stp$ and $\sspf$.  For datatypes,
the signature well-formedness relation~\swfsig{\psi} ensures that
datatypes are ordered, where later ones can refer to earlier ones, but
not vice versa.  Therefore, we could ``inline'' all datatype
declarations: rather than naming datatypes, we could replace each
datatype name \D\/ by an inductive type
$\mu[\overline{\sconstrdecl{\DC}{\phi}}]$.  The logical relation is defined using the
subterm ordering for this ``inlined'' syntax.  In addition to the usual
subterm ordering for types $\tau$ and functors $\phi$, we have that
datatypes that occur earlier in $\psi$ are smaller than later ones, and
if $\DC:(\phi\to\D)\in\psi$, then $\phi$ is smaller than~$\D$.

\begin{defi}[Bounding relation]~
\label{def:bounding-reln}
\begin{enumerate}
\item 
We write $\expbound\stm\mtm\stp$ to mean:  if $\evalin e v n$, then
\begin{itemize}
\item $n\leq \cst\mtm$; and
\item $v\vbounded_\tau\pot\mtm$.
\end{itemize}

\item  We write $\valbound v E \tau$ to mean:
\begin{itemize}
\item $\valbound v E {\sunit}$ is always true.  
\item \valbound{\spair{\sv_1}{\sv_2}}{\mtm}{\sprd {\stp_1}{\stp_2}} iff 
  \valbound{\sv_1}{\fst \mtm}{\stp_1} 
  and \valbound{\sv_2}{\snd \mtm}{\stp_2}.

\item \valbound{\sdelay{\stm}}{\mtm}{\ssusp \stp} iff 
  \expbound{\stm}{\mtm}{\stp}.

\item \valbound{\sv}{\mtm}{\D} is inductively defined by
\[
\infer{\valbound{\sdcon{\DC}{\sv}}{\mtm}{\D}}
      {\sdconbind{\DC}{\sspf}{\D} \in \ssig &
       \spfvalbound{\sv}{\mtm'}{\valbound{-}{-}{\D}}{\sspf} &
       \homj{\D}{\dcon{\DC}{\mtm'}}{\mtm}}
\]
  
\item \valbound{\slam{x}{\stm}}{\mtm}{\sarr {\stp_1}{\stp_2}} iff 
  (for all $\sv_1$ and $\mtm_1$, if
   \valbound{\sv_1}{\mtm_1}{\stp_1} then 
   \expbound{\subst{\stm}{\sv_1}{x}}{\app {\mtm}{\mtm_1}}{\stp_2}).  

\end{itemize}

\item We write
$\spfvalbound\sv{\pot\mtm} R\sspf$ to mean:
% $\spfexpbound\stm\mtm R\sspf$ to mean:
% \item 
% Define (there's some duplication here; not sure if we can avoid it)
%  \spfexpbound{\stm}{\mtm}{R}{\sspf} iff 
%   (if \seval{\stm}{\sc}{\sv}
%    then \homj{\C}{\sc}{\cst{\mtm}} and
%    \spfvalbound{\sv}{\pot{\mtm}}{R}{\sspf}).
%    
% \item 
% Define 
% 
\begin{itemize}
\item \spfvalbound{\sv}{\mtm}{R}{\reccall} if $R(\sv,\mtm)$.

\item \spfvalbound{\sv}{\mtm}{R}{\stp} if \valbound{\sv}{\mtm}{\stp}
($t$ not free in~$\tau$).

\item \spfvalbound{\spair{\sv}{\sv'}}{\mtm}{R}{\sprd{\sspf}{\sspf'}}
  if \spfvalbound{\sv}{\fst \mtm}{R}{\sspf}
  and \spfvalbound{\sv'}{\snd \mtm}{R}{\sspf'}.

\item
  \spfvalbound{\slam{x}{\stm_1}}{\mtm_1}{R}{\sarr{\tau}{\sspf}} if
  for all $\sv$ and $\mtm$, if \valbound{\sv}{\mtm}{\tau}, 
  then \spfexpbound{\subst{\stm_1}{\sv}{x}}{(\app{\mtm_1}{\mtm})}{R}{\sspf}.
\end{itemize}

% Define \expbound{\stm}{\mtm}{\stp} by: if \seval{\stm}{\sc}{\sv} then \homj{\C}{\sc}{\cst{\mtm}} and
% \valbound{\sv}{\pot{\mtm}}{\stp}.
\item 
We write $\spfexpbound\stm\mtm{R}\sspf$ to mean:  if
$\evalin e v n$, then 
\begin{itemize}
\item $n\leq\cst E$; and
\item $\spfvalbound\sv{\pot\mtm} R\sspf$.
\end{itemize}
\end{enumerate}
The inner inductive definition of $\valbound{v}{E}{\D}$
makes sense because $R$ occurs strictly
positively in \spfvalbound{-}{-}{R}{\sspf}, and because (by
signature formation) $\D$ cannot occur in $\sspf$, so
{\valbound{-}{-}{\D}} does not occur elsewhere in
\spfvalbound{-}{-}{R}{\sspf}.
The relation on open terms considers all
closed instances:

\begin{enumerate} \setcounter{enumi}{4}
\item For a source substitution $\ssub\oftype\sctx$ and
complexity substitution $\msub\oftype\mctx$, we write
$\subbound\ssub\msub\sctx$ to mean
that for all $(x\oftype\tau)\in\sctx$,
$\valbound{\ssub(x)}{\msub(x)}{\tau}$.

\item For $\typejudge\sctx\stm\tau$ and
$\typejudge\mctx\mtm{\trans\tau}$, we write
$\expbound\stm\mtm\tau$ to mean that for all
$\ssub\oftype\sctx$ and $\msub\oftype\mctx$, if
$\subbound\ssub\msub\sctx$, then
$\expbound{\cl\stm\ssub}{\cl\mtm\msub}{\tau}$.
\end{enumerate}
\end{defi}

We write $\isder{\mathcal E}{\mathcal{J}}$ to mean that $\mathcal
E$ is a derivation of any of the judgements just described.  Because the
relation for function types is a function between relations, derivations
are infinitely-branching trees.  A \emph{subderivation} of such
an~$\mathcal E$ is any subtree of $\mathcal E$, which includes any
application of an $\to$-type judgement.  For example, if $\isder{\mathcal
  E_1}{\spfvalbound{\lambda x.e_1}{\mtm_1}{R}{\sarr\stp\sspf}}$ and
$\isder{\mathcal E}{\valbound v E\stp}$, then the derivation of
$\spfvalbound{\subst{e_1}{v}{x}}{\mtm_1\,\mtm}{R}{\sspf}$ is a
subderivation of~$\mathcal E_1$.

Next, we establish some basic properties of the relation:

\begin{lemma}[Weakening] \label{lem:weakening}~
\begin{enumerate}
\item \label{item:exp-weak}
If \expbound{\stm}{\mtm}{\stp}
and \homj{\comptr{\stp}}{\mtm}{\mtm'} then 
\expbound{\stm}{\mtm'}{\stp}.

\item \label{item:val-weak}
If \valbound{\sv}{\mtm}{\stp}
and \homj{\pottr{\stp}}{\mtm}{\mtm'} then 
\valbound{\sv}{\mtm'}{\stp}.
\end{enumerate}
\end{lemma}
\begin{proof}
We prove both clauses simultaneously by induction on~$\stp$, using
congruence for $\fst{\hole}$, $\snd{\hole}$ and $\app{\hole}{\mtm}$.

\begin{enumerate}

\item 

Suppose \expbound{\stm}{\mtm}{\stp} and
and \homj{\prd{\C}{\pottr{\stp}}}{\mtm}{\mtm'}. 
We need to show \expbound{\stm}{\mtm'}{\stp}, so assume 
\seval{\stm}{\sc}{\sv}.  Because \expbound{\stm}{\mtm}{\stp} we have that
\homj{\C}{\sc}{\cst{\mtm}}
and
\valbound{\sv}{\pot{\mtm}}{\stp} so it suffices to show
\homj{\C}{\cst{\mtm}}{\cst{\mtm'}}
and 
\homj{\pottr{\stp}}{\pot{\mtm}}{\pot{\mtm'}}.  
Recalling that \cst{-} and \pot{-} are really just \fst{-} and \snd{-},
these are true using the congruence rule with $x.\fst{x}$ and
$x.\snd{x}$ on \homj{\prd{\C}{\pottr{\stp}}}{\mtm}{\mtm'}.

\item 

\proofcase{$\sprd{\stp_1}{\stp_2}$}
By the induction hypotheses, it
suffices to show that the assumption 
\homj{\prd{\pottr{\stp_1}}{\pottr{\stp_2}}}{\mtm}{\mtm'}
implies 
\homj{\stp_1}{\fst \mtm}{\fst {\mtm'}}
and similarly for \snd{}.  Apply the congruence rule with $x.\fst{x}$.

\proofcase{$\ssusp{\stp}$}
Immediate by the induction hypothesis~(\ref{item:exp-weak}).

\proofcase{$\sarr{\stp_1}{\stp_2}$}
Using the induction hypothesis~(\ref{item:exp-weak}) on $\stp_2$, it
suffices to show that the assumption
\homj{\arr{\pottr{\stp_1}}{\comptr{\stp_2}}}{\mtm}{\mtm'} implies
\homj{\comptr{\stp_2}}{\app{\mtm}{\mtm_1}}{\app{\mtm'}{\mtm_1}}.  
Use the congruence rule with $f.\app{f}{\mtm_1}$.  

\proofcase{$\D$}
Because weakening is built into the definition, this is
immediate by transitivity.  
\end{enumerate}
\end{proof}

\begin{lemma}[Compositionality] \label{lem:lr-compositionality} \mbox{} 

\begin{enumerate}
\item 
\label{item:comp-exp}
\spfexpbound{\stm}{\mtm}{\valbound{-}{-}{\stp}}{\sspf} 
iff
\expbound{\stm}{\mtm}{\sarg{\sspf}{\stp}}.

\item 
\label{item:comp-val}
\spfvalbound{\sv}{\mtm}{\valbound{-}{-}{\stp}}{\sspf} 
iff
\valbound{\sv}{\mtm}{\sarg{\sspf}{\stp}}.
\end{enumerate}
\end{lemma}
% \begin{proof}
% (\ref{item:comp-exp}) follows by post-composing with~(\ref{item:comp-val}),
% and (\ref{item:comp-val}) follows by induction
% on~$\sspf$.  See the supplementary materials.  
% \end{proof}
\begin{proof}
\begin{enumerate}
\item Post-compose with part (\ref{item:comp-val}).

\item By induction on $\sspf$:

\proofcase{$\phi=\reccall$}
$\sarg{\reccall}{\stp} = \stp$, so we need to show
that \spfvalbound{\sv}{\mtm}{\valbound{-}{-}{\stp}}{\reccall} iff
\valbound{\sv}{\mtm}{\stp}, which is true by definition.

\proofcase{$t$ not free in $\phi$}
We need to show 
\spfvalbound{\sv}{\mtm}{\valbound{-}{-}{\stp}}{\tau} iff
\valbound{\sv}{\mtm}{\tau}, which is true by definition.

\proofcase{$\phi = \phi_0\cross\phi_1$} \\
\begin{tabular}{rcl}
$\spfvalbound{\sv}{\mtm}{\valbound{-}{-}{\stp}}{\sprd{\sspf_0}{\sspf_1}}$
& iff & $\sv=\spair{\sv_0}{\sv_1}$
where 
\spfvalbound{\sv_0}{\fst \mtm}{\valbound{-}{-}{\stp}}{\sspf} and \\
& & \spfvalbound{\sv_1}{\snd \mtm}{\valbound{-}{-}{\stp}}{\sspf'} (by
definition) \\
& iff & 
\valbound{\sv_0}{\fst \mtm}{\sarg{\sspf}{\stp}}
  and \valbound{\sv_1}{\snd \mtm}{\sarg{\sspf'}{\stp}} (by IH) \\
& iff &
\valbound{\sv}{\mtm}{\sprd{\sarg{\sspf}{\stp}}{\sarg{\sspf'}{\stp}}}
(by definition)   \\
& iff &
\valbound{\sv}{\mtm}{\sarg{\sprd{\sspf}{\sspf'}}{\stp}}
(by definition).   \\
\end{tabular}

\proofcase{$\phi = \sarr{\stp}{\sspf_0}$} \\
\begin{tabular}{rcl}
\spfvalbound{\sv}{\mtm_1}{\valbound{-}{-}{\stp}}{\sarr{\stp}{\sspf_0}} 
& iff & $\sv$ is {\slam{x}{\stm_1}}
  where 
  for all \valbound{\sv}{\mtm}{\stp}, \\
& & \spfexpbound{\subst{\stm_1}{\sv}{x}}{(\app{\mtm_1}{\mtm})}{\valbound{-}{-}{\stp}}{\sspf_0}
(by definition) \\
& iff  &
  for all \valbound{\sv}{\mtm}{\tau}, 
  \expbound{\subst{\stm_1}{\sv}{x}}{(\app{\mtm_1}{\mtm})}{\sarg{\sspf_0}{\stp}}
(by IH~(\ref{item:comp-exp})) \\
& iff &
\valbound{\slam{x}{\stm_1}}{\mtm_1}{\sarg{\sarr{\stp}{\sspf_0}}{\stp}}
(by definition).  \\
& iff &
\valbound{\sv}{\mtm_1}{\sarg{(\sarr{\stp}{\sspf_0})}{\stp}}
(by definition).  
\end{tabular}
\end{enumerate}
\end{proof}

\begin{lemma}
\label{lem:vbound-build}
If $\valbound{\sv_i}{\mtm_i}{\stp_i}$ for $i=0, 1$, then
$\valbound{\spair{\sv_0}{\sv_1}}{\cpair{\mtm_0}{\mtm_1}}{\sprod{\stp_0}{\stp_1}}$.
\end{lemma}
\begin{proof}
We need to show that 
$\valbound{\sv_i}{\mproj_i\mpair{\mtm_0}{\mtm_1}}{\sprod{\stp_0}{\stp_1}}$.
By the step rule for pairs we have that
$\mtm_i\mle \mproj_i\mpair{\mtm_0}{\mtm_1}$, and so by weakening
it suffices to show that $\valbound{\sv_i}{\mtm_i}{\stp_i}$, which is
given.
\end{proof}

\subsection{The Fundamental Theorem}

First we state two lemmas which say that, when applied to related
arguments, source-language $\smapkw$ is bounded by complexity-language
$\cmapkw$, and that source-language $\sreckw$ is bounded by
complexity-language $\creckw$.

\begin{lemma}[Map] \label{lem:lr-map}
Suppose:
\begin{enumerate}
\item \label{item:lr-map-styping}
$\typejudge{\tmoftype x {\tau_0}}{v_1}{\tau_1}$ and
$\typejudge\emptyset{v_0}{\subst*{\phi}{\tau_0}}$;
\item \label{item:lr-map-ctyping}
$\typejudge{x\oftype\typot{\tau_0}}{E_1}{\typot{\tau_1}}$ and
$\typejudge\emptyset{E_0}{\subst*{\typot{\phi}}{\typot{\tau_0}}}$;
\item \label{item:lr-map-v_0-bound}
$\isder{\mathcal E}{\spfvalbound{v_0}{E_0}{\valbound--{\tau_0}}{\phi}}$;
\item \label{item:lr-map-subders}
Whenever $\mathcal E'$ is a subderivation of $\mathcal E$ such that
$\isder{\mathcal E'}{\valbound{v_0'}{E_0'}{\tau_0}}$, 
$\valbound{\subst{v_1}{v_0'}{x}}{\subst{E_1}{E_0'}{x}}{\tau_0}$; and
\item \label{item:lr-map-eval}
$\evalin{\smap\phi{\bind x {v_1}}{v_0}}{v}{n}$.
\end{enumerate}
Then $n = 0$
and $\valbound{v}{\mmap{\typot\phi}{\bind x E_1}{E_0}}{\subst*\phi{\tau_0}}$.
% Assume \softps{\tptm{x}{\stp_1}}{}{\sv}{\stp_2} and
% \softps{\emptyset}{}{\sv_1}{\sarg{\sspf}{\stp_1}} and
% \oftps{\tptm{x}{\pottr{\stp_1}}}{}{\mtm}{\pottr{\stp_2}} and
% \oftps{\emptyset}{}{\mtm_1}{\sarg{\pottr{\sspf}}{\pottr{\stp_1}}}.  If $\mathcal{E} ::
% \spfvalbound{\sv_1}{\mtm_1}{\valbound{-}{-}{\stp_1}}{\sspf}$ and (for
% all $R$-subderivations of $\mathcal{E}$ deriving
% \valbound{\sv_1'}{\mtm_1'}{\stp_1}, 
% \valbound{\subst{\sv}{\sv_1'}{x}}{\subst{\mtm}{\mtm_1'}{x}}{\stp_2})
% then 
% if \seval{\smap{\sspf}{x.\sv}{\sv_1}}{\sc}{\sv_2}
% then \homj{}{\sc}{\costzero}
% and
% \valbound{\sv_2}{\mmap{\pottr{\sspf}}{x.\mtm}{\mtm_1}}{\sarg{\sspf}{\stp_2}}.  
% 
\footnote{We could have said 
\expbound{\smap{\sspf}{x.\sv_1}{\sv_0}}{\pair{0}{\mmap{\pottr{\sspf}}{x.\mtm_1}{\mtm_0}}}{\sarg{\sspf}{\stp_0}}
but this version of the lemma avoids needing the symmetric copy of the step rule for pairs.}
\end{lemma}
\begin{proof}

%% for pairs, uses only directed \beta 

The proof is by induction on~$\sspf$.  Lemma~\ref{lem:map_no_cost} shows
that $n=0$.

%\proofcase{$\phi = \reccall$,
%$\phi = \tau$ ($\reccall\notin\fv\tau$),
%$\phi = \phi_0\cross\phi_1$}  These cases follow directly by
%unravelling definitions.

\proofcase{$\phi = \reccall$}
Then $\evalin{\smap\phi{\bind x {v_1}}{v_0}}{v}{n}$ implies that
$\evalin{\subst{v_1}{v_0}{x}}{v}{n}$, and
$\mmap\phi{\bind x {E_1}}{E_0} = \subst{E_1}{E_0}{x}$.
By~(\ref{item:lr-map-v_0-bound}),
$\spfvalbound{v_0}{E_0}{\valbound--{\tau_0}}{\reccall}$, and so by definition
$\valbound{v_0}{E_0}{\tau_0}$.  Hence
by~(\ref{item:lr-map-subders}),
$v = \valbound{\subst{v_1}{v_0}{x}}{\subst{E_1}{E_0}{x}}{\tau_1}$.

% Assume (a) \spfvalbound{\sv_1}{\mtm_1}{\valbound{-}{-}{\stp_1}}{\reccall}
% and (b) for all $R$-subderivations of $\mathcal{E}$ deriving
% \valbound{\sv_1'}{\mtm_1'}{\stp_1},
% \valbound{\subst{\sv}{\sv_1'}{x}}{\subst{\mtm}{\mtm_1'}{x}}{\stp_2}.  
% 
% Assume
% \seval{\smap{\reccall}{x.\sv}{\sv_1}}{\sc}{\sv'}.
% By inversion, $\sc{}$ is \costzero\/ and $\sv'$ is 
% \subst{\sv}{\sv_1}{x}.  
% Thus, we need to show that 
% \homj{}{\costzero}{\costzero}, which is true by reflexivity, and 
% that 
% \valbound{\subst{\sv}{\sv_1}{x}}{\mmap{\pottr{\reccall}}{x.\mtm}{\mtm_1}}{\stp_2}.
% By expanding the definition 
% \[
% \mmap{\pottr{\reccall}}{x.\mtm}{\mtm_1} = \subst{\mtm}{\mtm_1}{x}
% \]
% Thus, it suffices to show 
% \valbound{\subst{\sv}{\sv_1}{x}}{\subst{\mtm}{\mtm_1}{x}}{}.
% The assumption (a)
% \spfvalbound{\sv_1}{\mtm_1}{\valbound{-}{-}{\stp_1}}{\reccall}
% is an $R-$subderivation of 
% \valbound{\sv_1}{\mtm_1}{\stp_1}, so 
% the assumption (b) gives what we needed to show.  

\proofcase{$\phi = \tau$ ($\reccall\notin\fv\tau$)}
This follows directly from the assumptions and definitions.

\proofcase{$\phi = \phi_0\cross\phi_1$}
Then $v_0=\spair{v_{00}}{v_{01}}$
and by inversion we have
\begin{prooftree}
  \AXC{$\evalin{\smapkw^{\phi_0}(\bind x {v_1}, v_{00})}{w_0}{n_0}$}
  \AXC{$\evalin{\smapkw^{\phi_1}(\bind x {v_1}, v_{01})}{w_1}{n_1}$}
\ndBIC{$\evalin{\smapkw^\phi(\bind x {v_1}, \spair{v_{00}}{v_{01}})}{\spair{w_0'}{w_1'}}{n_0+n_1}$}
\end{prooftree}
We also have
$\mathcal E$-subderivations
$\isder{\mathcal E_{0i}}{\spfvalbound{v_{0i}}{\pi_i(E_{0})}{\valbound--{\tau_0}}{\phi_i}}$.
Any subderivation of $\mathcal E_{0i}$ is a subderivation of~$\mathcal E$,
and so the induction hypothesis applies to~$v_{0i}$ and $\cproj_iE_{0}$, from
which we conclude that
$\valbound{w_i}{\mmap{\typot{\phi_i}}{\bind x {E_1}}{\cproj_iE_{0}}}{\subst*{\phi_i}{\tau_0}}$.
Thus we have that
\begin{align*}
v = \spair{w_0}{w_1}
& \vbounded \pair{\mmap{\typot{\phi_0}}{\bind x {E_1}}{\cproj_0E_{0}}}{\mmap{\typot{\phi_1}}{\bind x {E_1}}{\cproj_1E_{0}}} & & \text{(Lemma~\ref{lem:vbound-build})} \\
& = \mmap{\typot{\phi_0\cross\phi_1}}{\bind x {E_1}}{\mtm_0} \\
& = \mmap{\typot\phi}{\bind x {E_1}}{E_0}.
\end{align*}

\proofcase{$\phi = \tau\to\phi_0$}
Then $v_0 = \lambda y.e_0$ and $\mathcal E$ proves that
for all $\valbound {v'}{E'}{\tau}$,
$\spfexpbound{\subst{e_0}{v'}{y}}{E_0(E')}{\valbound--{\tau_0}}{\phi_0}$.
Since $v_0 = \lambda y.e_0$, 
$v = \lambda y.\slet{e_0}z{\smap{\phi}{\bind x {v_1}}{z}}$, and
so we must show that
$\valbound{\lambda y.\slet{e_0}z{\smap{\phi}{\bind x {v_1}}{z}}}{\mmap{\typot{\tau\to\phi_0}}{\bind x {E_1}}{E_0}}{\tau\to\subst*{\phi_0}{\tau_0}}$.
To do so, suppose $\valbound w F \tau$; we must show that
\[
\tag{*}
\expbound{\slet{\subst{e_0}{w}{y}}z{\smap{\phi}{\bind x {v_1}}{z}}}
{\mmap{\trans{\phi_0}}{\bind x {E_1}}{E_0(F)}}{\subst*{\phi_0}{\tau_0}}.
\]

Suppose
\begin{prooftree}
  \AXC{$\evalin{\subst{e_0}{w}{y}}{w_0}{n_0}$}
  \AXC{$\evalin{\smap{\phi_0}{\bind x {v_1}}{w_0}}{v'}{n_1}$}
\ndBIC{$\evalin{\slet{\subst{e_0}{w}{y}}z{\smap{\phi_0}{\bind x {v_1}}{z}}}{v'}{n_0+n_1}$}
\end{prooftree}
Since $\valbound w F\tau$, we have that $\mathcal E$ derives
$\spfexpbound{\subst{e_0}{w}{y}}{E_0(F)}{\valbound--{\tau_0}}{\phi_0}$ and
hence we have a subderivation~$\mathcal E_0$ of~$\mathcal E$ such that
$\isder{\mathcal E_0}
{\spfvalbound{w_0}{\pot{(E_0(F))}}{\valbound--{\tau_0}}{\phi_0}}$.
We now verify that (\ref{item:lr-map-subders}) holds for~$\mathcal E_0$
so that we can apply the induction hypothesis to
to $\smap\phi{\bind x {v_1}}{w_0}$.
So suppose that $\mathcal E_0'$ is a subderivation of $\mathcal E_0$
such that $\isder{\mathcal E_0'}
{\valbound{w_0'}{F_0'}{\tau_0}}$.  We need to show that
$\valbound{\subst{v_1}{w_0'}{x}}{\subst{E_1}{F_0'}{x}}{\tau_0}$, and to do
so it suffices to note that $\mathcal E_0'$ is a subderivation
of $\mathcal E_0$, which in turn is a subderivation
of~$\mathcal E$.%

We can now apply the induction hypothesis to conclude that $n_1=0$ and so:
\begin{align*}
n_0+n_1 = n_0 \leq \cst{(E_0\,F)}=
&\cst{(\mmap{\trans\phi}{\bind x {E_1}}{E_0\,F})} \\
{v'}\vbounded_{\subst*\phi{\tau_0}}{\mmap{\typot\phi}{\bind x {E_1}}{\pot{(E_0\,F)}}} = &\pot{(\mmap{\trans\phi}{\bind x {E_1}}{E_0\,F})}.
\end{align*}
Using $\beta$ for pairs,
these are the two conditions that must be verified to show~(*), so this
completes the proof.
\end{proof}

\begin{lemma}[Recursor]
\label{lem:recursor}
Fix a datatype declaration
$\sdatadecl\D\sC\phi$.
If \valbound{\sv}{\mtm}{\D}
and for all~$\sC$, 
$\expbound{e_\sC}{E_\sC}{\subst*{\phi_\sC}{\D\cross\ssusp\tau}}$,
% $\mtm_2$ such that \valbound{\sv_2}{\mtm_2}{\sarg{\sspf_\DC}{\sprd{\D}{\ssusp{\stp}}}},
% \expbound{\subst{\stm_C}{\sv_2}{x}}{\subst{\mtm_C}{\mtm_2}{x}}{})
then 
\expbound{\srec*{\sv}{C}{x}{}{\stm_C}}{\prec{\mtm}{C}{x}{}{\costone
    \costplusfst \mtm_C}}{}
\end{lemma}
\begin{proof}

By induction on \valbound{\sv}{\mtm}{\D}. The only case is
\[
\infer{\valbound{\sdcon{\DC}{\sv'}}{\mtm}{\D}}
      {\sdconbind{\DC}{\sspf}{\D} \in \ssig &
       \spfvalbound{\sv'}{\mtm'}{\valbound{-}{-}{\D}}{\sspf} &
       \homj{\D}{\dcon{\DC}{\mtm'}}{\mtm}}
\tag{$\dagger$}
\]

Assume $\srec*{\dcon{\DC}{\sv'}}{C}{x}{}{\stm_C}$ evaluates.  Then
by inversion and Lemma~\ref{lem:value-eval-inv} it was by
\[
%   \AXC{\vbox{
%     \hbox{$\seval{\dcon{\DC}{\sv'}}{\costzero}{\sdcon{\DC}{\sv'}}$}
%     \hbox{$\seval{\smap{\sspf}{y.\spair{y}{\sdelay{\srec*{y}{C}{x}{}{\stm_C}}}}{\sv'}}{\costzero}{\sv''}$}
%     \hbox{$\seval{\subst{\stm_\dsd{\DC}}{\sv''}{x}}{\sc_2}{\sv}$}
%   }}
% \ndUIC{$\seval{\srec*{\dcon{\DC}{\sv'}}{C}{x}{}{\stm_C}}{\costzero \costplus \costone \costplus \sc_2}{\sv}$}
    \AXC{$\seval{\dcon{\DC}{\sv'}}{\costzero}{\sdcon{\DC}{\sv'}}$}
    \AXC{$\seval{\smap{\sspf}{y.\spair{y}{\sdelay{\srec*{y}{C}{x}{}{\stm_C}}}}{\sv'}}{\costzero}{\sv''}$}
    \AXC{$\seval{\subst{\stm_\dsd{\DC}}{\sv''}{x}}{\sc_2}{\sv}$}
\ndTIC{$\seval{\srec*{\dcon{\DC}{\sv'}}{C}{x}{}{\stm_C}}{\costzero \costplus \costone \costplus \sc_2}{\sv}$}
\DisplayProof
\tag{*}
\]
Using the premise that $\homj{\D}{\sC E'}{E}$ from~($\dagger$), 
$\beta$ for datatypes, and
congruence, we note that
\begin{align*}
\creckw(\mtm, \cC\mapsto\bind x{\costone\costplusfst\mtm_{\cC}})
&\mge \prec{\sC\,E'}{\sC}{x}{}{\costone\costplusfst \mtm_C} \\
&\mge {\costone \costplusfst \subst{\mtm_{\DC}}{\mmap{{\pottr{\sspf}}}{y.\pair{y}{\prec{y}{C}{x}{}{\costone \costplusfst \mtm_C}}}{\mtm'}}{x}}
\end{align*}
Let us write~$E^*$ for
$\mmap{{\pottr{\sspf}}}{y.\pair{y}{\prec{y}{C}{x}{}{\costone \costplusfst \mtm_C}}}{\mtm'}$.
Thus by congruence, transitivity, weakening, and $\beta$ for pairs, it
suffices to show
\begin{align*}
{\costone \costplus \sc_2}&\mle{}{\costone \costplus \cst{\subst{\mtm_{\DC}}{E^*}{x}}}\\
{\sv}&\vbounded{\pot{(\subst{\mtm_{\DC}}{E^*}{x})}}{}
\end{align*}

By congruence for \costplus, for the first goal it suffices to show
$\homj{}{\sc_2}{\cst{\subst{\mtm_{\DC}}{E^*}{x}}}$.
Thus, if we can show 
$\expbound{{\subst{\stm_\dsd{\DC}}{\sv''}{x}}}{\subst{\mtm_{\DC}}{E^*}{x}}{}$,
then applying it to the third evaluation premise of~(*) gives the
result.  We can use our assumption 
that \expbound{\stm_\DC}{\mtm_\DC}{}, as long as we show 
$\valbound{\sv''}{E^*}{}$.
To do so, we use Lemma~\ref{lem:lr-map} applied to the second
evaluation premise of~(*) with
%(\seval{\smap{\sspf}{y.\spair{y}{\sdelay{\srec*{y}{C}{x}{}{\stm_C}}}}{\sv'}}{\costzero}{\sv''})
\[
\begin{aligned}
\sv_1 &= \sv' \\
\mtm_1 &= \mtm'
\end{aligned}
\quad
\begin{aligned}
\sv &= y.\spair{y}{\sdelay{\srec*{y}{C}{x}{}{\stm_C}}} \\
\mtm &= {y.\pair{y}{\prec{y}{C}{x}{}{\costone \costplusfst \mtm_C}}}
\end{aligned}
\]
We have $\mathcal{E} ::
\spfvalbound{\sv'}{\mtm'}{\valbound{-}{-}{\D}}{\sspf}$ from the second
premise of~($\dagger$).
Thus, to finish calling the theorem, we need to show that for all
$R$-position subderivations of $\mathcal{E}$ deriving
\valbound{\sv_1'}{\mtm_1'}{\D}, 
% \[
% \valbound{\spair{\sv_1'}{\sdelay{\srec*{\sv_1'}{C}{x}{}{\stm_C}}}}{\pair{\mtm_1'}{\prec{\mtm_1'}{C}{x}{}{\costone
%       \costplusfst \mtm_C}}}{\sprd{\D}{\ssusp{\stp}}}
% \]
\[
{\spair{\sv_1'}{\sdelay{\srec*{\sv_1'}{C}{x}{}{\stm_C}}}}\vbounded_{\sprd{\D}{\ssusp{\stp}}} \\
{\pair{\mtm_1'}{\prec{\mtm_1'}{C}{x}{}{\costone \costplusfst \mtm_C}}}
\]
By definition of value bounding at product types, weakening and $\beta$
for pairs, it suffices to show
\begin{align*}
{{\sv_1'}}&\vbounded_\D{\mtm_1'} \\
{\sdelay{\srec*{\sv_1'}{C}{x}{}{\stm_C}}}&\vbounded_{\ssusp\tau}{\prec{\mtm_1'}{C}{x}{}{\costone \costplusfst \mtm_C}}
\end{align*}
The former we have, and for the latter by definition it suffices to show
\[
\expbound{\srec*{\sv_1'}{C}{x}{}{\stm_C}}{\prec{\mtm_1'}{C}{x}{}{\costone \costplusfst \mtm_C}}{\stp}
\]
Because \valbound{{\sv_1'}}{\mtm_1'}{\D} is an $R$-subderivation of
\spfvalbound{\sv'}{\mtm'}{\valbound{-}{-}{\D}}{\sspf}, and therefore a
strict subderivation of {\valbound{\sdcon{\DC}{\sv'}}{\mtm}{\D}}, we can
use the inductive hypothesis on it, which gives exactly what we needed
to show.  
\end{proof}

% We write $\tptm{\ssub}{\sctx}$ for a source substitition $\sv_i/x_i$ for
% all $\tptm{x_i}{\tau_i}$ in $\sctx$, where
% \oftps{\emptyset}{\ssig}{\sv_i}{\tau_i}.  
% 
% We write $\tptm{\msub}{\mctx}$ for a complexity substitition
% $(\oftps{\emptyset}{\msig}{\mtm_i}{\mtp_x})/x_i$ for all $\tptm{x_i}{\mtp_i}$ in $\mctx$.
% 
% We write \subbound{\ssub}{\msub}{\sctx} to mean
% \valbound{\ssub(x_i)}{\msub(x_i)}{\stp_i} for all $\tptm{x_i}{\stp_i} \in \sctx$.
% 
\begin{thm}[Bounding Theorem]
\label{thm:mon-bounding}
If \softps{\sctx}{}{\stm}{\stp}, then $\expbound e {\trans e} \tau$.
% and \subbound{\ssub}{\msub}{\sctx}
% then
% \expbound{\tsubst{\stm}{\ssub}}{\tsubst{\comptr{e}}{\msub}}{\comptr{\stp}}.  
\end{thm}
\begin{proof}
The proof is by induction on the derivation of
$\typejudge\sctx\stm\tau$.  In each case we state the last line of the
derivation, taking as given the premises of the typing rules in
Figure~\ref{fig:source_lang}.

%% VAR
\proofcase{$\typejudge{\sctx,x\oftype\tau}{x}{\tau}$}
By definition \tsubst{\comptr{x}}{\msub} = \pair{0}{x}.  
$\tsubst{x}{\ssub} = \sv$ and $\tsubst{\pair{\costzero}{x}}{\msub} =
\pair{\costzero}{\mtm}$ where by assumption \valbound{\sv}{\mtm}{\stp}.
We must show that \expbound{\sv}{\pair{\costzero}{\mtm}}{\tau}.  
Assume \seval{\sv}{\sc}{\sv'}.  Then by 
inversion (Lemma~\ref{lem:value-eval-inv})
\homj{}{\sc}{\costzero} and $\sv' = \sv$.  Thus, by transitivity and $\beta$
for pairs, 
{\homj{}{\sc}{\cst{\pair{\costzero}{\mtm}}}}
and by weakening and $\beta$ for pairs
\valbound{\sv}{\pot{{\pair{\costzero}{\mtm}}}}{}.

%% PAIRS
\proofcase{$\typejudge\sctx{\spair{\stm_0}{\stm_1}}{\stp_0\cross\stp_1}$}
% \textbf{Case for}
% 
% \[
% \infer{\softps{\sctx}{\ssig}{\spair{\stm_1}{\stm_2}}{\sprd{\stp_1}{\stp_2}}}
%       {\softps{\sctx}{\ssig}{\stm_1}{\stp_1} &
%         \softps{\sctx}{\ssig}{\stm_2}{\stp_2}}
% \]
Expanding the definitions, we need to show
\[
\expbound{\spair{\tsubst{\stm_0}{\ssub}}{\tsubst{\stm_1}{\ssub}}}
         {\pair{\cst*{{\mtm_0}} \costplus \cst*{{\mtm_1}}}{\pair{\pot*{{\mtm_0}}}{\pot*{{\mtm_1}}}}}{}
\]
where $\mtm_0 = \tsubst{\comptr{\stm_0}}{\msub}$
and $\mtm_1 = \tsubst{\comptr{\stm_1}}{\msub}$.  By the IH, 
\expbound{\tsubst{\stm_0}{\ssub}}{\mtm_0}{\stp_0}
and 
\expbound{\tsubst{\stm_1}{\ssub}}{\mtm_1}{\stp_1}.  

Suppose 
\[
\infer{\seval{\spair{\tsubst{\stm_0}{\ssub}}{\tsubst{\stm_1}{\ssub}}}{\sc_0 \costplus \sc_1}{\spair{\sv_0}{\sv_1}}}
{
\seval{\tsubst{\stm_0}{\ssub}}{\sc_0}{\sv_0} &
\seval{\tsubst{\stm_1}{\ssub}}{\sc_1}{\sv_1}
}
\]
By the IH we have that $\expbound{\cl{\stm_i}{\ssub}}{E_i}{}$ and hence
$n_i\mle\cst*{E_i}$ and $\valbound{\sv_i}{\pot*{E_i}}{}$ for $i=0,1$.
% Applying the IH to the premises gives 
% \[
% \begin{array}{l}
% \homj{}{\sc_0}{\cst{{\mtm_0}}} \\
% \homj{}{\sc_1}{\cst{{\mtm_1}}} \\
% \valbound{\sv_0}{\pot{{\mtm_0}}}{\stp_0} \\
% \valbound{\sv_1}{\pot{{\mtm_1}}}{\stp_1} \\
% \end{array}
% \]
% 
Thus we conclude that
\[
% \begin{array}{rcl}
% \sc_0 \costplus \sc_1 &\le& {\cst{{\mtm_0}}} \costplus {\cst{{\mtm_1}}} \\
% \spair{\sv_0}{\sv_1} &\vbounded& \pair{\pot*{E_0}}{\pot*{E_1}}
% \end{array}
\sc_0 \costplus \sc_1 \le {\cst{{\mtm_0}}} \costplus {\cst{{\mtm_1}}}
\qquad
\spair{\sv_0}{\sv_1} \vbounded \pair{\pot*{E_0}}{\pot*{E_1}}
\]
and the result follows by weakening and $\beta$ for pairs.

%% SPLIT
\proofcase{$\typejudge\sctx{\ssplit{\stm_0}{x_0}{x_1}{\stm_1}}{\stp}$}
Expanding definitions, we need to show
% \[
% \expbound{\ssplit{\cl{\stm_0}{\ssub}}{x_0}{x_1}{\cl{\stm_1}{\ssub,x_0/x_0,x_1/x_1}}}{\costpluscpy{\cst*{{E_0}}}{{E_1}}}{} =
% \pair{\cst*{E_0} + \cst*{E_1}}{\pot*{E_1}}
% \]
\[
{\ssplit{\cl{\stm_0}{\ssub}}{x_0}{x_1}{\cl{\stm_1}{\ssub,x_0/x_0,x_1/x_1}}}\bounded \\
{\costpluscpy{\cst*{{E_0}}}{{E_1}}} =
\pair{\cst*{E_0} + \cst*{E_1}}{\pot*{E_1}}
\]
% \[
% \begin{array}{l}
% \expbound{{\ssplit{\tsubst{\stm_0}{\ssub}}{x_0}{x_1}{\tsubst{\stm_1}{\ssub,x_0/x_0,x_1/x_1}}}}{}{}\\
% {\pair{\cst{\mtm}
%     \costplus
%     \cst{(\tsubst{\comptr{\stm_1}}{\ssub,\fst{\pot{\mtm}/{x},\snd{\pot{\mtm}}/{y}}})}}
%       {\pot{(\tsubst{\comptr{\stm_1}}{\ssub,\fst{\pot{\mtm}/{x},\snd{\pot{\mtm}}/{y}}})}}}
% {}
% \end{array}
% \]
where $\mtm_0 = \tsubst{\comptr{\stm_0}}{\msub}$ and
$\mtm_1 = \cl{\trans{e_1}}{\msub,\fst{\pot*{E_1}}/x_0,\snd{\pot*{E_1}}/x_1}$.

Suppose
\[
\infer{\seval{\ssplit{\tsubst{\stm_0}{\ssub}}{x_0}{x_1}{\stm_1}}{\sc_0 \costplus \sc_1}{\sv}}
      {\seval{\tsubst{\stm_0}{\ssub}}{\sc_0}{\spair{\sv_0}{\sv_1}} &
        \seval{{\tsubst{\stm_1}{\ssub,\sv_0/x_0,\sv_1/x_1}}}{\sc_1}{\sv}}
\]
We apply the induction hypothesis as follows:
\begin{enumerate}
\item From $\expbound{\cl{\stm_0}{\ssub}}{\mtm_0}{}$:
\begin{enumerate}
\item $n_0\leq \cst*{E_0}$;
\item $\valbound{\spair{\sv_0}{\sv_1}}{\pot*{E_0}}{}$ and hence
$\valbound{\sv_i}{\sproj_i(\pot*{E_0})}{}$ for $i=0, 1$.
\end{enumerate}
\item From $\expbound{\stm_1}{\trans{e_1}}{}$, $\subbound\ssub\msub{}$, and
$\valbound{\sv_i}{\sproj_i(\pot*{E_0})}{}$ for $i=0, 1$,
\begin{enumerate}
\item $\subbound{\ssub,\sv_0/x_0,\sv_1/x_1}{\msub,\sproj_0(\pot*{E_0})/x_0,\sproj_1(\pot*{E_1})/x_1}{}$, and hence
\item $n_1\leq\cst*{E_1}$;
\item $\valbound{\sv}{\pot*{E_1}}{}$.
\end{enumerate}
\end{enumerate}
Thus we conclude that
\[
% \begin{array}{rcl}
% n_0 + n_1 &\leq& \cst*{E_0} \costplus \cst*{E_1} \\
% \sv &\vbounded& \pot*{E_1}
% \end{array}
n_0 + n_1 \leq \cst*{E_0} \costplus \cst*{E_1}
\qquad
\sv \vbounded \pot*{E_1}
\]
and the result follows by monotoncity of~$\costplus$,
weakening, and $\beta$ for pairs.

%%% LAM and APP
%\proofcase{$\typejudge\sctx{\lambda x.e}{\sigma\to\tau}$,
%$\typejudge\sctx{\sapp{e_0}{e_1}}{\tau}$}
%These cases are nearly identical to the corresponding cases
%in \citet{danner-et-al:plpv13}; we refer the reader to that paper or
%the full paper for details.

%% LAM
\proofcase{$\typejudge\sctx{\lambda x.e}{\sigma\to\tau}$}
% \textbf{Case for}
% 
% \[
% \infer{\softps{\sctx}{\ssig}{\slam{x}{\stm}}{\sarr{\stp_1}{\stp_2}}}
%       {\softps{\sctx,\tptm{x}{\stp_1}}{\ssig}{\ssig}{\stm}{\stp_2}}
% \]
% 
Expanding the definitions, 
\begin{align*}
\tsubst{(\slam{x}{e})}{\ssub} &= \slam{x}{\tsubst{e}{\ssub,x/x}} \\
\tsubst{\comptr{\slam{x}{\stm}}}{\msub} &= \pair{\costzero}{\ulam{x}{\tsubst{\comptr{e}}{\msub,x/x}}}
\end{align*}
\sloppypar Assume \slam{x}{\tsubst{\stm}{\ssub,x/x}} evaluates.
By inversion we have
$
\seval{\slam{x}{\tsubst{\stm}{\ssub,x/x}}}{\costzero}{\slam{x}{\stm}}
$.
Applying transitivity/weakening and $\beta$ for pairs we need to show that
\homj{}{\costzero}{\costzero} (trivial) and 
$
\valbound{\slam{x}{\tsubst{\stm}{\ssub,x/x}}}{\ulam{x}{\tsubst{\comptr{e}}{\msub,x/x}}}{\sarr{\sigma}{\stp}}
$.
Assume \valbound{\sv_1}{\mtm_1}{\stp_1}; we need to show
$
\expbound{\subst{\tsubst{\stm}{\ssub,x/x}}{\sv_1}{x}}{\app{(\ulam{x}{\tsubst{\comptr{e}}{\msub,x/x}})}{\mtm_1}}{\stp}
$.
By weakening, $\beta$ for functions, and Lemmas
~\ref{lem:source-subst-composition} and
~\ref{lem:mono-subst-composition}, it suffices to show
$
\expbound{\tsubst{\stm}{\ssub,\sv_1/x}}{\tsubst{\comptr{e}}{\msub,\mtm_1/x}}{\stp}
$.
Because \subbound{\ssub}{\msub}{} and \valbound{\sv_1}{\mtm_1}{\stp_1}, we
have \subbound{\ssub,\sv_1/x}{\msub,\mtm_1/x}{}, so the IH gives the
result.  

%% APP
\proofcase{$\typejudge\sctx{\sapp{e_0}{e_1}}{\tau}$}
% \textbf{Case for}
% \[
% \infer{\softps{\sctx}{\ssig}{\sapp{\stm_0}{\stm_1}}{\stp}}
%       {\softps{\sctx}{\ssig}{\stm_0}{\sarr{\stp_2}{\stp}} &
%         \softps{\sctx}{\ssig}{\stm_1}{\stp_2}}
% \]
% 
By definition, 
$
\tsubst{(\sapp{\stm_0}{\stm_1})}{\ssub} = 
\sapp{\tsubst{\stm_0}{\ssub}}{\tsubst{\stm_1}{\ssub}}
$
and 
$\cl{\trans{\sapp{e_0}{e_1}}}{\msub} =
%\costpluscpy{(\cst{\trans{\stm_0}} + \cst{\trans{\stm_1}})}{\app{\pot{\trans{\stm_0}}}{\pot{\trans{\stm_1}}}} =
\pair{\cst*{E_0} + \cst*{E_1} + \cst E}{\pot E}
$
where $\mtm_i = \tsubst{\comptr{\stm_i}}{\msub}$ for $i=0,1$
and $E = \app{\pot*{E_0}}{\pot*{E_1}}$.
Suppose that
\[
\infer{\seval{\sapp{\tsubst{\stm_0}{\ssub}}{\tsubst{\stm_1}{\ssub}}}{\sc_0 \costplus \sc_1 \costplus \sc}{\sv}}
      {\seval{\tsubst{\stm_0}{\ssub}}{\sc_0}{\ulam{x}{\stm_0'}} &
        \seval{\tsubst{\stm_1}{\ssub}}{\sc_1}{\sv_1} &
        \seval{\subst{\stm_0'}{\sv_1}{x}}{\sc}{\sv}}
\]
We have the following facts from the induction hypothesis:
\begin{inparaenum}[(1)]
\item From $\expbound{\cl{\stm_0}{\ssub}}{E_0}{}$:
\begin{inparaenum}[(a)]
\item $n_0\leq \cst*{E_0}$ and
\item $\valbound{\ulam x {\stm_0'}}{\pot*{E_0}}{}$;
\end{inparaenum}
\item From $\expbound{\cl{\stm_1}{\ssub}}{E_1}{}$:
\begin{inparaenum}[(a)]
\item $n_1\leq\cst*{E_1}$ and
\item $\valbound{\sv_1}{\pot*{E_1}}{}$;
\end{inparaenum}
\item From (1b) and (2b) and
the definition of $\vbounded$, $\expbound{\subst{\stm_0'}{\sv_1}{x}}E{}$, so
\begin{inparaenum}[(a)]
\item $n\leq \cst E$ and
\item $\valbound{\sv}{\pot E}{}$.
\end{inparaenum}
\end{inparaenum}
Thus we conclude
\[
% \begin{array}{rcl}
% n_0+n_1+n &\mle& \cst*{E_0} + \cst*{E_1} + \cst E \\
% v &\vbounded& \pot E
% \end{array}
n_0+n_1+n \mle \cst*{E_0} + \cst*{E_1} + \cst E
\qquad
v \vbounded \pot E
\]
and the result follows from weakening and $\beta$ for pairs.

%% DELAY
\proofcase{$\typejudge\sctx{\sdelay\stm}{\ssusp\tau}$}
Expanding definitions, we need to show
$\expbound{\sdelay{\tsubst{\stm}{\ssub}}}{\pair{\costzero}{\tsubst{\comptr{\stm}}{\msub}}}{}$,
so suppose
$\seval{\sdelay{\tsubst{\stm}{\ssub}}}{\costzero}{\sdelay{\tsubst{\stm}{\ssub}}}$.
We have $\homj{}{\costzero}{\pot{\pair{\costzero}{\ldots}}}$ by $\beta$ 
for pairs.
For the potential goal, we must show that
\valbound{{\sdelay{\tsubst{\stm}{\ssub}}}}{{\tsubst{\comptr{\stm}}{\msub}}}{\ssusp{\tau}}.
By definition, this means showing
\expbound{\tsubst{\stm}{\ssub}}{{\tsubst{\comptr{\stm}}{\msub}}}{\tau},
which is exactly the IH.  The result follows from weakening and
$\beta$ for pairs.

%% FORCE
\proofcase{$\typejudge\sctx{\sforce\stm}{\stp}$}
Expanding definitions, we need to show
$\expbound{\sforce{\tsubst{\stm}{\ssub}}}{\pair{\cst{\mtm} \costplus \cst*{\pot{\mtm}}}{\pot*{\pot{\mtm}}}}{\stp}$
where $\mtm = \tsubst{\comptr{\stm}}{\msub}$.  
Suppose
\[
\infer{\seval{\sforce{\tsubst{\stm}{\ssub}}}{\sc_0 \costplus \sc_1}{\sv}}
      {\seval{\cl\stm\ssub}{\sc_0}{\sdelay{\stm'}} &
        \seval{\stm'}{\sc_1}{\sv}}
\]
Since $\expbound{\cl\stm\ssub}{E}{}$, we have that
$n_0\leq \cst E$ and $\valbound{\sdelay{e'}}{\pot E}{}$.  From
the definition of $\vbounded$, 
$\expbound{e'}{\pot E}{}$, and hence
$n'\leq\cst*{\pot E}$ and
$\valbound v {\pot*{\pot E}}{}$.  The result follows from
monotonicity of~$\costplus$ and $\beta$ for pairs.

%% Datatype constructor
\proofcase{$\typejudge\sctx{\DC\stm}\D$}
We must show that $\expbound{\cl{\DC\stm}{\ssub}}{\pair{\cst\mtm}{\DC(\pot\mtm)}}{}$, where
$\mtm = \cl{\trans e}{\msub}$.  Suppose
\[
\infer{\seval{{\dcon{\DC}{\tsubst{\stm}{\ssub}}}}{\sc}{\dcon{\DC}{\sv}}}
      {\seval{\tsubst{\stm}{\ssub}}{\sc}{\sv}}.
\]
Since $\expbound{\cl e\ssub}{E}{}$,
$n\leq\cst E$ (satisfying the cost goal) and 
$\valbound v {\pot E}{\subst*\phi\D}$.
By Lemma~\ref{lem:lr-compositionality},
$\spfvalbound v {\pot\mtm}{\valbound--\D}{\phi}$.
Since $\DC(\pot\mtm)\leq\DC(\pot\mtm)$ by reflexivity, we have that
$\valbound{\DC v}{\DC(\pot\mtm)}{}$ by definition of $\vbounded$.

%REC
\proofcase{$\typejudge\sctx{\srec\stm\DC {\bind x {\stm_C}}}{\stp}$}
We need to show
\[
\expbound{\srec*{\tsubst{\stm}{\ssub}}{C}{x}{}{\tsubst{\stm_C}{\ssub,x/x}}}
         {\pair{\cst{\mtm} \costplus \cst*{{\mtm_r}}} {\pot*{{\mtm_r}}}}
         {}
\]
where $\mtm = \tsubst{\comptr{\stm}}{\msub}$ and
$\mtm_r = {\prec{\pot\mtm}{C}{x}{}{(\costone \costplusfst \tsubst{\comptr{\stm_C}}{\msub,x/x})}}$.
Suppose
\begin{prooftree}
      \AXC{$\evalin{\cl e\ssub}{\DC v_0}{n_0}$}
      \AXC{$\evalin{\smap{\phi_C}{\bind y{\spair{y}{\sdelay{\srec y\DC {\bind x{\cl{\stm_C}{\extend\ssub x x}}}}}}}{v_0}}{v_1}{0}$}
      \AXC{$\evalin{\cl{e_C}{\extend\ssub x {v_1}}}{v}{n_2}$}
\ndTIC{$\evalin{\srec{\cl\stm\ssub}\DC {\bind x {\cl{\stm_C}{\extend\ssub x x}}}}{v}{1+n_0+n_2}$}
\end{prooftree}

By the induction hypothesis $\expbound{\cl e\ssub}{E}{}$, so
$n_0\leq\cst E$ and $\valbound{\DC v_0}{\pot E}{}$.
By Lemma~\ref{lem:value-eval-inv} we can derive
\begin{prooftree}
      \AXC{$\evalin{\DC v_0}{\DC v_0}{_0}$}
      \AXC{$\evalin{\smap{\phi_C}{\bind y{\spair{y}{\sdelay{\srec y\DC {\bind x{\cl{\stm_C}{\extend\ssub x x}}}}}}}{v_0}}{v_1}{0}$}
      \AXC{$\evalin{\cl{e_C}{\extend\ssub x {v_1}}}{v}{n_2}$}
\ndTIC{$\evalin{\srec{\DC v_0}\DC {\bind x {\cl{\stm_C}{\extend\ssub x x}}}}{v}{1+n_2}$}
\end{prooftree}
So by Lemma~\ref{lem:recursor} we have that
$1+n_2\leq\cst*{E_r}$ and $\valbound v {\pot*{E_r}}{}$.  Putting these
together, we have what we needed to show:
\[
1 + n_0+n_2 \leq \cst E + \cst*{E_r}
\qquad
v \vbounded \pot*{E_r}
\]

%MAP
\proofcase{$\typejudge\sctx{\smap\phi{\bind x {v_1}}{v_0}}{\subst*{\phi}{\tau_1}}$}
Because $\sv_1$ is a sub-syntactic-class of $\stm$, we can upcast it and apply
\comptr{\sv_1} to it, producing a complexity expression.  
We must show that
\[
%\expbound{\smap{\sspf}{x.\tsubst{\sv_1}{\ssub,x/x}}{\tsubst{\sv_0}{\ssub}}}{\pair{\costzero}{\mmap{\pottr{\sspf}}{x.\pot{\tsubst{\comptr{\sv_1}}{\msub,x/x}}}{\pot{\tsubst{\comptr{\sv_0}}{\msub}}}}}{},
{\smap{\sspf}{x.\tsubst{\sv_1}{\ssub,x/x}}{\tsubst{\sv_0}{\ssub}}}\bounded
{\pair{\costzero}{\mmap{\pottr{\sspf}}{x.\pot{\tsubst{\comptr{\sv_1}}{\msub,x/x}}}{\pot{\tsubst{\comptr{\sv_0}}{\msub}}}}},
\]
so suppose
\seval{{\smap{\sspf}{x.\tsubst{\sv_1}{\ssub,x/x}}{\tsubst{\sv_0}{\ssub}}}}{\sc}{\sv}.  
By transitivity/weakening with $\beta$ for pairs, it suffices to show:
\[
{\sc} \mle  {\costzero}
\qquad
{\sv} \vbounded {\mmap{\pottr{\sspf}}{\pot{\tsubst{\comptr{\sv_1}}{\msub,x/x}}}{\pot{\tsubst{\comptr{\sv_0}}{\msub}}}} {}
\tag{*}
\]
We will apply Lemma~\ref{lem:lr-map} with
\[
\begin{aligned}
\sv_0 &= {\tsubst{\sv_0}{\ssub}}\\
\mtm_0 &= {\pot{\tsubst{\comptr{\sv_0}}{\msub}}}\\
\end{aligned}
\qquad
\begin{aligned}
\sv_1 &= \tsubst{\sv_1}{\ssub,x/x} \\
\mtm_1 &= \pot{\tsubst{\comptr{\sv_1}}{\msub,x/x}}
\end{aligned}
\]

To establish condition~(\ref{item:lr-map-v_0-bound}) we apply the IH
to~$v_0$ to conclude that
$
\expbound{\tsubst{\sv_0}{\ssub}}{\tsubst{\comptr{\sv_0}}{\msub}}{\sarg{\sspf}{\stp_0}}
$.
Since {\tsubst{\sv_0}{\ssub}} is a value, by
Lemma~\ref{lem:value-eval-inv}, it evaluates to itself.  Therefore
$
\valbound{\tsubst{\sv_0}{\ssub}}{\pot{\tsubst{\comptr{\sv_0}}{\msub}}}{\sarg{\sspf}{\stp_0}}
$ and so by
Lemma~\ref{lem:lr-compositionality},
$
\spfvalbound{\tsubst{\sv_0}{\ssub}}{\pot{\tsubst{\comptr{\sv_0}}{\msub}}}{\valbound{-}{-}{\stp_0}}{\sspf}
$.

To establish condition~(\ref{item:lr-map-subders}), 
assume \valbound{\sv_0'}{\mtm_0'}{\tau_0}
(which is an $R$-subderivation of the above, but we won't use this
fact).  Using the substitution lemmas we need to show
$
\valbound {\tsubst{\sv_1}{\ssub,\sv_0'/x}} {\pot{\tsubst{\comptr{\sv_1}}{\msub,\mtm_0'/x}}} {}
$.
Since \subbound{\ssub,\sv_0'/x}{\msub,\mtm_0'/x}{}, the IH on $\sv_1$ gives 
$
\expbound{\tsubst{\sv_1}{\ssub,\sv_0'/x}}{\tsubst{\comptr{\sv_1}}{\msub,\mtm_0'/x}} {}
$
and since {\tsubst{\sv_1}{\ssub,\sv_0'/x}} is a value, it evaluates to
itself, so
$
\valbound{\tsubst{\sv_1}{\ssub,\sv_0'/x}}{\pot{\tsubst{\comptr{\sv_1}}{\msub,\mtm_0'/x}}} {}
$
as we needed to show.  

Now we apply Lemma~\ref{lem:lr-map} to
\seval{{\smap{\sspf}{x.\tsubst{\sv_1}{\ssub,x/x}}{\tsubst{\sv_0}{\ssub}}}}{\sc}{\sv}
to conclude~(*).

%% LET
\proofcase{$\typejudge\sctx{\slet{e_0} x {e_1}}{\tau}$}
% \textbf{Case for} 
% 
% \[
% \infer{\softps{\sctx}{\ssig}{\slet{\stm_1}{x}{\stm_2}}{\stp_2}}
%       {\softps{\sctx}{\ssig}{\stm_1}{\stp_1} &
%         \softps{\sctx,\tptm{x}{\stp_1}}{\ssig}{\stm_2}{\stp_1}}
% \]
Applying the substitution lemmas, we need to show
% \[
% \expbound{\slet{\tsubst{\stm_0}{\ssub}}{x}{\tsubst{\stm_1}{\ssub,x/x}}}{\pair{\cst{{\mtm_0}}
%   \costplus \cst{{\tsubst{\comptr{\stm_1}}{\msub,\pot{{\mtm_0}}/x}}}}{\pot{{\tsubst{\comptr{\stm_1}}{\msub,\pot{{\mtm_0}}/x}}}}}{}
% \]
\[
{\slet{\tsubst{\stm_0}{\ssub}}{x}{\tsubst{\stm_1}{\ssub,x/x}}}\bounded
{\pair{\cst{{\mtm_0}}
  \costplus \cst{{\tsubst{\comptr{\stm_1}}{\msub,\pot{{\mtm_0}}/x}}}}{\pot{{\tsubst{\comptr{\stm_1}}{\msub,\pot{{\mtm_0}}/x}}}}}
\]
where $\mtm_0 = \tsubst{\comptr{\stm_0}}{\msub}$.  

Assume let evaluates, then by inversion and applying the substitution
lemma, 
\[
\infer{\seval{\slet{\tsubst{\stm_0}{\ssub}}{x}{\tsubst{\stm_1}{\ssub,x/x}}}{\sc_0 \costplus \sc_1}{\sv_1}}
      {\seval{\tsubst{\stm_0}{\ssub}}{\sc_0}{\sv_0} &
        \seval{\tsubst{\stm_1}{\ssub,\sv_0/x}}{\sc_1}{\sv_1}}
\]
Applying the IH to~$e_0$ gives
% \begin{array}{l}
% \homj{} {\sc_0}{\cst{{\mtm_0}}} \\
% \valbound {\sv_0} {\pot{{\mtm_0}}} {} \\
% \end{array}
$\homj{} {\sc_0}{\cst{{\mtm_0}}}$
and
$\valbound {\sv_0} {\pot{{\mtm_0}}} {}$.
Therefore \subbound{\ssub,\sv_0/x}{\msub,{\pot{{\mtm_0}}}/x}, so applying
the IH to the evaluation of~$\cl{\stm_1}{\extend\ssub{v_0}x}$ gives
\[
% \begin{array}{l}
% \homj {} {\sc_1} {\cst{\tsubst{\comptr{\stm_1}}{\msub,{\pot{{\mtm_0}}}/x}}} \\
% \valbound {\sv_1} {\pot{\tsubst{\comptr{\stm_1}}{\msub,{\pot{{\mtm_0}}}/x}}} {}.
% \end{array}
\homj {} {\sc_1} {\cst{\tsubst{\comptr{\stm_1}}{\msub,{\pot{{\mtm_0}}}/x}}}
\qquad
\valbound {\sv_1} {\pot{\tsubst{\comptr{\stm_1}}{\msub,{\pot{{\mtm_0}}}/x}}} {}.
\]
Monotonicity of $\costplus$ gives
$
\homj {} {\sc_0 \costplus \sc_1} {{\cst{{\mtm_0}}} \costplus  {\cst{\tsubst{\comptr{\stm_1}}{\msub,{\pot{{\mtm_0}}}/x}}}}
$
so transitivity/weakening and $\beta$ for pairs gives the results.  
\end{proof}

\section{Models of the Complexity Language}
\label{sec:mon-interp-examples}

A model of the complexity language consists of an interpretation of 
types as preorders, and of terms as maps between elements of those preorders, validating
the rules of Figure~\ref{fig:monotonic_type_theory}.  The congruence
contexts $\congctx$, but not all terms, need to be monotone maps.  
%% For any such model, we have
%% \begin{thm}[Semantic Bounding Theorem] \label{thm:sem-bounding}
%% If ${\emptyset} \vdash {e} : {\tau}$ and $\evalin{e}{v}{n}$ 
%% then $\den n {} \leq \den{{\trans e}_c}{}$.  
%% \end{thm}
%% \begin{proof}
%% By Theorem~\ref{thm:mon-bounding}, $\homj{}{n}{{\trans e}_c}$ in the
%% complexity language.  Because the preorder rules are sound for the
%% model, $\den n {} \le {\den {{\trans e}_c}{}}$.  
%% \end{proof}

\subsection{The Size-Based Complexity Semantics}
\label{sec:size-based-is-mon}

We showed in Section~\ref{sec:size_semantics} that the size-based
semantics interpets the syntax of the complexity language; it is also a
model of the preorder rules of Figure~\ref{fig:monotonic_type_theory}.
Congruence is established by induction on $\congctx$; we do
not need programmer-defined size functions to be monotonic, because
there is no congruence context for datatype constructors.
%% In the case that $\congctx = \cC\mtm_0$
%% the verification requires the additional assumption that $\semsize$ be
%% monotontic.  Though a natural assumption, it is not necessary, because
%% as noted, the proof of the bounding theorem does not require congruence
%% for contexts of the form~$\sC\hole$.  
The step rule for the recursor is verified as follows:
\begin{align*}
\llbracket\creckw&(\cC\mtm_0,\overline{x\mapsto\mtm_{\cC}})\rrbracket\xi \\
%\den{\crec {\cC\mtm_0} {x} {\mtm_{\cC}}}{\xi} \\
  &= \bigmax_{\semsize z\leq \den{\cC\mtm_0}\xi}\semcase(z,(\dots,f_\cC,\dots)) \\
  &= \bigmax_{\semsize z\leq \semsize(\cC\den{\mtm_0}\xi)}\semcase(z,(\dots,f_\cC,\dots)) \\
  &\geq \semcase(\cC\den{\mtm_0}\xi,(\dots,f_\cC,\dots)) \\
  % &= f_{\cC}(\mtm_0) \\
%  &= \den{\mtm_{\cC}}{\extendenv\xi x {\den{\cmap{\Phi_{\cC}}{\bind w {\cpair{w}{\crec{w}{x}{\mtm_{\cC}}}}}{y}}{\extendenv\xi{y}{\den{\mtm_0}\xi}}}} \\
  &= \den{\mtm_{\cC}}{\extendenv\xi x {\den{\cmap{\Phi_{\cC}}{\bind w {\cpair{w}{\crec{w}{x}{\mtm_{\cC}}}}}{\mtm_0}}{\xi}}}.
\end{align*}
Therefore, Theorem~\ref{thm:bounding} is a corollary of Theorem~\ref{thm:mon-bounding}.

\subsection{Infinite-Width Trees}

Infinite-width trees can be defined by a
datatype declaration with a function argument, such as
\[
\sdatatypekw\:\stree = \sconstrdecl E \sunit \sconstrsep \sconstrdecl N {\sprd{\sint}{(\sarr{\snat}{\stree})}}
\]
Though every branch in such a tree is of finite length, the height of a
tree is in general not a finite natural number.%
\footnote{Because we can only construct values using
$\sreckw$, we cannot define infinite-length branches (i.e.,
coinductively-defined data) in our source language.}  
However, the size-based semantics adapts easily to interpret $\stree$ by
a suitably large infinite successor ordinal, and then defining
$\semsize(N(x, f)) = \bigmax_{y\in\den\snat{}}f(y) + 1$.

\subsection{A Semantics Without Arbitrary Maximums}

The language studied in \citet{danner-et-al:plpv13} can be viewed as a
specific signature in the present language.  Their language has a
type of booleans, a type $\sint$ of fixed-size integers, and a type $\slist$
of integer lists.  As in Example~\ref{ex:tree-mem}, we can treat $\sint$
and $\sbool$ as enumerated datatypes with unit-cost operations.
The $\slist$ type is defined as a datatype
and its case and fold operators are easily defined
using $\sreckw$.

For this specific signature, we can give a semantics of the complexity
language that does not require arbitrary maximums in the semantics of
each type, and where we interpret $\slist$ by $\N$, the natural numbers.
Set $\den{\cnil}{\xi} = 0$ and $\den{\ccons(\mtm_0,\mtm_1)}{\xi} =
\den{\mtm_1}\xi+1$.   
Define a semantic primitive recursion
operator
$\semrec[\sigma]\oftype\N\cross\sigma\cross(\N\cross\sigma\to\sigma)\to\sigma$
by
% \begin{align*}
% \semrec(0, a, f) &= a \\
% \semrec(n+1, a, f) &= a\bmax f(n, \semrec(n, a, f)).
% \end{align*}
\[
\semrec(0, a, f) = a 
\qquad
\semrec(n+1, a, f) = a\bmax f(n, \semrec(n, a, f)).
\]
Finally, set
\[
\den{\creckw(E)}\xi =  \\
\semrec(\den\mtm\xi, \den{\mtm_{\cnil}}{\xi}, \llambda n,w.\den{\mtm_{\ccons}}{\extendenv\xi{x,xs,r}{1,n,w}}).
\]
where
$\creckw(\mtm) =
{\crecm{\mtm}{\cnil\mapsto {\mtm_{\cnil}},\ccons\mapsto\bind{\cpair{x}{\spair{xs}{r}}}{\mtm_{\ccons}}}}$.
Verifying the preorder rules from Figure~\ref{fig:monotonic_type_theory}
is straightforward in
all cases except the last, which we verify as follows:
\begin{align*}
\den{\creckw(\cnil)}\xi
&= \den{\mtm_{\cnil}}{\extendenv\xi x 1} \\
&= \den{\subst{\mtm_{\cnil}}{\ctriv}{x}}\xi \\
&= \den{\subst{\mtm_{\cnil}}{\cmapkw^{\cunit}(\bind y {\cpair{y}{\creckw(y)}}, \ctriv)}{x}}\xi
\end{align*}
and
\begin{align*}
\llbracket\creckw(\ccons(\mtm_0,\mtm_1))\rrbracket\xi
&= (\den{\mtm_{\cnil}}{\extendenv\xi{x}{1}})\bmax \\
&\qquad(\den{\mtm_{\ccons}}{\extendenv\xi{x,xs,r}{1,\den{\mtm_1}\xi,\semrec(\den{\mtm_1}\xi,\dots)}}) \\
&\geq \den{\mtm_{\ccons}}{\extendenv\xi{x,xs,r}{1,\den{\mtm_1}\xi,\semrec(\den{\mtm_1}\xi,\dots)}} \\
&=\den{\subst{\mtm_{\ccons}}{\mtm_0, \cpair{\mtm_1}{\creckw(\mtm_1)}}{x,\cpair{xs}{r}}}\xi \\
&= \den{\subst{\mtm_{\ccons}}{\mtm_0,\cmapkw(\bind y {\spair{y}{\creckw(y)}},\cpair{\mtm_0}{\mtm_1})}{x,xs,r}}\xi.
\end{align*}
A natural question is why we must take 
$\semrec(n+1,a,f) = a\bmax f(\semrec(n, a, f))$, since the above proof
seems to carry through with $\semrec(n+1,a,f) = f(\semrec(n, a, f))$.
The problem is that if we use this latter definition, then the resulting
interpretation fails to satisfy the congruence axiom for
contexts of the form $\creckw(\hole,\dots)$.

\subsection{Exact Costs}

If we wish to reason about exact costs, we can symmetrize the
inequalities in Figure~\ref{fig:monotonic_type_theory} into equalities,
and add congruence for all contexts, which makes the $\homj{}{E_0}{E_1}$
judgement into a standard notion of definitional equality.  Then we can
take the term model in the usual way, interpreting each type as a set of
terms quotiented by this definitional equality.
% with the following
% equations, writing $\creckw(E)$ for 
% $\crec{\mtm}{\cC}{\bind x{\mtm_{\cC}}}$:
% \begin{align*}
% (\lambda x.\mtm_0)\mtm_1 &= \subst{\mtm_0}{\mtm_1}{x} \\
% \mproj_i\mpair{\mtm_0}{\mtm_1} &= \mtm_i & & (i=0,1)\\
% \creckw(\cC\mtm_0) &= \subst{\mtm_{\cC}}{\cmap{}{\bind y {\mpair{y}{\creckw(y)}}}{\mtm_0}}{x}.
% \end{align*}
The preorder judgement is interpreted as equality.  In this
interpretation ${\trans e}_c$ is a recurrence that gives the exact cost
of evaluating~$e$, but reasoning about such a recurrence involves
reasoning about all of the details of the program.

%% One might argue that this exactitude
%% is a detraction.  The models previously described blur it out in favor
%% of less-precise bounds;
%% %, for which reasoning is easier.
%% the advantage to such bounds is that it can be
%% easier to reason about them.

\subsection{Infinite Costs}
\label{sec:inf-costs}

Next, we consider a size-based model in which we drop the ``increasing'' 
requirement on the $\semsize$ functions
from Section~\ref{sec:size_semantics}.  
Rather than requiring a well-founded partial order for each datatype,
we require an arbitrary partial order $(S^\tau,\leq_\tau)$ which we
also interpret as a flat CPO (we do not require the interpretation of
non-datatypes to be CPOs).
The interpretation of $\creckw$ expressions is then
in terms of a general fixpoint operator.  Define
$\infty = \bigvee S^\DD$ and identify~$\infty$ with the 
bottom element of the CPO ordering.
In this setting it may be that the interpretation of a $\creckw$ expression
does not terminate and hence, by our identification, evaluates to~$\infty$.
This turns out to be exactly the right behavior, as we can see in the
following example.

Take the standard
inductive definition of $\snat$ and interpret $\cnat$ as some
one-element set~$\set{1}$
in the complexity language, so $\semsize_{\cnat}$ is a constant
function---that is, declare that all $\snat$ values have the
same size.  Now compute the interpretation of the identity function:
\[
\begin{split}
&\den{\trans{\sreckw(y, \sZero\mapsto \sZero, \sSucc\mapsto x.\sSucc\,x)}}{} \\
&\quad= \creckw(1, \cZero\mapsto (0, 1) \mid \cSucc\mapsto \bind {\ctuple{x, r}} {(1+r_c, 1)}) \\
% &\quad =e(1) \\
&\quad= \bigvee_{\semsize z\leq 1}\semcase( z, {\cZero\mapsto {(0, 1)} \mid \cSucc\mapsto \bind {\ctuple{x, r}} {1+e_c(x)}})
\end{split}
\] where
\[
e(x) = \crecdeclm x {\cZero\mapsto {(0, 1)}\mid \cSucc\mapsto\bind {\ctuple{x, r}} {(1+r_c, 1)}}
\]

Since $\semsize(\cSucc(1)) = 1 \leq 1$, one of the $\semcase$ expressions
in the maximum is~$e_c(1)$.
In other words, we have a non-terminating recursion in computing the
complexity.
We conclude $\den{\cst{\trans{\sreckw(\dots)}}}{} = \infty$; in other words,
we can draw no useful conclusion about the cost of this expression.
This a feature of our approach rather than a bug.
What we have done in this example
is to declare that we cannot distinguish values of type $\snat$ by
size (they all have the same size), and then we attempt to compute the cost
of a recursive function on |nat|s in terms of the size of the recursion 
argument.  The bounding theorem still applies in this setting, and hence the interpretation gives us a 
bound on the cost of the computation.  In this case, the bound is just
not a useful one; it does not even tell us that the computation terminates.

\section{Related Work}
\label{sec:related_work}

There is a reasonably extensive literature over the last several
decades on (semi-)automatically constructing resource bounds
from source code.
The first work concerns itself with first-order
programs.
\citet{wegbreit:cacm75} describes a system for analyzing
simple Lisp programs that produces closed forms that bound
running time.  An interesting aspect of this system is that it
is possible to describe probability distributions on the input
domain
%%  (e.g., the probability that the head of an input list
%% will be some specified value), 
and the generated bounds incorporate
this information.
\citet{rosendahl:auto_complexity_analysis} proposes a system
based on step-counting functions and abstract interpretation
for a first-order subset of Lisp.  More recently the
COSTA project (see, e.g., 
\citet{albert-et-al:tcs12:cost-analysis})
has focused on automatically computing cost relations for
imperative languages (actually, bytecode) and solving them
(more on that in the next section).
\citet{debray-lin:toplas93:cost-analysis-logic-programs} develop
a system for analyzing logic programs and
\citet{navas-et-al:iclp07:user-definable-resource-bounds} extend it
to handle user-defined resources.

The Resource Aware ML project (RAML) takes a different approach to the
one we have described here, one based on type assignment.
\citet{jost-et-al:popl10} describe a formalism that automatically infers
linear resource bounds for higher-order programs, provided that the
input program does in fact have a linear resource cost.
\citet{hoffmann-hofmann:esop10} and
\citet{hoffmann-et-al:toplas12:multivariate-amortized} extend this work
to handle polynomial bounds, though for first-order programs only, and
\citet{hoffmann-shao:esop15:parallel} extend it to parallel programs.
RAML uses a source language that is similar to ours, but in which the
types are annotated with variables corresponding to resource usage.
Type inference in the annotated system comes down to solving a set of
constraints among these variables.  A very nice feature of this work is
that it handles cases in which amortized analysis is typically employed
to establish tight bounds, while our approach can only conclude (worst-case) 
bounds.  
%%We would like to adapt their techniques to our setting.
%% For example, RAML correctly infers that repeated
%% binary increment runs in linear time.  
%% Our approach can only
%% conclude (worst-case) linear time for binary increment, and hence
%% the repetition of binary increment leads to a quadratic over-estimate
%% of the cost.  We would like to adapt their techniques to our setting.

\citet{danielsson:popl08} uses an annotated monad (similar to $\C \times
-$, but dependent on the cost) to track running time in a dependently
typed language, where size reasoning can be done via types. He
emphasizes reasoning about amortized cost of lazy programs.  However, he
relies on explicit annotation of the program, which our complexity
translation inserts automatically, and his correctness theorem is only
for closed programs, whereas we use a logical relation 
to validate extracted recurrences.

%%% Is LeMatayer first-order or higher-order?
We now turn to work that is closest in spirit to ours, focusing on
those aspects related to analysis of higher-order languages.
\citeauthor{lematayer:toplas88}'s \citeyearpar{lematayer:toplas88}
ACE system is a two-stage system
that first converts FP
programs \citep{backus:fp} to recursive FP programs describing 
the number of recursive calls of the source program, then
attempts to transform the result using various program-transformation
techniques to obtain a closed form.
\citet{shultis:complexity} defines a denotational semantics for
a simple higher-order language that models both the value and the
cost of an expression.  As a part of the cost model, he develops a
system of ``tolls,'' which play a role similar to the potentials
we define in our work.  The tolls and the semantics are not used
directly in calculations, but rather as components in a logic for
reasoning about them.
\citet{sands:thesis} puts forward a translation scheme in which
programs in a source language are translated into programs in
the same language that
incorporate cost information; several source languages are discussed,
including a higher-order call-by-value language.  Each identifier~$f$
in the source language is associated to a \emph{cost closure} that incorporates
information about the value $f$ takes
on its arguments; the cost of applying~$f$ to arguments; and arity.
Cost closures are intended to address the same issue our higher-type
potentials do:  recording information about the future cost
of a partially-applied function.
\citet{van-stone:thesis} annotates the operational semantics for
a higher-order language with cost information.  She then
defines a category-theoretic denotational
semantics that uses ``cost structures'' 
%% (which are related to monads) 
to capture cost information and shows
that the latter is sound with respect to the former.
\citet{benzinger:tcs04} annotates NuPRL's call-by-name operational semantics
with complexity estimates.  The language for the annotations is left
somewhat open so as to allow greater flexibility.  The analysis of
the costs is then completed using a combination of NuPRL's proof
generation and Mathematica.  
%% seems dangerous and doesn't really help our case
%% Benzinger's is the only system to explicitly involve automated theorem
%% proving, though Sand's could also do so.  
In all of these approaches the cost domain incorporates information
about values in the source language so as to provide exact costs.  
Our approach provides a uniform framework that can be more or less
precise about the source language values that are represented.
While we can implement a version that
handles exact costs, we can also implement a version in which we focus
just on upper bounds, which we might hope leads to simpler recurrences.

\section{Conclusions and Further Work}
\label{sec:concl_further_work}

We have described a denotational complexity analysis for a higher-order
language with a general form of inductive datatypes that yields an upper
bound on the cost of any well-typed program in terms of the size of the input.  The two steps are to
translate each source-language program~$e$ into a program~$\trans e$ in
a complexity language, which makes costs explicit, and then to abstract
values to sizes.  We prove a bounding theorem for the translation, a
consequence of which is that the cost component of~$\trans e$ is an
upper bound on the evaluation cost of~$e$.  The proof the
bounding theorem is purely syntactic, and therefore applies in all
models of the complexity language.  By varying the semantics of
the complexity language (and in particular, the notion of size), we can
perform analyses at different levels of granularity.  We give several
different choices for the notion of size, but ultimately this is too
important a decision to take out of the hands of the user through
automation.

The complexity translation of Section~\ref{sec:complexity_lang} can
easily be adapted to other cost models.  For example, we could charge
different amounts for different steps.  Or, we could analyze
the work and span of parallel programs by taking $\C$ to be
series-parallel cost graphs, something we plan to investigate
in future work.% we plan to investigate
%size abstractions in a parallel setting.  

%%% TO DO:
%%% * General recursion.
%%% * Formalization/certification.
%%% * Solution to recurrences.

Another direction for future work is to handle
different evaluation strategies.
Compositionality is
a thorny issue when considering call-by-need evaluation and lazy
datatypes, and as noted by 
\citet{okasaki:purely-functional-data-structures},
it may be that amortized cost is at least
as interesting as worst-case cost.
\citet{sands:thesis}, \citet{van-stone:thesis},
and \citet{danielsson:popl08} address laziness in their work, and
as we already noted, RAML already performs
amortized analyses.
%% The call-by-push-value paradigm \citep{levy:tlca99} gives
%% an alternative perspective on our complexity analysis. Call-by-push-value
%% disinguishes values from computations in a monadic-like approach
%% under the maxim ``a value is, a computation does.'' With this in
%% mind, 
%% we might present our work following the maxim
%% ``potential measures what is, cost measures what happens.''
% we might adopt the following statement with respect to complexities:
% ``potential measures what is, cost measures what happens.'' An alternative
% presentation of our work might utilize a call-by-push-value target language
% to emphasize the distinction between computation expressions
% and value expressions and what those mean for the complexity analysis. 

We plan to extend the source language to handle general
recursion.  Part of the difficulty here is that the bounding relation
presupposes termination of the source
program (so that the derivation of $\evalin e v n$, and
hence cost, is well-defined).
One approach would be to require the user to supply
a proof of termination of the program to be analyzed.  Or,
% it should be possible to define
one could define
the operational semantics of the source language co-inductively
(as done by, e.g., \cite{leroy-grall:ic09-coind-big-step}), thereby
allowing explicitly for non-terminating computations.  
%Alternatively, we might try to follow the lead of the
Another approach is to adapt the
partial big-step operational
semantics described by
\citet{hoffmann-et-al:toplas12:multivariate-amortized}.
%An interesting point here is that 
Since our source language supports
inductive datatype definitions of the form
$\sdatatypekw~{\srckeyw{strm}}=\sconstrdecl{\scons}{\sarr\sunit{\snat\cross\srckeyw{strm}}}$,
adding general recursion will force us to understand how our
complexity semantics plays out in the presence of what are essentially
coinductively defined values.
One could
also hope to prove termination in the source language
by first extracting complexity bounds and
then proving that these bounds in fact define total functions.
Another interesting idea along these lines would be to define
a complexity semantics in which the cost domain is two-valued, with
one value representing termination and the other non-termination
(or maybe more accurately, known termination and not-known-termination);
such an approach might be akin to an abstract interpretation based
approach for termination analysis.

% Although we have carried out most of the formalization of an earlier
% version of this work~\citep{danner-et-al:plpv13}, a fresh formalization
% is called for with the more general appraoch we have taken here.  When
% complete, such a formalization may provide the foundation for 
% certified cost certificates.
% The
% foundations for formalizing such denotational semantics have already been
% carried out by
% \citet{paulin-mohring:den-semantics-coq} and
% \citet{benton-et-al:domain-theory-coq}.  

%%% Ref. Benzinger here for solving recurrences and in general for
%%% translation into a proof assistant.
The programs $\trans \stm$ are complex higher-order recurrences that
call out for solution techniques.
\citet{benzinger:tcs04} addresses this idea, as do
\citet{albert-et-al:jar11,
albert-et-al:tocl13:inference} of the COSTA project.
Another relevant
aspect of the COSTA work is that their cost relations use non-determinism;
it would be very interesting to see if we could employ a similar approach
instead of the maximization operators that we used in our examples.
Ultimately we should have
a library of tactics for transforming the recurrences produced by
the translation function to closed (possibly asymptotic) forms when possible.
% \cite{benzinger:tcs04} addresses this in his system.

\bibliographystyle{abbrvnat}
\bibliography{icfp2015}

\end{document}

%% file: defs.tex
\let\union=\cup

\let\cross=\times

\let\isom=\cong
\let\cong=\equiv

\let\equiv=\sim
\def\isomto{\mathrel{\hbox{$\to$\kern-.85em\raise1ex\hbox{{$\scriptstyle \isom$}}}}\;}
\def\ndiv{{\not \kern -.05em |\ }}

\let\comp=\circ

\def\dom{\mathrm{Dom}\;}

\expandafter\ifx\csname Bbb\endcsname\relax

\else

\fi

\newcommand{\seq}[2][\relax]{
  \ifx#1\relax
    \langle#2\rangle
  \else
    \ifx#1\left
      \left\langle#2\right\rangle
    \else
      \csname #1l\endcsname\langle#2\csname #1r\endcsname\rangle
    \fi
  \fi
}

\newcommand{\set}[2][\relax]{
  \ifx#1\relax
    \{#2\}
  \else
    \ifx#1\left
      \left\{#2\right\}
    \else
      \csname #1l\endcsname\{#2\csname #1r\endcsname\}
    \fi
  \fi
}

%% file: lambda-defs.tex
\def\arrow{\mathbin{\rightarrow}}

\def\llambda{{\lambda\hskip-.45em\lambda}}	% Semantic lambda

\def\fv{\mathop{\mathrm{fv}}\nolimits}

\def\oftype{\mathbin{:}}

\def\subst#1#2#3{#1[#2/#3]}

\let\den\tmden

%% file: local_defs.tex
\def\trans#1{\lVert{#1}\rVert}

\def\letexp#1#2#3{\lstinline!let $\;#1\;$ = $\;#2\;$ in $\;#3$!}
\def\ifexp#1#2#3{\lstinline!if $\;#1\;$ then $\;#2\;$ else $\;#3$!}
\def\caseexp#1#2#3#4#5{\lstinline!case $\;#1\;$ of ($#2, [#3, #4]#5$)!}
\def\foldexp#1#2#3#4#5#6{\lstinline!fold $\;#1\;$ of ($#2, [#3, #4, #5]#6$)!}
\def\foldexpv#1#2#3#4#5#6{\lstinline!fold$^v$ $\;#1\;$ of ($#2, [#3, #4, #5]#6$)!}

\let\evalto\downarrow
\def\oftype{\mathbin{:}}
\let\proves\vdash

\EnableBpAbbreviations
\def\doRule#1{\ifx#1\relax\else\RightLabel{#1}\fi}
\newcommand{\ndAXC}[2][\relax]{\AXC{$\mathstrut$}\doRule{#1}\UIC{#2}}
\newcommand{\ndUIC}[2][\relax]{\doRule{#1}\UIC{#2}}
\newcommand{\ndBIC}[2][\relax]{\doRule{#1}\BIC{#2}}
\newcommand{\ndTIC}[2][\relax]{\doRule{#1}\TIC{#2}}

\def\evalin#1#2#3{{#1}\evalto^{#3}{#2}}

\def\tmoftype#1#2{{#1}\oftype{#2}}
\newcommand{\typepf}[3][{}]{\proves_{#1}\tmoftype{#2}{#3}}
\newcommand{\typejudge}[4][{}]{{#2}\typepf[{#1}]{#3}{#4}}
\newcommand{\typejudgeM}[3][{}]{\typejudge[#1]\mctx{#2}{#3}}
\newcommand{\typejudgeS}[3][{}]{\typejudge[#1]\sctx{#2}{#3}}

\newcommand{\tmle}[3][{}]{{#2}\leq_{#1}{#3}}
\newcommand{\lepf}[4][{}]{\proves_{#1}{\tmle[#2]{#3}{#4}}}
\newcommand{\lejudge}[5][{}]{{#2}\lepf[{#1}]{#3}{#4}{#5}}
\newcommand{\lejudgeMP}[4][{}]{\lejudge[]\mctx{#2}{#3}{#4}}
\newcommand{\lejudgekw}[0]{\leq}

\newcommand{\typing}[3][\relax]{\ifx#1\relax\tmoftype{#2}{#3}\else\typejudge{#1}{#2}{#3}\fi}
\newcommand{\typingG}[3][{}]{\typejudge[#1]\Gamma{#2}{#3}}

\newcommand{\isder}[2]{{#1}\mathrel{::}{#2}}

\def\OK{\mathord{\srckeyw{ok}}}
\def\ok#1#2{{#1}\proves{#2}\,\OK}

\def\TYPE{\mathord{\srckeyw{type}}}
\newcommand\istype[2]{{#1}\proves{#2}\,\TYPE}

\newcommand\cl[2]{{#1}[{#2}]}
\WithSuffix\def\cl*#1#2{\cl{(#1)}{#2}}

\def\bindenv#1#2{{#1}\mapsto{#2}}

\def\unitenv#1#2{\{\bindenv{#1}{#2}\}}
\def\extend#1#2#3{#1,#2/#3}
\def\extendenv#1#2#3{#1{\{#2\mapsto#3\}}}

\def\typot#1{\langle\!\langle#1\rangle\!\rangle}
\let\cpxy\trans

\def\N{\mathbf{N}}
\def\Z{\mathbf{Z}}

\let\bmax\vee
\let\bigmax\bigvee
\def\pcaseexp#1#2#3#4#5{\lstinline!pcase $\;#1\;$ of ($#2, [#3, #4]#5$)!}
\def\pfoldexp#1#2#3#4#5#6{\lstinline!pfold $\;#1\;$ of ($#2, [#3, #4, #5]#6$)!}

\let\bounded\sqsubseteq
\def\vbounded{\bounded^{\mathrm{val}}}

\def\srckeyw#1{\lstinline!#1!}

\def\sprod#1#2{{#1}\times{#2}}

\def\DC{C}
\def\sconstrdecl#1#2{{#1}~\lstinline!of!~{#2}}
\def\sconstrsep{\mid}
\def\sdatatypekw{\srckeyw{datatype}}
\def\sdatadecl#1#2#3{\srckeyw{datatype}\:{#1} = \overline{\sconstrdecl{#2}{#3}}}
\WithSuffix\def\sdatadecl*#1#2{\srckeyw{datatype}\:{#1} = {#2}}

\def\srcconstr#1{\mathop{\srckeyw{#1}}\nolimits}
\def\scase{\srcconstr{case}}
\def\sC{C}
\def\sdelaykw{\srcconstr{delay}}
\newcommand\sdelay[1]{\sdelaykw(#1)}
\def\sforcekw{\srcconstr{force}}
\newcommand\sforce[1]{\sforcekw(#1)}
\def\sproj{\pi}
\def\ssplitkw{\srcconstr{split}}
\newcommand{\ssplit}[4]{\ensuremath{\ssplitkw(#1,\bind{{#2}.{#3}}{#4})}}
\def\sletkw{\srcconstr{let}}
\newcommand{\slet}[3]{\sletkw({#1}, \bind{#2}{#3})}
\def\smapkw{\srcconstr{map}}
\newcommand{\smap}[3]{\smapkw^{{#1}}({#2}, {#3})}
\def\sreckw{\srcconstr{rec}}
\newcommand{\srec}[4][{}]{\sreckw^{#1}(#2, \overline{{#3}\mapsto{#4}})}
\WithSuffix\def\srec*#1#2#3#4#5{\srec{#1}{#2}{\bind{#3}{#5}}}
\newcommand{\srecm}[3][{}]{\sreckw^{#1}(#2, {#3})}

\def\stuple#1{\ensuremath{\langle #1\rangle}}
\def\spair#1#2{\stuple{#1,#2}}
\def\ssusp#1{\srckeyw{susp}\:#1}
\def\striv{\langle\,\rangle}

\def\cpykeyw#1{\mathop{\mathsf{#1}}\nolimits}

\def\cunit{\cpykeyw{unit}}
\let\ctriv\striv
\def\cnat{\cpykeyw{nat}}
\def\clist{\cpykeyw{list}}
\def\cbool{\cpykeyw{bool}}
\def\cint{\cpykeyw{int}}
\def\ctree{\cpykeyw{tree}}

\def\cemp{\cpykeyw{Emp}}
\def\cnode{\cpykeyw{Node}}
\def\cnil{\cpykeyw{Nil}}
\def\ccons{\cpykeyw{Cons}}
\def\cZero{\cpykeyw{Zero}}
\def\cSucc{\cpykeyw{Succ}}

\def\cC{C}

\def\cproj{\pi}

\def\cmapkw{\cpykeyw{map}}
\newcommand{\cmap}[3]{\cmapkw^{{#1}}({#2}, {#3})}
\def\creckw{\cpykeyw{rec}}
\newcommand{\crec}[4][{}]{\creckw^{#1}(#2, \overline{{#3}\mapsto{#4}})}
\WithSuffix\def\crec*#1#2#3#4#5{\crec{#1}{#2}{\bind{#3}{#5}}}
\newcommand{\crecm}[3][{}]{\creckw^{#1}(#2, {#3})}

\newcommand\clistrec[4]{\creckw(#1,\cnil \mapsto #2, \ccons \mapsto \bind{#3}{#4})}

\def\ctuple#1{\ensuremath{\langle {#1} \rangle}}
\def\cpair#1#2{\ctuple{{#1},{#2}}}
\def\ctriv{\langle\,\rangle}

\def\cconstrdecl#1#2{{#1}\:\cpykeyw{of}\:{#2}}
\def\cdatatypekw{\cpykeyw{datatype}}
\def\cdatadecl#1#2#3{\cpykeyw{datatype}\:{#1} = \overline{\cconstrdecl{#2}{#3}}}
\WithSuffix\def\cdatadecl*#1#2{\cpykeyw{datatype}\:{#1} = {#2}}

\def\sdatatype{\srckeyw{datatype}}

\newcommand{\crecdecl}[4][{}]{\cpykeyw{rec}^{#1}(#2, \overline{{#3}\mapsto{#4}})}
\newcommand{\crecdeclm}[3][{}]{\cpykeyw{rec}^{#1}(#2, {#3})}

\def\D{\ensuremath\delta}
\def\DD{\ensuremath\Delta}
\def\sbool{\srckeyw{bool}}
\def\sint{\srckeyw{int}}
\def\slist{\srckeyw{list}}
\def\snat{\srckeyw{nat}}
\def\stree{\srckeyw{tree}}
\def\sunit{\srckeyw{unit}}

\def\snil{\srckeyw{Nil}}
\def\scons{\srckeyw{Cons}}
\def\sZero{\srckeyw{Zero}}
\def\sSucc{\srckeyw{Succ}}
\def\strue{\srckeyw{True}}
\def\sfalse{\srckeyw{False}}
\def\semp{\srckeyw{Emp}}
\def\snode{\srckeyw{Node}}

\def\smem{\srckeyw{mem}}
\def\slistmap{\srckeyw{listmap}}
\def\streemap{\srckeyw{treemap}}

\def\ctrue{\cpykeyw{True}}
\def\cfalse{\cpykeyw{False}}

\def\plusc{+_c}
\def\costpluscpy#1#2{{#1}\plusc{#2}}
\def\costplusone#1{\costpluscpy{1}{#1}}

\def\subst#1#2#3{\ensuremath{{#1}[{#2}/{#3}]}}
\WithSuffix\def\subst*#1#2{\ensuremath{{#1}[#2]}}

\newcommand{\dcd}[1]{\ensuremath{\mathit{#1}}}
\newcommand{\dsd}[1]{\ensuremath{\mathsf{#1}}}
\newcommand{\tm}[1]{\dcd{#1}}

\let\tptmns\tmoftype
\newcommand{\oftps}[4]{\ensuremath{\typejudge[{#2}]{#1}{#3}{#4}}}

\newcommand\arr[2]{\ensuremath{{#1}\to{#2}}}
\newcommand\prd[2]{{#1}\cross{#2}}

\newcommand\softps\oftps
\newcommand\sprd\prd
\newcommand\sarr\arr
\newcommand\slam\ulam
\newcommand\sapp\app
\newcommand\sdcon[2]{\ensuremath{{#1} \: #2}}

\newcommand\seval[3]{\ensuremath{\evalin{#1}{#3}{#2}}}
\newcommand\sctx\gamma  %% source context
\newcommand\ssig\psi    %% source signature
\newcommand{\emptysig}{\langle\,\rangle}
\newcommand\sspf\phi %% source dcon arg spec
\newcommand\stm{e} %% source term
\newcommand\sv{v} %% source value
\newcommand\stp{\tau} %% source type
\renewcommand\sc{n} %% source cost annotation on evaluation  
\newcommand\ssub\theta %% source substitution

\newcommand\sarg[2]{\ensuremath{\subst*{#1}{#2}}}

\newcommand\swfsig[1]{\ensuremath{#1 \: \srckeyw{sig}}}

\newcommand\sdconbind[3]{\ensuremath{{#1} : (#2 \to {#3})}}
\newcommand\datadecl\sdatadecl

\newcommand{\mle}[1][{}]{\le_{#1}}
\newcommand{\mge}[1][{}]{\ge_{#1}}
\newcommand\homj[3]{\ensuremath{#2 \le_{#1} #3}}

\newcommand\dcon[2]{\ensuremath{{#1} \: #2}}
\renewcommand{\prec}[5]{\ensuremath{\dsd{rec}(#1, #2 \mapsto #3.#5)}}

\newcommand\C[0]{\ensuremath{\mathbf{C}}}

\newcommand\mmap[3]{\ensuremath{\dsd{map}^{#1}(#2,#3)}}

\newcommand\mctx\Gamma %% monotonic ctx
\newcommand\msig\Psi %% monotonic sig
\newcommand\mspf\Phi %% monotonic dcon arg spec
\newcommand\mtm{E} %% monotonic term
\newcommand\mtp{T} %% monotonic type
\newcommand\mdtp{\Delta} %% monotonic inductive datatype.
\newcommand\msub\Theta %% monotonic substitution

\newcommand\mprod[2]{{#1}\cross{#2}}
\let\mpair\spair
\newcommand\mproj{\mathop\pi\nolimits}

\newcommand\marg[2]{\ensuremath{#1[#2]}}

\newcommand{\semkw}[1]{\mathit{#1}}
\newcommand{\seminj}{\semkw{inj}}
\newcommand{\semcase}{\semkw{case}}
\newcommand{\semmapkw}{\semkw{map}}
\newcommand{\semsize}{\mathop{\semkw{size}}\nolimits}
\newcommand{\semsz}{\semkw{sz}}
\newcommand{\semlam}[2]{\llambda{#1}.{#2}}
\newcommand{\sempair}[2]{({#1}, {#2})}
\newcommand{\semrec}[1][{}]{\semkw{rec}^{{#1}}}

\newcommand\pottr[1]{\ensuremath{\langle\!\langle #1 \rangle\!\rangle}}
\newcommand\comptr[1]{\ensuremath{\trans{#1}}}

\newcommand\cst[1]{\ensuremath{{#1_{c}}}}
\WithSuffix\def\cst*#1{\cst{(#1)}}
\newcommand\pot[1]{\ensuremath{{#1_{p}}}}
\WithSuffix\def\pot*#1{\pot{(#1)}}

\newcommand\subbound[3]{\ensuremath{ #1 \sqsubseteq^{\textit{sub}}_{#3} #2}}

\newcommand{\expbound}[3]{\ensuremath{{#1}\bounded_{#3}{#2}}}
\newcommand{\valbound}[3]{\ensuremath{{#1}\vbounded_{#3}{#2}}}
\newcommand{\spfexpbound}[4]{\ensuremath{{#1}\bounded_{#4,#3}{#2}}}
\newcommand{\spfvalbound}[4]{\ensuremath{{#1}\vbounded_{#4,#3}{#2}}}

\newcommand\reccall{t}
\newcommand\costzero{\ensuremath{0}}
\newcommand\costone{\ensuremath{1}}
\newcommand\costplus{\ensuremath{+}}
\newcommand\costplusfst{\ensuremath{+}_\dsd{c}}

\let\pair\spair
\def\fst{\ensuremath{\mathop{\pi_0}}}
\def\snd{\ensuremath{\mathop{\pi_1}}}

\newcommand{\ulam}[2]{\tm{\lambda{#1}.{#2}}}

\newcommand{\fn}[3]{\tm{\fn\,\tptmns{#1}{#2}.\,#3}}

\newcommand{\app}[2]{\dcd{#1 \: #2}}

\newcommand{\tsubst}[2]{\subst*{#1}{#2}}

\newcommand{\congctx}{\mathcal{C}}
\newcommand{\hole}{[\,]}

%% file: icfp2015-full.bbl
\begin{thebibliography}{26}
\providecommand{\natexlab}[1]{#1}
\providecommand{\url}[1]{\texttt{#1}}
\expandafter\ifx\csname urlstyle\endcsname\relax
  \providecommand{\doi}[1]{doi: #1}\else
  \providecommand{\doi}{doi: \begingroup \urlstyle{rm}\Url}\fi

\bibitem[Albert et~al.(2011)Albert, Arenas, Genaim, and
  Puebla]{albert-et-al:jar11}
E.~Albert, P.~Arenas, S.~Genaim, and G.~Puebla.
\newblock Closed-form upper bounds in static cost analysis.
\newblock \emph{Journal of Automated Reasoning}, 46:\penalty0 161--203, 2011.
\newblock \doi{10.1007/s10817-010-9174-1}.

\bibitem[Albert et~al.(2012)Albert, Arenas, Genaim, Puebla, and
  Zanardini]{albert-et-al:tcs12:cost-analysis}
E.~Albert, P.~Arenas, S.~Genaim, G.~Puebla, and D.~Zanardini.
\newblock Cost analysis of object-oriented bytecode programs.
\newblock \emph{Theoretical Computer Science}, 413\penalty0 (1):\penalty0
  142--159, 2012.
\newblock \doi{10.1016/j.tcs.2011.07.009}.

\bibitem[Albert et~al.(2013)Albert, Genaim, and
  Masud]{albert-et-al:tocl13:inference}
E.~Albert, S.~Genaim, and A.~N. Masud.
\newblock On the inference of resource usage upper and lower bounds.
\newblock \emph{ACM Transactions on Computational Logic}, 14\penalty0
  (3):\penalty0 22:1--22:35, 2013.
\newblock \doi{10.1145/2499937.2499943}.

\bibitem[Backus(1978)]{backus:fp}
J.~Backus.
\newblock Can programming be liberated from the von {N}eumann style? {A}
  functional style and its algebra of programs.
\newblock \emph{Communications of the Association for Computing Machinery},
  21\penalty0 (8):\penalty0 613--641, 1978.
\newblock \doi{10.1145/359576.359579}.

\bibitem[Benzinger(2004)]{benzinger:tcs04}
R.~Benzinger.
\newblock Automated higher-order complexity analysis.
\newblock \emph{Theoretical Computer Science}, 318\penalty0 (1-2):\penalty0 79
  -- 103, 2004.
\newblock \doi{10.1016/j.tcs.2003.10.022}.

\bibitem[Danielsson(2003)]{danielsson:popl08}
N.~A. Danielsson.
\newblock Lightweight semiformal time complexity analysis for purely functional
  data structures.
\newblock In A.~Aiken and G.~Morrisett, editors, \emph{Proceedings of the 30th
  ACM SIGPLAN-SIGACT Symposium on Principles of Programming Languages}, pages
  133--144. ACM Press, 2003.
\newblock \doi{10.1145/1328438.1328457}.

\bibitem[Danner and Royer(2009)]{danner-royer:two-algs}
N.~Danner and J.~S. Royer.
\newblock Two algorithms in search of a type system.
\newblock \emph{Theory of Computing Systems}, 45\penalty0 (4):\penalty0
  787--821, 2009.
\newblock \doi{10.1007/s00224-009-9181-y}.

\bibitem[Danner et~al.(2013)Danner, Paykin, and Royer]{danner-et-al:plpv13}
N.~Danner, J.~Paykin, and J.~S. Royer.
\newblock A static cost analysis for a higher-order language.
\newblock In M.~Might and D.~V. Horn, editors, \emph{Proceedings of the 7th
  workshop on {P}rogramming languages meets program verification}, pages
  25--34. ACM Press, 2013.
\newblock \doi{10.1145/2428116.2428123}.

\bibitem[Debray and
  Lin(1993)]{debray-lin:toplas93:cost-analysis-logic-programs}
S.~K. Debray and N.-W. Lin.
\newblock Cost analysis of logic programs.
\newblock \emph{ACM Transactions on Programming Languages and Systems},
  15\penalty0 (5):\penalty0 826--875, 1993.
\newblock \doi{10.1145/161468.161472}.

\bibitem[Harper(2013)]{harper:pfpl}
R.~Harper.
\newblock \emph{Practical Foundations for Programming Languages}.
\newblock Cambridge University Press, 2013.

\bibitem[Hoffmann and Hofmann(2010)]{hoffmann-hofmann:esop10}
J.~Hoffmann and M.~Hofmann.
\newblock Amortized resource analysis with polynomial potential: A static
  inference of polynomial bounds for functional programs.
\newblock In A.~D. Gordon, editor, \emph{Programming Languages and Systems:
  19th European Symposium on Programming, ESOP 2010}, volume 6012 of
  \emph{Lecture Notes in Computer Science}, pages 287--306. Springer-Verlag,
  2010.
\newblock \doi{10.1007/978-3-642-11957-6_16}.

\bibitem[Hoffmann and Shao(2015)]{hoffmann-shao:esop15:parallel}
J.~Hoffmann and Z.~Shao.
\newblock Automatic static cost analysis for parallel programs.
\newblock In J.~Vitek, editor, \emph{Programming Languages and Systems: 24th
  European Symposium on Programming, ESOP 2015}, volume 9032 of \emph{Lecture
  Notes in Computer Science}, pages 132--157. Springer-Verlag, 2015.
\newblock \doi{10.1007/978-3-662-46669-8_6}.

\bibitem[Hoffmann et~al.(2012)Hoffmann, Aehlig, and
  Hofmann]{hoffmann-et-al:toplas12:multivariate-amortized}
J.~Hoffmann, K.~Aehlig, and M.~Hofmann.
\newblock Multivariate amortized resource analysis.
\newblock \emph{ACM Transactions on Programming Languages and Systems},
  34\penalty0 (3):\penalty0 14:1--14:62, 2012.
\newblock \doi{10.1145/2362389.2362393}.

\bibitem[Jost et~al.(2010)Jost, Hammond, Loidl, and Hofmann]{jost-et-al:popl10}
S.~Jost, K.~Hammond, H.-W. Loidl, and M.~Hofmann.
\newblock Static determination of quantitative resource usage for higher-order
  programs.
\newblock In M.~Hermenegildo, editor, \emph{Proceedings of the 37th Annual ACM
  SIGPLAN-SIGACT Symposium on Principles of Programming Languages}, pages
  223--236. ACM Press, 2010.
\newblock \doi{10.1145/1706299.1706327}.

\bibitem[Le~M{\'e}tayer(1988)]{lematayer:toplas88}
D.~Le~M{\'e}tayer.
\newblock {ACE}: an automatic complexity evaluator.
\newblock \emph{ACM Transactions on Programming Languages and Systems},
  10\penalty0 (2):\penalty0 248--266, 1988.
\newblock \doi{10.1145/42190.42347}.

\bibitem[Leroy and Grall(2009)]{leroy-grall:ic09-coind-big-step}
X.~Leroy and H.~Grall.
\newblock Coinductive big-step operational semantics.
\newblock \emph{Information and Computation}, 207\penalty0 (2):\penalty0
  284--304, 2009.
\newblock \doi{10.1016/j.ic.2007.12.004}.

\bibitem[Moggi(1991)]{moggi:ic91:monads}
E.~Moggi.
\newblock Notions of computation and monads.
\newblock \emph{Information And Computation}, 93\penalty0 (1):\penalty0 55--92,
  1991.
\newblock \doi{10.1016/0890-5401(91)90052-4}.

\bibitem[Navas et~al.(2007)Navas, Mera, L\'opez-Garcia, and
  Hermenegildo]{navas-et-al:iclp07:user-definable-resource-bounds}
J.~Navas, E.~Mera, P.~L\'opez-Garcia, and M.~V. Hermenegildo.
\newblock User-definable resource bounds analysis for logic programs.
\newblock In V.~Dahl and I.~Niemel\"a, editors, \emph{Proceedings of Logic
  Programming: 23rd International Conference, ICLP 2007}, volume 4670 of
  \emph{Lecture Notes in Computer Science}, pages 348--363, 2007.
\newblock \doi{10.1007/978-3-540-74610-2_24}.

\bibitem[Okasaki(1998)]{okasaki:purely-functional-data-structures}
C.~Okasaki.
\newblock \emph{Purely Functional Data Structures}.
\newblock Cambridge University Press, 1998.

\bibitem[Rosendahl(1989)]{rosendahl:auto_complexity_analysis}
M.~Rosendahl.
\newblock Automatic complexity analysis.
\newblock In J.~E. Stoy, editor, \emph{Proceedings of the Fourth International
  Conference on Functional Programming Languages and Computer Architecture},
  pages 144--156. ACM Press, 1989.
\newblock \doi{10.1145/99370.99381}.

\bibitem[Sands(1990)]{sands:thesis}
D.~Sands.
\newblock \emph{Calculi for Time Analysis of Functional Programs}.
\newblock PhD thesis, University of London, 1990.

\bibitem[Shultis(1985)]{shultis:complexity}
J.~Shultis.
\newblock On the complexity of higher-order programs.
\newblock Technical Report CU-CS-288-85, University of Colorado at Boulder,
  1985.

\bibitem[Van~Stone(2003)]{van-stone:thesis}
K.~Van~Stone.
\newblock \emph{A Denotational Approach to Measuring Complexity in Functional
  Programs}.
\newblock PhD thesis, School of Computer Science, Carnegie Mellon University,
  2003.

\bibitem[Wadler(1987)]{wadler:popl87:views}
P.~Wadler.
\newblock Views: A way for pattern matching to cohabit with data abstraction.
\newblock In \emph{Proceedings of the 14th ACM SIGACT-SIGPLAN Symposium on
  Principles of Programming Languages}, pages 307--313, 1987.
\newblock \doi{10.1145/41625.41653}.

\bibitem[Wadler(1992)]{wadler:popl92:essence}
P.~Wadler.
\newblock The essence of functional programming.
\newblock In R.~Sethi, editor, \emph{Proceedings of the 19th ACM SIGPLAN-SIGACT
  Symposium on Principles of Programming Languages}, pages 1--14. ACM Press,
  1992.
\newblock \doi{10.1145/143165.143169}.

\bibitem[Wegbreit(1975)]{wegbreit:cacm75}
B.~Wegbreit.
\newblock Mechanical program analysis.
\newblock \emph{Communications of the Association for Computing Machinery},
  18\penalty0 (9):\penalty0 528--539, 1975.
\newblock \doi{10.1145/361002.361016}.

\end{thebibliography}
